%% file: main.tex
\documentclass[11pt]{article}
\def \defvmargin{1.88cm}
\def \defhmargin{2cm}
\def \arxiv{1}
\def \todotoggle{0}
\input{macros}

\newcommand{\ml}{\mathrm{ml}}
\newcommand{\nml}{\mathrm{nml}}
\newcommand{\ANIS}{\text{\sc Gap-NIS}}

\newcommand{\fullver}[2]{\ifnum\arxiv=0 {#1} \else {#2} \fi}
\allowdisplaybreaks
\begin{document}

\input{abstract}

\thispagestyle{empty}
\newpage
\thispagestyle{empty}
\tableofcontents
\tableoftodos
\newpage
\setcounter{page}{1}
\input{sec_intro}

\input{sec_prelim}

\input{sec_dim-reduction}

\input{sec_smoothing}

\input{sec_non-multilinear}

\input{sec_noise-stability}

\input{sec_non-int-sim}

\input{ack}

%% For arXiv version: Comment the next two lines, uncomment the third line.
%\bibliographystyle{alpha}
%\bibliography{refs}
\input{main.bbl}

\appendix

\input{apx_dim-reduction}

\input{apx_smoothing}

\input{apx_non-multilinear}

\input{sec_regularity}

\input{sec_invariance}

\input{sec_decidability}

\end{document}

%% file: macros.tex
\providecommand{\todotoggle}{1} % "\def \todotoggle{0}" to hide all todos.

\usepackage{amsmath, amssymb, amsthm, amsfonts, bm}
\usepackage[tbtags]{mathtools}
\usepackage{graphicx}
\usepackage{verbatim}
\usepackage{enumerate}
\usepackage[usenames,dvipsnames]{xcolor}
\usepackage{units}
\usepackage{tikz}
\usepackage{array}
\usepackage{subfigure}
\usepackage{setspace}
\usepackage[utf8]{inputenc}
\usepackage[T1]{fontenc}
\usepackage[pdfstartview=FitH, pdfpagemode=None, colorlinks=true, citecolor=OliveGreen, linkcolor=BlueViolet, urlcolor=magenta, pagebackref=true, unicode=false]{hyperref}
\usepackage[ruled,linesnumbered]{algorithm2e}
\usepackage{multirow}
\usepackage[nameinlink]{cleveref}
\crefname{equation}{Equation}{Equations}
\usepackage{thmtools,thm-restate}

\providecommand{\arxiv}{0}
\ifnum\arxiv=1
\usepackage{mathpazo}
\providecommand{\defvmargin}{1in} % Do "\def \defmargin{2cm}" or such.
\providecommand{\defhmargin}{1in} % Do "\def \defmargin{2cm}" or such.
\usepackage[top=\defvmargin, bottom=\defvmargin, left=\defhmargin, right=\defhmargin]{geometry}
\fi
\usepackage{xargs} % Use more than one optional parameter in a new commands
\usepackage[colorinlistoftodos,prependcaption,textsize=footnotesize]{todonotes}

\newcommand{\mytodo}[1]{\ifnum\todotoggle=1{#1}\fi}

\newcommandx{\unsure}[2][1=]{\mytodo{\todo[linecolor=blue,backgroundcolor=blue!25,bordercolor=blue,#1]{#2}}}
\newcommandx{\change}[2][1=]{\mytodo{\todo[linecolor=cyan,backgroundcolor=cyan!25,bordercolor=cyan,#1]{#2}}}
\newcommandx{\info}[2][1=]{\mytodo{\todo[linecolor=green,backgroundcolor=green!25,bordercolor=green,#1]{#2}}}
\newcommandx{\improve}[2][1=]{\mytodo{\todo[linecolor=red,backgroundcolor=red!25,bordercolor=red,#1]{#2}}}
\newcommandx{\thiswillnotshow}[2][1=]{\todo[disable,#1]{#2}}
\newcommand{\tableoftodos}{\ifnum\todotoggle=1 \listoftodos[Comments/To Do's] \fi}

\usetikzlibrary{positioning,snakes,mindmap,calc,fit,arrows,shapes,shadings}

\tikzstyle{box} = [shape=rectangle,draw=black,thick]

\newcolumntype{L}[1]{>{\raggedright\let\newline\\\arraybackslash\hspace{0pt}}m{#1}}
\newcolumntype{C}[1]{>{\centering\let\newline\\\arraybackslash\hspace{0pt}}m{#1}}
\newcolumntype{R}[1]{>{\raggedleft\let\newline\\\arraybackslash\hspace{0pt}}m{#1}}

\newcommand{\myignore}[1]{}

\newcommand{\TealBlue}[1]{\textcolor{TealBlue}{#1}}
\newcommand{\Peach}[1]{\textcolor{Peach}{#1}}
\newcommand{\Cyan}[1]{\textcolor{cyan}{#1}}
\newcommand{\Red}[1]{\textcolor{red}{#1}}
\newcommand{\Navy}[1]{\textcolor{Blue}{#1}}
\newcommand{\Blue}[1]{\textcolor{Blue}{#1}}
\newcommand{\Green}[1]{\textcolor{OliveGreen}{#1}}
\newcommand{\Black}[1]{\textcolor{black}{#1}}
\newcommand{\White}[1]{\textcolor{white}{#1}}
\newcommand{\Code}[1]{\Cyan{\texttt{#1}}}
\newcommand{\Gray}[1]{\textcolor{gray}{#1}}
\newcommand{\Magenta}[1]{\textcolor{magenta}{#1}}

\newcommand{\tTealBlue}[1]{\ifnum\todotoggle=1\TealBlue{#1}\else{#1}\fi}
\newcommand{\tPeach}[1]{\ifnum\todotoggle=1\Peach{#1}\else{#1}\fi}
\newcommand{\tCyan}[1]{\ifnum\todotoggle=1\Cyan{#1}\else{#1}\fi}
\newcommand{\tRed}[1]{\ifnum\todotoggle=1\Red{#1}\else{#1}\fi}
\newcommand{\tNavy}[1]{\ifnum\todotoggle=1\Navy{#1}\else{#1}\fi}
\newcommand{\tBlue}[1]{\ifnum\todotoggle=1\Blue{#1}\else{#1}\fi}
\newcommand{\tGreen}[1]{\ifnum\todotoggle=1\Green{#1}\else{#1}\fi}
\newcommand{\tBlack}[1]{\ifnum\todotoggle=1\Black{#1}\else{#1}\fi}
\newcommand{\tWhite}[1]{\ifnum\todotoggle=1\White{#1}\else{#1}\fi}
\newcommand{\tCode}[1]{\ifnum\todotoggle=1\Code{#1}\else{#1}\fi}
\newcommand{\tGray}[1]{\ifnum\todotoggle=1\Gray{#1}\else{#1}\fi}
\newcommand{\tMagenta}[1]{\ifnum\todotoggle=1\Magenta{#1}\else{#1}\fi}

\newcommand{\inbrace}[1]{\left \{ #1 \right \}}
\newcommand{\inparen}[1]{\left ( #1 \right )}
\newcommand{\insquare}[1]{\left [ #1 \right ]}
\newcommand{\inangle}[1]{\left \langle #1 \right \rangle}
\newcommand{\infork}[1]{\left \{ \begin{matrix} #1 \end{matrix} \right .}
\newcommand{\inmat}[1]{\begin{matrix} #1 \end{matrix}}
\newcommand{\inbmat}[1]{\begin{bmatrix} #1 \end{bmatrix}}

\newcommand{\inabs}[1]{\begin{vmatrix} #1 \end{vmatrix}}
\newcommand{\norm}[2]{\begin{Vmatrix} #2 \end{Vmatrix}_{#1}}

\newcommand{\set}[1]{\inbrace{#1}}
\newcommand{\setdef}[2]{\set{#1 : #2}}
\newcommand{\defeq}{\stackrel{\mathrm{def}}{=}}

\DeclareMathOperator*{\Ex}{\mathbb{E}}
\DeclareMathOperator*{\Var}{\mathsf{Var}}
\DeclareMathOperator*{\Cov}{\mathsf{Cov}}
\DeclareMathOperator*{\Stab}{\mathsf{Stab}}

\newcommand{\DSBS}{\mathrm{DSBS}}

\newcommand{\sAND}{\ \mathrm{and} \ }

\newcommand{\sIF}{\ \mathrm{if} \ }

\newcommand{\indicator}{\mathbf{1}}

\newcommand{\bit}{\set{0,1}}
\newcommand{\sbit}{\set{1, -1}}

\newcommand{\wtilde}[1]{\widetilde{#1}}
\newcommand{\what}[1]{\widehat{#1}}

\newcommand{\dTV}{d_{\mathrm{TV}}}

\newcommand{\eps}{\varepsilon}

\newcommand{\bbN}{{\mathbb N}}

\newcommand{\bbR}{{\mathbb R}}

\newcommand{\bbZ}{{\mathbb Z}}

\newcommand{\ba}{{\bm a}}
\newcommand{\bb}{{\bm b}}
\newcommand{\bc}{{\bm c}}

\newcommand{\be}{{\bm e}}

\newcommand{\bu}{{\bm u}}
\newcommand{\bv}{{\bm v}}
\newcommand{\bw}{{\bm w}}
\newcommand{\bx}{{\bm x}}
\newcommand{\by}{{\bm y}}

\newcommand{\m}{{\bm m}}

\newcommand{\bU}{\mathbf{U}}
\newcommand{\bV}{\mathbf{V}}
\newcommand{\bX}{\mathbf{X}}
\newcommand{\bY}{\mathbf{Y}}
\newcommand{\bZ}{\mathbf{Z}}

\newcommand{\bsigma}{{\boldsymbol \sigma}}
\newcommand{\balpha}{{\boldsymbol \alpha}}

\newcommand{\calA}{\mathcal{A}}

\newcommand{\calD}{\mathcal{D}}

\newcommand{\calF}{\mathcal{F}}
\newcommand{\calG}{\mathcal{G}}

\newcommand{\calN}{\mathcal{N}}

\newcommand{\calR}{\mathcal{R}}

\newcommand{\calX}{\mathcal{X}}
\newcommand{\calY}{\mathcal{Y}}
\newcommand{\calZ}{\mathcal{Z}}

\newcommand{\poly}{\mathrm{poly}}

\newcommand{\Inf}{\mathrm{Inf}}

\newcommand{\Supp}{\mathrm{Supp}}

\newcommand{\Ber}{\mathrm{Ber}}
\newtheorem{theorem}{Theorem}[section]
\newtheorem{defn}[theorem]{Definition}
\newtheorem{definition}[theorem]{Definition}

\newtheorem{lem}[theorem]{Lemma}
\newtheorem{proposition}[theorem]{Proposition}
\newtheorem{cor}[theorem]{Corollary}
\newtheorem{fact}[theorem]{Fact}
\newtheorem{problem}[theorem]{Problem}

\newtheorem{remark}[theorem]{Remark}

\newenvironment{proof-overview}{\noindent {\em Proof Overview.} \hspace*{1mm}}{\hfill $\Box$}
\newenvironment{proofof}[1]{\noindent {\em Proof of #1.} \hspace*{1mm}}{\hfill $\Box$}
\newenvironment{proof-sketch}{\noindent {\em Sketch of Proof.} \hspace*{1mm}}{\hfill $\Box$}
\newenvironment{proof-idea}{\noindent{\em Proof Idea.}\hspace*{1mm}}{\qed\bigskip}
\newenvironment{proof-attempt}{\noindent{\em Proof Attempt.}\hspace*{1mm}}{\qed\bigskip}
\definecolor{Gred}{RGB}{219, 50, 54}
\definecolor{Ggreen}{RGB}{60, 186, 84}
\definecolor{Gblue}{RGB}{72, 133, 237}
\definecolor{Gyellow}{RGB}{247, 178, 16}

%% file: abstract.tex
\title{Dimension Reduction for Polynomials over Gaussian Space\\ and Applications}

\def\authorsize{\normalsize}
\author{
	\authorsize Badih Ghazi \thanks{MIT, Supported in parts by NSF CCF-1650733 and CCF-1420692. Email: \texttt{badih@mit.edu}} \and
	\authorsize Pritish Kamath \thanks{MIT. Supported in parts by NSF CCF-1420956, CCF-1420692, CCF-1218547 and CCF-1650733. Email: \texttt{pritish@mit.edu}} \and
	\authorsize Prasad Raghavendra \thanks{UC Berkeley. Research supported by Okawa Research Grant and NSF CCF-1408643. Email: \texttt{raghavendra@berkeley.edu}}
	%\authorsize Madhu Sudan \thanks{Harvard University, \texttt{madhu@seas.harvard.edu}}
}

\maketitle

\begin{abstract}
	In this work we introduce a new technique for reducing the dimension of the ambient space of low-degree polynomials in the Gaussian space while preserving their relative correlation structure. As applications, we address the following problems:
	
	\begin{enumerate}[(I)]
		\item {\bf Computability of the Approximately Optimal Noise Stable function over Gaussian space.} The goal here is to find a partition of $\bbR^n$ into $k$ parts, that maximizes the noise stability. An $\eps$-optimal partition is one which is within additive $\eps$ of the optimal noise stability.
		
		De, Mossel \& Neeman (CCC 2017) raised the question of an explicit (computable) bound on the dimension $n_0(\eps)$ in which we can find an $\eps$-optimal partition.
		
		De et al. already provide such an explicit bound. Using our dimension reduction technique, we are able to obtain improved explicit bounds on the dimension $n_0(\eps)$.
		
		\item {\bf Decidability of Approximate Non-Interactive Simulation of Joint Distributions.} A {\em non-interactive simulation} problem is specified by two distributions $P(x,y)$ and $Q(u,v)$: The goal is to determine if two players, Alice and Bob, that observe sequences $X^n$ and $Y^n$ respectively where $\set{(X_i, Y_i)}_{i=1}^n$ are drawn i.i.d. from $P(x,y)$ can generate pairs $U$ and $V$ respectively  (without communicating with each other) with a joint distribution that is arbitrarily close in total variation to $Q(u,v)$. Even when $P$ and $Q$ are extremely simple, it is open in several cases if $P$ can simulate $Q$.%: e.g., $P$ is uniform on the triples $\{(0,0), (0,1), (1,0)\}$ and $Q$ is a ``doubly symmetric binary source'', i.e., $U$ and $V$ are uniform $\pm 1$ variables with correlation say $0.49$, it is open if $P$ can simulate $Q$.
		
		Ghazi, Kamath \& Sudan (FOCS 2016) formulated a gap problem of deciding whether there exists a non-interactive simulation protocol that comes $\eps$-close to simulating $Q$, or does every non-interactive simulation protocol remain $2\eps$-far from simulating $Q$? The main underlying challenge here is to determine an explicit (computable) upper bound on the number of samples $n_0(\eps)$ that can be drawn from $P(x,y)$ to get $\eps$-close to $Q$ (if it were possible at all).
		
		While Ghazi et al. answered the challenge in the special case where $Q$ is a joint distribution over $\bit \times \bit$, it remained open to answer the case where $Q$ is a distribution over larger alphabet, say $[k] \times [k]$ for $k > 2$. Recently De, Mossel \& Neeman (in a follow-up work), address this challenge for all $k \ge 2$. In this work, we are able to recover this result as well, with improved explicit bounds on $n_0(\eps)$.
	\end{enumerate}
	
%	We present a simple technique for reducing the dimension of polynomials in the Gaussian space while preserving their structure. As applications, we recover the recent results of De, Mossel and Neeman \cite{DMN_NoiseStabilityComputable, DMN_NIS_decidable} on the \emph{approximate computability of the optimal noise stable function over Gaussian space} and the \emph{approximate decidability of non-interactive simulation of joint distributions}.

Our technique of dimension reduction for low-degree polynomials is simple and analogous to the Johnson-Lindenstrauss lemma, and could be of independent interest.
	
	%involving some delicate calculations of higher moments of some underlying variables along with hypercontractivity bounds (and in some cases the invariance principle). %\Red{Do we want to keep this sentence?}
\end{abstract}

%% file: sec_intro.tex
\newcommand{\R}{\mathbb{R}}

\section{Introduction} \label{sec:intro}

\subsection{Gaussian Isoperimetry \& Noise Stability} 

Isoperimetric problems over the Gaussian space have become central in various areas of theoretical computer science such as hardness of approximation and learning. 
In its simplest and classic form, the central question in isoperimetry is to determine what is the smallest possible surface area for a body of a given volume.
Alternately, isoperimetric problems can be formulated in terms of the notion of {\em Noise stability}.

Fix a real number $\rho \in [0,1]$.
Suppose $f : \R^n \to [0,1]$ denotes the indicator function of a subset (say $\calA_f$) of the $n$-dimensional Gaussian space ($\mathbb{R}^n$ with the Gaussian measure), then its noise stability $\Stab_\rho(f)$ is the probability that two $\rho$-correlated Gaussians $\bX$, $\bY$ both fall in $\calA_f$.
Specifically, if $\calG_\rho^{\otimes n}$ denotes the distribution of $\rho$-correlated Gaussians in $n$ dimensions, that is, $\bX \sim \gamma_n$ and $(\bY | \bX) \sim (\rho \bX + \sqrt{1-\rho^2} \bZ)$ for $\bZ \sim \gamma_n$.
Then, we can equivalently define noise stability as, $\Stab_\rho(f) = \Pr_{(\bX, \bY) \sim \calG_\rho^{\otimes n}} [f(\bX) = f(\bY)]$.  
More formally, the Ornstein-Uhlenbeck operator $U_\rho$, defined for each $\rho \in [0,1]$, acts on any $f : \bbR^n \to \bbR$ as
\begin{equation*}
(U_\rho f)(\bX) = \int\limits_{\bZ \in \mathbb{R}^n} f\inparen{\rho \cdot \bX + \sqrt{1-\rho^2} \cdot \bZ} \ d\gamma_n(\bZ),
\end{equation*}
The noise stability is then defined as $\Stab_\rho(f) \defeq \Ex_{\bX \sim \gamma_n} [f(\bX) \cdot U_\rho f (\bX)]$.

In terms of noise stability, the simplest isoperimetric problem is to determine, what is the largest possible value of $\Stab_\rho(f)$ for a function $f: \R^n \to [0,1]$ with a given expectation $\Ex[f] = \mu$.
The seminal isoperimetric theorem of Borell \cite{borell1985geometric} shows that indicator function of halfspaces are the most noise-stable among all functions $f: \R^n \to [0,1]$ with a given expectation.
%\vspace*{-0.2cm}

Borell's theorem (along with the invariance principle of \cite{mossel2005noise, mossel2010gaussianbounds}) has had fundamental applications in theoretical computer science, e.g., in the hardness of approximation for Max-Cut under the Unique Games conjecture \cite{khot2007optimal}, and in voting theory \cite{mossel2010gaussianbounds}.

%\vspace*{-0.2cm}

In this work, we will be interested in higher analogues of Borell's theorem for partitions of the Gaussian space in to more than two subsets, or equivalently noise stability of functions $f$ taking values over $[k] = \{0,\ldots, k-1\}$.  
%
%
%Isoperimetric problems have been a corner stone of mathematics, and more recently have been central for some applications in theoretical computer science. In this work, we consider the isoperimetric problem over the Gaussian space, also known as {\em Noise Stability}. The fundamental question here is to understand what is the partition of $\bbR^n$ (under the standard Gaussian measure $\gamma_n$) which maximizes the noise stability of the partition.
%
%
Towards stating these higher analogues of Borell's theorem, let's state Borell's theorem in a more general notation. Let $\Delta_k$ be the probability simplex in $\bbR^k$ (i.e. convex hull of the basis vectors $\set{\be_1, \ldots, \be_k}$). The Ornstein-Uhlenbeck operator naturally extends to vector valued functions $f : \bbR^n \to \bbR^k$ as $U_\rho f = (U_\rho f_1, \ldots, U_\rho f_k)$ (where $f = (f_1, \ldots, f_k)$). The noise stability of functions $f : \bbR^n \to \Delta_k$, is now defined as $\Stab_{\rho}(f) := \Ex_{\bX \sim \gamma_n}[\inangle{f(\bX), U_{\rho}f(\bX)}]$ where $\inangle{ \cdot, \cdot}$ denotes the inner product over $\mathbb{R}^k$. We can similarly define the noise stability of a function $f:\mathbb{R}^n \to [k]$ by embedding $[k]$ in $\Delta_k$, i.e., identifying coordinate $i \in [k]$ with the standard basis vector $\be_i \in \Delta_k$. We can now state Borell's theorem in this notation as follows:\\

\vspace*{-0.2cm}

\noindent {\bf Borell's Theorem} \cite{borell1985geometric}. {\em For any $f : \bbR^n \to \Delta_2$, consider the halfspace function $h = (h_1, h_2) : \bbR^n \to \Delta_2$ given by $h_1(\bX) = \indicator_{\set{\inangle{a, \bX} \ge b}}$ and $h_2(\bX) = 1 - h_1(\bX)$, for suitable $a \in \bbR^n$, $b \in \bbR$ such that $\Ex[f] = \Ex[h]$.\\ Then, $\Stab_\rho(f) \le \Stab_\rho(h)$.}\\

\noindent {\bf Question.} [Maximum Noise Stability (MNS)] Given a positive integer $k \geq 2$ and a tuple $\balpha \in \Delta_k$, what is the maximum noise stability of a function $f:\bbR^n \to \Delta_k$ satisfying the constraint that $\Ex[f] = \balpha$?\\

\vspace*{-0.2cm}

The above question remains open even for $k = 3$. In the particular case where $\balpha = (\frac{1}{k}, \ldots, \frac{1}{k})$, the {\em Standard Simplex Conjecture}\footnote{also referred to as the {\em Peace-Sign Conjecture} when $k =3$.} posits that the maximum noise stability is achieved by a ``standard simplex partition'' \cite{khot2007optimal, isaksson2012maximally}. Even in the special case when $k=3$ and $\balpha = (\frac{1}{3},\frac{1}{3},\frac{1}{3})$, the answer is still tantalizingly open. In fact, a suprising result of \cite{heilman2016standard} shows that when the $\alpha_i$'s are not all equal, the standard simplex partition (an appropriately shifted version thereof) \emph{does not} achieve the maximum noise stability. This indicates that the case $k \geq 3$ is fundamentally different than the case where $k = 2$. The fact that we don't understand optimal partitions for $k \ge 3$, led De, Mossel \& Neeman \cite{DMN_NoiseStabilityComputable} to ask whether the optimal partition is realized in any finite dimension. More formally:\\

\vspace*{-0.2cm}

\noindent {\bf Question.} Given $k \ge 2$, $\rho \in (0,1)$, and $\balpha \in \Delta_k$, let $S_n(\balpha)$ be the optimal noise stability of a function $f : \bbR^n \to \Delta_k$, subject to $\Ex[f] = \balpha$. Is there an $n_0$ such that $S_n(\balpha) = S_{n_0}(\balpha)$ for all $n \ge n_0$?\\

\vspace*{-0.2cm}

Even the above question remains open as of now! In this light, De, Mossel \& Neeman \cite{DMN_NoiseStabilityComputable} ask whether one can obtain an {\em explicitly computable} $n_0 = n_0(k,\rho,\eps)$ such that $S_{n_0}(\balpha) \ge S_n(\balpha) - \eps$ for all $n \in \bbN$, in other words, there exists a function $f : \bbR^{n_0} \to \Delta_k$ that comes $\eps$-close to the maximum achievable noise stability. Note that the challenge is really about $n_0$ being ``explicit'', since some $n_0(\rho,k,\eps)$ always exists, as $S_n(\balpha)$ is a converging sequence as $n \to \infty$.

Indeed, De, Mossel and Neeman obtain such an {\em explicitly computable} function. To do so, they use and build on the theory of \emph{eigenregular polynomials} that were previously studied by \cite{de2014efficient}, which in turn uses other tools such as Malliavin calculus. 
%In addition, they use more tools such as co-area formula and gradient bounds. %and notions from It\^o Calculus, Malliavin Calculus, as well as results about

In this work, we introduce fundamentally different techniques (elaborated on shortly), thereby recovering the result of \cite{DMN_NoiseStabilityComputable}. In particular, we show the following (we use $\calR : \bbR^k \to \Delta_k$ to denote the ``rounding operator'', as in \Cref{def:rounding_op}).

\begin{theorem}[Dimension Bound on Approximately Optimal Noise Stable Function]\label{th:noise-stability-informal}
	Given parameters $k \ge 2$, $\rho \in [0,1]$ and $\eps > 0$, there exists an explicitly computable $n_0 = n_0(\rho, k, \eps)$ such that the following holds:
	
	\noindent For any $n \in \bbN$, let $f : \bbR^n \to \Delta_k$. Then, there exists a function $\wtilde{f} : \bbR^{n_0} \to \Delta_k$ such that
	\begin{enumerate}
		\item $\norm{1}{\Ex[f] - \Ex[\wtilde{f}]} \le \eps$.
		\item $\Stab_\rho(\wtilde{f}) \ge \Stab_\rho(f) - \eps$.
	\end{enumerate}
	Moreover, there exists an explicitly computable $d_0 = d_0(\rho, k, \eps)$ for which there is a degree-$d_0$ polynomial $g : \bbR^{n_0} \to \bbR^k$, such that, $\wtilde{f}(\ba) = \calR\inparen{g\inparen{\frac{\ba}{\|\ba\|_2}}}$.\\%Here $\calR : \bbR^k \to \Delta_k$ is the ``rounding operator''.\\
	
	\noindent The explicit $n_0$ and $d_0$ are upper bounded as $n_0 \le \exp \inparen{\poly\inparen{k, \frac{1}{1-\rho}, \frac{1}{\eps}}}$ and $d_0 \le \poly\inparen{k, \frac{1}{1-\rho}, \frac{1}{\eps}}$.%\footnote{
	%A few remarks are in order:\\	
	%(i) While we do obtain an actual {\em explicit} bound on $n_0$ and $d_0$, we skip it in the theorem statement in order to stress on the qualitative nature of the bound.\\
	%(ii) \cite{DMN_NoiseStabilityComputable} doesn't explicitly state the bound they obtain on $n_0$, instead, it is left buried in the details of the proof. So we are unable to compare our bound to theirs.\\
	%(iii) A subtle point in our theorem is that the range of $\wtilde{f}$ is $\Delta_k$ and not $[k]$. Interestingly however, it follows from a thresholding lemma in \cite[Lemma 15 \& 16]{DMN_NoiseStabilityComputable} that any such $\wtilde{f}$ can be modified to have range $[k]$, while preserving $\Ex[\wtilde{f}]$ without decreasing the noise stability.}
\end{theorem}
\paragraph{Remarks.}
\begin{enumerate}[(i)]
	\item While we do obtain an actual {\em explicit} bound on $n_0$ and $d_0$, we skip it in the theorem statement in order to stress on the qualitative nature of the bound. In contrast, it is mentioned in \cite{DMN_NoiseStabilityComputable} that their bound on $n_0$ is not primitive recursive and has an Ackermann-type growth (which is introduced by their application of the regularity lemma from \cite{de2014efficient}).
	\item A subtle point in our theorem is that the range of $\wtilde{f}$ is $\Delta_k$ and not $[k]$. Interestingly however, it follows from a thresholding lemma in \cite[Lemma 15 \& 16]{DMN_NoiseStabilityComputable} that any such $\wtilde{f}$ can be modified to have range $[k]$, while preserving $\Ex[\wtilde{f}]$ without decreasing the noise stability.
\end{enumerate}

\noindent The above theorem has an immediate application of showing that approximately most-stable voting schemes (among all low-influential voting schemes) can be computed efficiently. We refer the reader to \cite{DMN_NoiseStabilityComputable} for the details of this application.

In order to prove \Cref{th:noise-stability-informal}, we in fact turn to the more general and seemingly harder problem of {\em non-interactive simulation of joint distributions}.

\subsection{Non-Interactive Simulation of Joint Distributions}

Suppose that two players, Alice and Bob, observe the sequence of random variables $(x_1, \ldots, x_n)$ and $(y_1, \ldots, y_n)$ respectively, where each pair $(x_i, y_i)$ is independently drawn from a {\em source} joint distribution $\mu(x,y)$. The fundamental question here is to understand which other {\em target} joint distributions $\nu$ can Alice and Bob simulate, without communicating with each other? How many samples from $\mu$ are needed for the same, or in other words, what is the {\em simulation rate}?

This setup, referred to as the \emph{Non-Interactive Simulation (NIS) of Joint Distributions}, has been extensively studied in Information Theory, and more recently in Theoretical Computer Science. The history of this problem goes back to the classical works of G\'acs and K\"orner \cite{gacs1973common} and Wyner \cite{Wyner_CommonInfo}. Specifically, consider the distribution $\mathsf{Eq}$ over $\bit \times \bit$ where both marginals are $\Ber(1/2)$ and the bits identical with probability $1$. G\'acs and K\"orner studied the special case of this problem corresponding to the target distribution $\nu = \mathsf{Eq}$. They characterized the simulation rate in this case, showing that it is equal to what is now known as the {\em G\'acs-K\"orner common information} of $\mu$. On the other hand, Wyner studied the special case corresponding to the source distribution $\mu = \mathsf{Eq}$. He characterized the simulation rate in this case, showing that it is equal to what is now known as {\em Wyner common information} of $\nu$.

Another particularly important work was by Witsenhausen \cite{witsenhausen1975sequences} who studied the case where the target distribution $\nu = \calG_{\rho}$ is the distribution of $\rho$-correlated Gaussians. In this case, he showed that the largest correlation (i.e., largest value of $\rho$) that can be simulated is exactly the well-known ``maximal correlation coeffcient'' $\rho(\mu)$ (see \Cref{def:max_corr}) which was first introduced by Hirschfeld \cite{hirschfeld1935connection} and Gebelein \cite{gebelein1941statistische} and then studied by R{\'e}nyi \cite{renyi1959measures}. This immediately gives a polynomial time algorithm to decide if $\calG_\rho$ can be simulated from samples from a given $\mu$, since the maximal correlation coefficient $\rho(\mu)$ is efficiently computable. In the same work \cite{witsenhausen1975sequences}, Witsenhausen also considered the case where the target distribution $\nu = \mathrm{DSBS}_\rho$, which is a pair of $\rho$-correlated bits (i.e. a pair of $\pm 1$ random variables with correlation $\rho$), and gave an approach to simulate correlated bits by first simulating $\calG_{\rho}$ starting with samples from $\mu$, and then applying half-space functions to get outputs in $\set{\pm 1}$. Starting with $\mu$, such a approach simulates $\DSBS_{\rho'}$ where $\rho' = 1 - \frac{2\arccos \rho(\mu)}{\pi}$. Indeed, this is morally same as the rounding technique employed in Goemans-Williamson's approximation algorithm for MaxCut~\cite{goemans1995maxcut} 20 years later!

We will consider the modern formulation of the NIS question as defined in \cite{kamath2015non}. This formulation ignores the simulation rate, and only focuses on whether simulation is even possible or not, given infinitely many samples from $\mu$ -- that is, whether the simulation rate is non-zero or not.

\begin{defn}[Non-interactive Simulation of Joint Distributions \cite{kamath2015non}]
	Let $(\calZ \times \calZ, \mu)$ and $([k] \times [k], \nu)$ be two joint probability spaces. We say that the distribution $\nu$ can be {\em non-interactively simulated} from distribution $\mu$, if there exists a sequence of functions\footnote{we will often refer to such functions as {\em strategies} of players Alice and Bob.} $\set{A^{(n)} : \calZ^n \to [k]}_{n \in \bbN}$ and $\set{B^{(n)} : \calZ^n \to [k]}_{n \in \bbN}$ such that the joint distribution $\nu_n = (A^{(n)}(\bx), B^{(n)}(\by))_{(\bx, \by) \sim \mu^{\otimes n}}$ over $[k] \times [k]$ is such that $\lim\limits_{n \to \infty}\dTV(\nu_n, \nu) = 0$.
\end{defn}

\input{fig_NonIntSim}

The notion of non-interactive simulation is summarized in \Cref{fig:non_int_sim}. Note that even though the definition itself doesn't give Alice and Bob access to private randomness, they can nevertheless take extra samples and use them as private randomness. We will model the use of private randomness, by allowing $A^{(n)}$ and $B^{(n)}$ to map to the simplex $\Delta_k$, instead of $[k]$. We will then interpret $A^{(n)}_i(\bx)$ (resp. $B^{(n)}_j(\bx)$) as the probability of Alice outputting $i$ (resp. Bob outputting $j$). For convenience, we will still use $(A^{(n)}(\bx), B^{(n)}(\by))_{(\bx, \by) \sim \mu^{\otimes n}}$ to denote the joint distribution generated over $[k] \times [k]$.

A central question that was left open following the work of Witsenhausen is: given distributions $\mu$ and $\nu$, can $\nu$ be non-interactively simulated from $\mu$? Can this be decided algorithmically? %and a parameter $\eps > 0$, is there an algorithm that runs in time bounded by some function on $\mu$, $\nu$ and $\eps$, and that determines if there exists a non-interactive simulation protocol whose output is $\eps$-close to $\nu$?
Even when $\mu$ and $\nu$ are extremely simple, e.g. $\mu$ is uniform on the triples $\{(0,0), (0,1), (1,0)\}$ and $\nu$ is the doubly symmetric binary souce $\mathrm{DSBS}_{0.49}$, it is open if $\mu$ can simulate $\nu$. This problem was formalized as a natural gap-version of the non-interactive simulation problem in a work by a subset of the authors along with Madhu Sudan \cite{GKS_NIS_decidable}. Here we state a slightly more generalized version.

\begin{problem}[$\ANIS((\calZ \times \calZ, \mu), V, k, \eps)$, cf. \cite{GKS_NIS_decidable}] \label{prob:approx_decide_full}
	Given a joint probability space $(\calZ \times \calZ, \mu)$ and another family of joint probability spaces $V$ supported over $[k] \times [k]$, and an error parameter $\eps > 0$, distinguish between the following cases:
	\begin{enumerate}[(i)]
		\item there exists $N$, and functions $A : \calZ^N \to \Delta_k$ and $B : \calZ^N \to \Delta_k$, for which the distribution $\nu' $ of $(A(\bx), B(\by))_{(\bx, \by) \sim \mu^{\otimes N}}$ is such that $\dTV(\nu', \nu) \le \eps$ for some $\nu \in V$.
		\item for all $N$ and all functions $A : \calZ^N \to \Delta_k$ and $B : \calZ^N \to \Delta_k$, the distribution $\nu'$ of $(A(\bx), B(\by))_{(\bx, \by) \sim \mu^{\otimes N}}$ is such that $\dTV(\nu', \nu) > 2 \eps$ for all $\nu \in V$. \footnote{the choice of constant $2$ is arbitrary. Indeed we could replace it by any constant greater than $1$.}
	\end{enumerate}
\end{problem}

\subsection{NIS from Gaussian Sources \& the MNS question}

We now remark on why the NIS question is a more general question than the Maximum Noise Stability question. For any distribution $\nu$, define the agreement probability of $\nu$ supported over $[k] \times [k]$ as $\mathrm{agr}(\nu) = \Pr_{(u,v) \sim \nu} [u = v]$. Recall that for $f : \bbR^n \to \Delta_k$, the stability can equivalently be defined as $\Stab_\rho(f) = \Ex_{(\bX, \bY) \sim \calG_\rho^{\otimes n}} \inangle{f(\bX), f(\bY)} = \sum_{i=1}^{k} \Ex_{(\bX, \bY) \sim \calG_\rho^{\otimes n}} [f_i(\bX), f_i(\bY)]$. Basically, the MNS question can be interpretted as asking: what is the maximum ``agreement probability'' of any distribution $\nu$ that can be non-interactively simulated from $\mu = \calG_\rho$, with both marginal distributions given by $\balpha$, and with an additional constraint that both Alice and Bob use the same strategy, i.e. $A = B = f$. Thus, to understand the MNS question, we turn to understanding which target distributions $\nu$ can be non-interactively simulated with the source distribution $\mu = \calG_{\rho}$; we will ignore, for the moment, the restriction that Alice and Bob need to use the same strategy.

Recall that implicit in \cite{witsenhausen1975sequences}, was an approach to non-interactively simulate target distributions $\nu$ over $\bit \times \bit$ from $\mu = \calG_{\rho}$ using {\em half-space} functions (using only one sample of $\mu$). Combining Witsenhausen's approach and Borell's theorem \cite{borell1985geometric} gives us an exact characterization of all distributions $\nu$ over $\bit \times \bit$ that can be simulated from $\mu = \calG_{\rho}$. Moreover, any distribution $\nu$ that can be simulated from $\calG_{\rho}$ can in fact be simulated using only one sample from $\calG_{\rho}$ (potentially in addition to private randomness).

For $k > 2$, we do not have such an exact characterization of the distributions $\nu$ that can be simulated from $\calG_{\rho}$. The challenges underlying here are the same as those underlying in the {\em Standard Simplex Conjecture}. Nevertheless, we prove a bound on the number of samples needed to come $\eps$-close to simulating $\nu$, if it were possible at all, in the form of the following theorem,

\begin{theorem}[NIS from correlated Gaussian source]\label{th:NIS_Gaussian_src}
	Given parameters $k \ge 2$, $\rho \in (0,1)$ and $\eps > 0$, there exists an explicitly computable $n_0 = n_0(\rho, k, \eps)$ such that the following holds:
	
	\noindent For any $N$, and any $A : \bbR^N \to \Delta_k$ and $B : \bbR^N \to \Delta_k$, there exist functions $\wtilde{A} : \bbR^{n_0} \to \Delta_k$ and $\wtilde{B} : \bbR^{n_0} \to \Delta_k$ such that,
	\[\dTV\inparen{(A(\bX),B(\bY))_{(\bX,\bY) \sim \calG_\rho^{\otimes N}},\ (\wtilde{A}(\ba), \wtilde{B}(\bb))_{(\ba, \bb) \sim \calG_\rho^{\otimes n_0}}} ~\le~ \eps\;.\]
	Moreover, there exists an explicitly computable $d_0 = d_0(\rho, k, \eps)$ for which there are degree-$d_0$ polynomials $A_0 : \bbR^{n_0} \to \bbR^k$ and $B_0 : \bbR^{n_0} \to \bbR^k$, such that, $\wtilde{A}(\ba) = \calR\inparen{A_0\inparen{\frac{\ba}{\|\ba\|_2}}}$ and $\wtilde{B}(\bb) = \calR\inparen{B_0\inparen{\frac{\bb}{\|\bb\|_2}}}$.\\
	The explicit $n_0$ and $d_0$ are upper bounded as $n_0 \le \exp \inparen{\poly\inparen{k, \frac{1}{1-\rho}, \frac{1}{\eps}}}$ and $d_0 \le \poly\inparen{k, \frac{1}{1-\rho}, \frac{1}{\eps}}$.\\
	
	\noindent In fact, the transformation satisfies a stronger property that there exists an ``oblivious'' randomized transformation (with a shared random seed) to go from $A$ to $\wtilde{A}$ and from $B$ to $\wtilde{B}$, which works with probability at least $1-\eps$. Since the same transformation is applied on $A$ and $B$ simultaneously with the same random seed, if $A = B$, then the transformation gives $\wtilde{A}=\wtilde{B}$ as well.
\end{theorem}

\noindent It is now easy to see that \Cref{th:noise-stability-informal} follows simply as a corollary of the above theorem, when applied on functions $A = B = f$.

By an ``oblivious'' randomized transformation, we mean that to obtain $\wtilde{A}$ from $A$, we only need to know $A$ and a shared random seed $M$. That is, the transformation doesn't use the knowledge of $B$. Similarly, to obtain $\wtilde{B}$ from $B$, we only need to know $B$ and the same shared random seed $M$. This hinted at in \cite{GKS_NIS_decidable} as a potential barrier for showing decidability of $\ANIS$ when $k \ge 2$. Indeed our transformation overcomes this barrier and we elaborate more on this in \Cref{subsec:proof_outline}.

\subsection{NIS from Arbitrary Discrete Sources}
In prior work \cite{GKS_NIS_decidable}, it was shown that $\ANIS$ for discrete distributions $\mu$ and $\nu$ is decidable, in the special case where $k = 2$. This was done by introducing a framework, which reduced the problem to only understanding $\ANIS$ for the special case where $\mu = \calG_\rho$. Indeed, the reason why the case of $k=2$ was easier was precisely because combining Witsenhausen \cite{witsenhausen1975sequences} and Borell's theorem \cite{borell1985geometric}, gives an exact characterization of the distributions over $[2] \times [2]$ that can be simulated from $\calG_\rho$. The lack of understanding of the distributions over $[k] \times [k]$ that can be simulated from $\calG_\rho$ was suggested in \cite{GKS_NIS_decidable} as a barrier for extending their result to $k > 2$.

Following up on \cite{DMN_NoiseStabilityComputable}, De, Mossel \& Neeman were able to extend their techniques to show the decidability of $\ANIS$ for all $k \ge 2$ \cite{DMN_NIS_decidable}. To do so, they follow the same high level framework of using a Regularity Lemma and Invariance Principle introduced in \cite{GKS_NIS_decidable}. In addition, they build on the tools developed in \cite{DMN_NoiseStabilityComputable} along with a new smoothing argument inspired by boosting procedures in learning theory and potential function arguments in complexity theory and additive combinatorics.

In this work, we are able to recover this result using a fundamentally different and more elementary approach, by only using \Cref{th:NIS_Gaussian_src} along with the framework introduced in \cite{GKS_NIS_decidable}, thereby showing decidability of $\ANIS$ for all $k \ge 2$. The central underlying theorem to prove decidability of $\ANIS$ is the following.

\begin{theorem}[NIS from Discrete Sources]\label{th:non-int-sim}
	Let $(\calZ \times \calZ, \mu)$ be a joint probability space. Given parameters $k \ge 2$ and $\eps > 0$, there exists an explicitly computable $n_0 = n_0(\mu, k, \eps)$ such that the following holds:
	
	Let $A : \calZ^N \to \Delta_k$ and $B : \calZ^N \to \Delta_k$. Then there exist functions $\wtilde{A} : \calZ^{n_0} \to \Delta_k$ and $\wtilde{B} : \calZ^{n_0} \to \Delta_k$ such that,
	\[\dTV\inparen{(A(\bx),B(\by))_{(\bx,\by)\sim \mu^{\otimes N}}, \ (\wtilde{A}(\ba), \wtilde{B}(\bb))_{\ba, \bb \sim \mu^{\otimes n_0}}} \le \eps\;.\]
	In particular, $n_0$ is an explicit function upper bounded by $\exp\inparen{\poly\inparen{k, \frac{1}{\eps}, \frac{1}{1-\rho_0}, \log\inparen{\frac{1}{\alpha}}}}$, where $\alpha = \alpha(\mu)$ is the smallest atom in $\mu$ and $\rho_0 = \rho(\mu)$ is the maximal correlation of $\mu$.
\end{theorem}

\noindent The decidability of $\ANIS$ follows quite easily from the above theorem. The main idea is, once we know a bound on the number of samples of $\mu$ that are needed to get $\eps$-close to $\nu$ (if it were possible at all), we can brute force over all possible strategies of Alice and Bob. For completeness, we provide a proof of the following theorem in \Cref{sec:decidability}.

\begin{theorem}[Decidability of $\ANIS$] \label{thm:decidability}
	Given a joint probability space $(\calZ \times \calZ, \mu)$ and a family of joint probability spaces $V$ supported over $[k] \times [k]$, and an error parameter $\eps > 0$, there exists an algorithm that runs in time $T(\mu, k, \eps)$ (which is an explicitly computable function), and decides $\ANIS((\calZ \times \calZ, \mu), V, k, \eps)$.
	
	\noindent The run time $T(\mu, k, \eps)$ is upper bounded by $\exp\exp\exp\inparen{\poly\inparen{k, \ \frac{1}{\eps}, \ \frac{1}{1-\rho_0}, \ \log\inparen{\frac{1}{\alpha}}}}$,
	where $\rho_0 = \rho(\mu)$ is the maximal correlation of $(\calZ \times \calZ, \mu)$ and $\alpha \defeq \alpha(\mu)$ is the minimum non-zero probability in $\mu$.
\end{theorem}

\subsection{Dimension Reduction for Low-Degree Polynomials over Gaussian Space} \label{subsec:dimred}

We now describe the main technique of ``{\em dimension reduction for low-degree polynomials}'' that we introduce in this work, which could be of independent interest. 
%We stress that this technique is the main conceptual message that we wish to convey in this work, and we believe that this technique could be of independent interest.

Let's start with \Cref{th:NIS_Gaussian_src}, and see how we might even begin proving it. We are given two vector-valued functions\footnote{recall that we think of a vector valued function $A : \bbR^n \to \bbR^k$ as a tuple $(A_1, \ldots, A_k)$, where each $A_i : \bbR^n \to \bbR$} $A: \bbR^n \to \Delta_k$ and $B: \bbR^n \to \Delta_k$. We wish to reduce the dimension $n$ of the Gaussian space on which $A$ and $B$ act, while preserving the joint distribution $(A(\bX), B(\bY))_{(\bX, \bY) \sim \calG_\rho^{\otimes n}}$ over $[k] \times [k]$. Observe that $\Ex_{(\bX, \bY) \sim \calG_\rho^{\otimes n}} [A_i(\bX) \cdot B_j(\bY)]$ is the probability of the event [{\em Alice outputs $i$ and Bob outputs $j$}]. We succinctly write this expectation as $\inangle{A_i, B_j}_{\calG_\rho^{\otimes n}}$. In order to approximately preserve the joint distribution $(A(\bX), B(\bY))_{(\bX, \bY) \sim \calG_\rho^{\otimes n}}$, it suffices to approximately preserve $\inangle{A_i, B_j}_{\calG_\rho^{\otimes n}}$ for each $(i, j) \in [k] \times [k]$.

Thus, to prove \Cref{th:NIS_Gaussian_src}, we wish to find an explicit constant $n_0 = n_0(\rho, k, \eps)$, along with functions $\wtilde{A}: \bbR^{n_0} \to \Delta_k$ and $\wtilde{B}: \bbR^{n_0} \to \Delta_k$ such that $\inangle{\wtilde{A}_i, \wtilde{B}_j}_{\calG_{\rho}^{\otimes n_0}} \approx_\eps \inangle{A_i, B_j}_{\calG_{\rho}^{\otimes n}}$. Achieving this directly is highly unclear, since a priori, we have no structural information about $A$ and $B$! To get around this, we show that it is possible to first do a structural transformation on $A$ and $B$ to make them low-degree multilinear polynomials (see \Cref{subsec:analysis_prelim} for formal definitions) -- such transformations are described in \Cref{sec:smoothing,sec:non-multilinear}. This however creates a new problem that the transformed $A$ and $B$ no longer map to $\Delta_k$. Nevertheless, we show that after the said transformations we still have that the outputs of $A$ and $B$ are close to $\Delta_k$ in expected $\ell_2^2$ distance (for now, let's informally denote this by $\mathrm{dist}(A, \Delta_k)$). We show that this ensures that ``rounding'' the outputs of $A$ and $B$ to $\Delta_k$ will approximately preserve the correlations $\inangle{A_i, B_j}_{\calG_\rho^{\otimes n}}$.

We are now able to revise our objective as follows: Given two (vector-valued) degree-$d$ polynomials $A: \bbR^n \to \bbR^k$ and $B: \bbR^n \to \bbR^k$, does there exist an explicitly computable function $D$ of $k$, $d$, and $\delta$, along with polynomials $\wtilde{A}: \bbR^D \to \bbR^k$ and $\wtilde{B}: \bbR^D \to \bbR^k$ that $\delta$-approximately preserves (i) the correlation $\inangle{A_i, B_j}_{\calG_{\rho}^{\otimes n}}$ for all $(i,j) \in [k] \times [k]$ and (ii) closeness of the outputs of $A$ and $B$ to $\Delta_k$ in expected $\ell_2^2$ distance, that is, $\mathrm{dist}(A, \Delta_k)$ and $\mathrm{dist}(B, \Delta_k)$.

We introduce a very simple and natural dimension-reduction procedure for low-degree multilinear polynomials over Gaussian space. Specifically, for an i.i.d. sequence of $\rho$-correlated Gaussians  $(a_1,b_1)$, $(a_2,b_2)$, $\cdots$, $(a_{D},b_{D})$, we set
\begin{equation}\label{eqn:dim-red-substitution}
\wtilde{A}(\ba) := A\inparen{\frac{M \ba}{\norm{2}{\ba}}} \qquad \sAND \qquad \wtilde{B}(\bb) := B\inparen{\frac{M \bb}{\norm{2}{\bb}}}
\end{equation}
where $M$ is a randomly sampled $N \times D$ matrix with i.i.d. standard Gaussian entries. Our main dimension-reduction theorem for low-degree polynomials is stated as follows,

\begin{theorem}[Dimension Reduction Over Gaussian Space]\label{th:dim_red_Gauss_space}
	Given parameters $k \ge 2$, $d \in \calZ_{\ge 0}$, $\rho \in (0,1)$ and $\delta > 0$, there exists an \emph{explicitly computable} $D = D(d, k, \delta)$, such that the following holds:
	
	Let $A: \bbR^N \to \bbR^k$ and $B: \bbR^N \to \bbR^k$ be degree-$d$ multilinear polynomials. Additionally, suppose that $\mathrm{dist}(A, \Delta_k), \mathrm{dist}(B, \Delta_k) \le \delta$. Consider the functions $\wtilde{A} : \bbR^D \to \bbR^k$ and $\wtilde{B} : \bbR^D \to \bbR^k$ as defined in \Cref{eqn:dim-red-substitution}. With {\em probability} at least $1-3\delta$ over the random choice of $M \sim \calN(0,1)^{N \times D}$, the following will hold:
	\begin{itemize}
		\item For every $i, j \in [k]$ : $\inabs{\inangle{A_i, B_j}_{\calG_\rho^{\otimes N}} ~-~ \inangle{\wtilde{A}_i, \wtilde{B}_j}_{\calG_\rho^{\otimes D}}} ~\le~ \delta$.
		\item $\mathrm{dist}(\wtilde{A}, \Delta_k) \le \sqrt{\delta}$ and $\mathrm{dist}(\wtilde{B}, \Delta_k) \le \sqrt{\delta}$.
	\end{itemize}
	%where $\mathrm{dist}(A, \Delta_k)$ is the expected $\ell_2^2$ distance between $A(\bX)$ and $\Delta_k$ when $\bX \sim \gamma_N$.\\
	In particular, $D$ is an explicit function upper bounded by $\exp\inparen{\poly\inparen{d, \log k, \log(\frac{1}{\delta})}}$.
\end{theorem}

It is clear from the construction of $\wtilde{A}$ and $\wtilde{B}$ that this theorem is giving us an ``oblivious'' randomized transformation, as also remarked in \Cref{th:NIS_Gaussian_src}. The proof of \Cref{th:dim_red_Gauss_space} is obtained by combining \Cref{thm:dim-reduction} and \Cref{prop:dim-red-close-to-simplex} in \Cref{sec:dim-reduction}.

\paragraph{Analogy with the Johnson-Lindenstrauss lemma.} We will now highlight a few parallels between \Cref{th:dim_red_Gauss_space} and the proof of the Johnson-Lindenstrauss lemma. Suppose we have two unit vectors $u, v \in \bbR^n$. We wish to obtain a randomized transformation $\Psi_s : \bbR^n \to \bbR^D$ (where $s$ is the random seed), such that, $\inangle{u,v} \approx_\delta \inangle{\Psi_s(u), \Psi_s(v)}$ holds with probability $1-\delta$, over the randomness of seed $s$; note that here $\inangle{\cdot, \cdot}$ denotes the inner product over $\bbR^n$ and $\bbR^D$ respectively. Indeed, there is such a transformation, namely, $\Psi_M(u) = \frac{M \cdot u}{\sqrt{D}}$ where $M \sim \calN(0,1)^{\otimes D \times n}$. Let $F(M) = \inangle{\Psi_M(u), \Psi_M(v)}$. Such a transformation satisfies that,
\[ \Ex_{M} [F(M)] = \inangle{u,v} \quad \sAND \quad \Var_M \inparen{F(M)} ~=~ \frac{\inangle{u,v}^2 + \|u\|_2^2 \|v\|_2^2}{D} ~\le~ \frac{2}{D},\]
where we use that $u$ and $v$ are unit vectors. Thus, if we choose $D = 2/\delta^3$, then we can make the variance smaller than $\delta^3$. Thereby, using Chebyshev's inequality, we get that with probability $1-\delta$, it holds that $|\inangle{\Psi_M(u), \Psi_M(v)} - \inangle{u,v}| \le \delta$. Thus, we have a {\em oblivious} randomized dimension reduction that reduced the dimension of any pair of unit vectors to $O(1/\delta^3)$, independent of $n$. Note that, instead of using Chebyshev's inequality, we could use a much sharper concentration bound to show that $D = O(1/\eps^2 \log (1/\delta))$ suffices to preserve the inner product upto an additive $\eps$, with probability $1-\delta$. However, we described the Chebyshev's inequality version as this is what our proof of \Cref{th:dim_red_Gauss_space} does at a high level.

The problem we are facing, although morally similar, is technically entirely different. We want the reduce the dimension of the domain of a pair of polynomials $A : \bbR^n \to \bbR$ and $B : \bbR^n \to \bbR$. For the moment, consider the transformation such that $\Psi_M(A) : \bbR^D \to \bbR$ is given by $A(M\ba / \sqrt{D})$. Similarly, $\Psi_M(B) = B(M\bb/\sqrt{D})$. Our proof of \Cref{th:dim_red_Gauss_space} proceeds along similar lines as the above proof of Johnson-Lindenstrauss Lemma, that is, by considering $F(M) = \Ex_{(\ba,\bb)\sim \mu^{\otimes D}} [\Psi_M(A)(\ba) \cdot \Psi_M(B)(\bb)]$, and proving bounds on $\Ex_M [F(M)]$ and $\Var(F(M))$. This turns out to be quite delicate! We don't even have $\Ex_M[F(M)] = \Ex_{(\bx,\by)\sim \mu^{\otimes n}} [A(\bx) \cdot B(\by)]$. What we do show is that,
\[ \inabs{\Ex_M[F(M)] ~-~ \Ex_{(\bx,\by)\sim \mu^{\otimes n}} [A(\bx) \cdot B(\by)]} ~\le~ o_D(1) \quad \sAND \quad \Var_M(F(M)) ~\le~ o_D(1)\;, \]
that is, both are decreasing functions in $D$ (with some dependence on $d$, which is the degree of $A$ and $B$). Interestingly however, in the case of $d=1$, it turns out that $F(M)$ is in fact an unbiased estimator. Indeed, this is not a coincidence! We leave it to the interested reader to figure out that in the case of $d=1$, our tranformation is in fact identical to the above described Johnson-Lindenstrauss transformation on the $n$-dimensional space of Hermite coefficients of $A$ and $B$.

Our actual transformation is slightly different, namely $\Psi_M(A) = A(M\ba/\|\ba\|_2)$. This is to ensure the second point in \Cref{th:dim_red_Gauss_space}, about preserving the closeness of the output of $A$ to $\Delta_k$. The proof gets a little more technical, but this is intuitively similar to the above transformation since $\|\ba\|_2$ is tightly concentrated around $\sqrt{D}$.

\subsection{Related Work and Other Motivations}

\textbf{Information Theory.} We point out that several previous works in information theory and theoretical computer science study ``non-interactive simulation'' type of questions. For instance, the non-interactive simulation of joint distributions question studied in this work is a generalization of the Non-Interactive Correlation Distillation problem\footnote{which considered the problem of maximizing agreement on a single bit, in various multi-party settings.} which was studied by \cite{mossel2004coin, mossel2006non}. Moreover, recent works in the information theory community \cite{kamath2015non,beigi2015duality} derive analytical tools (based on hypercontractivity and the so-called \emph{strong data processing constant}) to prove impossibility results for NIS. While these results provide stronger bounds for some sources, they do not give the optimal bounds in general. Finally, the ``non-interactive agreement distillation'' problem studied by \cite{bogdanov2011extracting} can also be viewed as a particular case of the NIS setup.

\textbf{Randomness in Computation.} As discussed in \cite{GKS_NIS_decidable}, one motivation for studying NIS problems stems from the study of the role of randomness in distributed computing. Specifically, recent works in
cryptography \cite{ahlswede1993common,ahlswede1998common,brassard1994secret,csiszar2000common,maurer1993secret, renner2005simple}, quantum computing \cite{nielsen1999conditions, chitambar2008tripartite, delgosha2014impossibility} and communication
complexity \cite{bavarian2014role,CGMS_ISR,ghazi2015communication} study how the ability to solve various computational tasks gets affected by weakening the source of randomness. In this context, it is very natural to ask how well can a source of randomness be transformed into another (more structured) one, which is precisely the setup of non-interactive simulation.

\textbf{Tensor Power problems.} Another motivation comes from the fact that NIS belongs to the class of {\em tensor power} problems, which have been very challenging to analyze. In such questions, the goal is to understand how some combinatorial quantity behaves in terms of the dimensionality of the problem as the dimension tends to infinity. A famous instance of such problems is the \emph{Shannon capacity of a graph} \cite{shannon1956zero,lovasz1979shannon} where the aim is to understand how the independence number of the power of a graph behaves in terms of the exponent. The question of showing the computability of the Shannon capacity remains open to this day \cite{alon2006shannon}. Other examples of such open problems (which are more closely related to NIS) arise in the problems of {\em local state transformation of quantum entanglement} \cite{Beigi_QuantumMaximalCorrelation, DelgoshaBeigi_QuantumHypercontractivity}, the problem of computing the {\em entangled value of a 2-prover 1-round game} (see for, e.g., \cite{KKMTV_HardnessApprox_EntangledGames} and also the open problems \cite{openQIwiki_all_bell_inequalities}) and the problem of computing the {\em amortized value of parallel repetitions of a 2-prover 1-round game} \cite{Raz_ParallelRep,Holenstein_ParallelRep,Rao_ParallelRep,Raz_CounterexampleStrongParallelRep,BHHRRS_RoundingParallelRepUG}. Yet another example of a tensor-power problem is the task of computing the {\em amortized communication complexity of a communication problem}. Braverman-Rao \cite{braverman2011information} showed that this equals the information complexity of the communication problem, however the computability of information complexity was shown only recently \cite{braverman2015information}.

We hope that the recent progress on the Non-Interactive Simulation problem would stimulate progress on these other notable tensor-power problems. A concrete question is whether the techniques used for NIS (regularity lemma, invariance principle, etc.) can be translated to any of the above mentioned setups.

\subsection{Comparisons with recent works of De, Mossel \& Neeman}

Our main theorems \Cref{th:noise-stability-informal} and \Cref{th:non-int-sim} were proved by De, Mossel \& Neeman \cite{DMN_NoiseStabilityComputable, DMN_NIS_decidable} (only qualitatively, although with worse explicit bounds on $n_0$). Our work was inspired by \cite{DMN_NoiseStabilityComputable,DMN_NIS_decidable} through several high-level ideas, such as the use of smoothing and multilinearization transformations (although these tranformations are technically different in our case, so we state and prove our lemmas from scratch). However, the authors hold the opinion that the key insight into ``{\em why dimension reduction is possible}'' provided by the works of De Mossel \& Neeman and the current work are fundamentally different. 
%
%For the problems considered in this work, both techniques seem to have worked. But it is conceivable that for some other problem, only one technique works and other doesn't; maybe vice versa for some third problem! Although, for the problems considered in this work (and also in \cite{DMN_NoiseStabilityComputable,DMN_NIS_decidable}) our technique can be argued to be simpler.

The key insight for dimension reduction in the work of De, Mossel \& Neeman is (quoting \cite{DMN_NoiseStabilityComputable}): {\em ``the fact that a collection of homogeneous polynomials can be replaced by polynomials in bounded dimensions is a tensor analogue of the fact that for any $k$ vectors in $\bbR^n$, there exist $k$ vectors in $\bbR^k$ with the same matrix of inner products''}. In our work, the main intuition for the dimension reduction is an ``oblivious'' dimension reduction technique, much similar to the Johnson-Lindenstrauss Lemma, as described in \Cref{subsec:dimred}.

	While inspired by the works of De, Mossel \& Neeman, we believe that our technique offers a fresh perspective on why it is possible to obtain explicit bounds for the above problems. Moreover, our bound on $n_0(\eps)$ in both cases is ``merely'' exponential in the parameters of the problem, whereas, the bounds in the works of De et al. are not primitive recursive and have an Ackermann-type growth.

\subsection{Outline of Proofs} \label{subsec:proof_outline}
\paragraph{Dimension Reduction for Polynomials.} We being with describing the main ideas behind \Cref{th:dim_red_Gauss_space}. For polynomials $A : \bbR^N \to \bbR$ and $B : \bbR^N \to \bbR$, we apply a second-moment argument to the random variable
\[F(M) := \inangle{A_M, B_M}_{\calG_{\rho}^{\otimes D}},\]
where $A_M$ and $B_M$ are the substitutions in \Cref{eqn:dim-red-substitution}. Specifically, we compute bounds on the mean and variance of $F(M)$ (\Cref{lem:mean_var_bound}); the key point being that these bounds go to $0$ as $D$ gets larger. Thus, we can get an explicit bound on how large $D$ needs to be in order to make the mean deviation and the variance small. Assuming \Cref{lem:mean_var_bound}, we simply apply Chebyshev's inequality in order to upper-bound the probability that this random variable significantly deviates from its mean.

\Cref{lem:mean_var_bound} is the most technical and novel part of this work, and is proved in \Cref{sec:app_dim_red}. To prove these mean and variance bounds, we first analyze the case when $A$ and $B$ are multi-linear monomials (\Cref{subsec:bds_for_multilinear_monom}). Then, via a simple application of hypercontractivity, we use the monomial calculations in order to obtain bounds on the mean and variance for general multilinear polynomials (\Cref{subec:bds_for_gen_multilinear_polys}).

\paragraph{NIS from Gaussian Sources.} We now turn to the proof of \Cref{th:NIS_Gaussian_src} (which immediately implies \Cref{th:noise-stability-informal}). We are given $A : \bbR^N \to \Delta_k$ and $B : \bbR^N \to \Delta_k$, and we want to construct functions $\wtilde{A} : \bbR^{n_0} \to \Delta_k$ and $\wtilde{B} : \bbR^{n_0} \to \Delta_k$ such that the joint distribution of $(\wtilde{A}, \wtilde{B})$ is close (in total variation distance) to that of $(A,B)$.

For any $i, j \in [k]$, we consider the quantity $\Ex_{XY} [A_i(X) \cdot B_j(Y)]$ which is the probability of the event that [{\em Alice outputs $i$ and Bob outputs $j$}]. Across multiple steps, we modify Alice's and Bob's strategies while approximately preserving this quantity for every $i, j$. Note that if we preserve this quantity for every $i, j$ up to an additive $\eps/k^2$, then this ensures that the joint distribution of Alice and Bob's outputs is preserved up to a total variation distance of $\eps$. The first step is a smoothing operation (\Cref{lem:smoothing_main}) that transforms $A$ and $B$ into polynomials $A^{(1)}, B^{(1)}: \bbR^N \to \bbR^k$ that are guaranteed to have (constant) degree $d$. In the second step, we use a multilinearization operation (\Cref{lem:multilin_main}) to convert $A^{(1)}, B^{(1)}$ into \emph{multilinear} degree-$d$ polynomials $A^{(2)}, B^{(2)} : \bbR^{Nt} \to \bbR^k$ (this operation increases the number of variables by a multiplicative $t$ factor). Both these operations preserve the correlation $\inangle{A_i, B_j}$, and keeps the expected $\ell_2^2$ distance of $A$ and $B$ from $\Delta_k$ small. We then apply our main dimension-reduction procedure (\Cref{th:dim_red_Gauss_space}) to obtain \emph{constant-dimensional} functions $A^{(3)}, B^{(3)} : \bbR^{n_0} \to \bbR^k$ that preserve the structure and correlations of $A^{(2)}$ and $B^{(2)}$. At the final step, we set $\wtilde{A}$ and $\wtilde{B}$ to be the roundings of $A^{(3)}$ and $B^{(3)}$ (repsectively) to the closest functions mapping to $\Delta_k$. Our analysis ensures that in each of the above steps, the two correlations of the two functions as well as their individual distances to the probability simplex are approximately preserved.

\input{fig_outline}

\paragraph{NIS from Arbitrary Discrete Sources.} Our proof of \Cref{th:non-int-sim} proceeds by a reduction to \Cref{th:NIS_Gaussian_src}. The reduction uses the framework already developed in \cite{GKS_NIS_decidable}) of using the invariance principle (\Cref{sec:invariance}) as obtained in \cite{mossel2010gaussianbounds,isaksson2012maximally} and a Regularity Lemma (\Cref{sec:regularity}) inspired by \cite{diakonikolas2010regularity}. We also need to use additional smoothing and multilinearization steps (the full details are in \Cref{sec:non-int-sim}).

One key point about \Cref{th:NIS_Gaussian_src} that is crucial for this application is the ``{\em oblivious}'' nature of the dimension reduction. In the framework of \cite{GKS_NIS_decidable}, we need to apply \Cref{th:NIS_Gaussian_src} on a family of strategies $\set{A^{(1)}, \ldots, A^{(T)}}$ and $\set{B^{(1)}, \ldots, B^{(T)}}$, where each $A^{(i)}, B^{(i)} : \bbR^N \to \bbR^k$. That is, we want to be able to reduce the dimensionality of all the $A^{(i)}$'s and $B^{(i)}$'s while simultaneously preserving the joint distribution $(A^{(i)}, B^{(j)})$ for at least a $(1-\eps)$ fraction of the pairs $(i, j) \in [T] \times [T]$. The {\em oblivious} randomized transformation in \Cref{th:NIS_Gaussian_src} gives us that the tranformation done on $A^{(i)}$ depends only on the random seed and not on which $B^{(j)}$ we are comparing it against. Moreover, this transformation works with ``high'' probability, so in expectation we get that the joint distribution is approximately preserved for atleast a $(1-\eps)$-fraction of the pairs $(i,j) \in [T] \times [T]$.

\subsection{Organization of the Paper}
In \Cref{sec:prelim}, we summarize some useful definitions, and prove a couple of simple lemmas that will be useful in the paper. In \Cref{sec:dim-reduction}, we state our main technique of dimension reduction for polynomials, i.e. \Cref{th:dim_red_Gauss_space}, with the key lemmas and proofs in \Cref{sec:app_dim_red}. In \Cref{sec:smoothing,sec:non-multilinear} we describe the transformations to make functions low-degree and multilinear, with proofs deferred to \Cref{apx:smoothing,apx:non-multilinear}.

In \Cref{sec:noise-stability}, we prove \Cref{th:NIS_Gaussian_src}, deriving \Cref{th:noise-stability-informal} as a corollary. Finally, in \Cref{sec:non-int-sim}, we prove \Cref{th:non-int-sim}, for which we need more tools such as the Regularity Lemma and the Invariance Principle, which we provide in \Cref{sec:regularity,sec:invariance}. Finally, for sake of completeness, the proof of \Cref{thm:decidability} is provided in \Cref{sec:decidability}.

To ease the task of navigating the paper, we provide an outline of the paper in \Cref{fig:outline}.

%% file: fig_NonIntSim.tex
\begin{figure}
\begin{center}
\begin{tikzpicture}[scale=0.8, transform shape]

\def \h{1.2}
\def \w{3}
\node[box, minimum width=2.25cm, minimum height=1.5cm, fill=Gred!20] (A) at (0,\h) {\Large \bf Alice};
\node[box, minimum width=2.25cm, minimum height=1.5cm, fill=Gblue!10] (B) at (0,-\h) {\Large \bf Bob};

\node[left] (X) at (-\w, \h)  {\Large $\bx$} edge[thick,-latex] (A);
\node[left] (Y) at (-\w, -\h) {\Large $\by$} edge[thick,-latex] (B);
\node at ([shift={(-0.88,0.1)}]X) {\Large $\calZ^n \ni$};
\node at ([shift={(-0.88,0.1)}]Y) {\Large $\calZ^n \ni$};

\node (U) at (\w, \h)   {\Large $u$}   edge[thick,latex-] (A);
\node (V) at (\w, -\h)  {\Large $v$}   edge[thick,latex-] (B);
\node at ([shift={(0.85,0.03)}]U) {\Large $\in [k]$};
\node at ([shift={(0.85,0)}]V) {\Large $\in [k]$};

\node at (0,2.5*\h) {private randomness} edge[thick,-latex] (A);
\node at (0,-2.5*\h) {private randomness} edge[thick,-latex] (B);

\draw[snake, segment length=2mm] (X) -- (Y) node[left,midway] {\Large $\mu^{\otimes n}\ \ $};
\draw[snake, segment length=2mm] (U) -- (V) node[right,midway] {\Large $\ \ \nu\ ?$};

\end{tikzpicture}
\caption{Non-Interactive Simulation, e.g., as studied in \cite{kamath2015non}}
\label{fig:non_int_sim}
\end{center}
\end{figure}
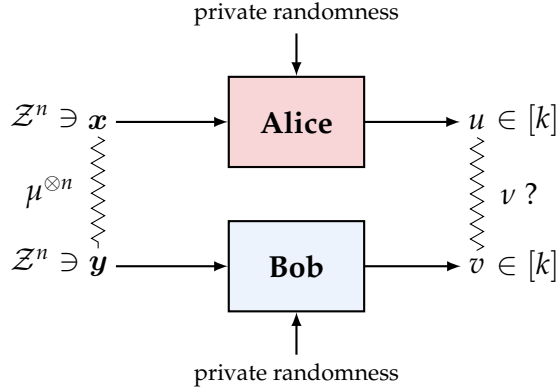

%% file: fig_outline.tex
\begin{figure}
\begin{center}
\begin{tikzpicture}[scale=0.7, transform shape]

\def\ColSmooth{Gblue}
\def\ColMultiLin{Gred}
\def\ColDimRed{Ggreen}
\def\ColRound{Gyellow}
\def\ColRegularity{Ggreen}
\def\ColInvariance{Gyellow}
\def\Brightness{40}

\node[box,text width=3.3cm, align=center] (prelim) at (-1,0) {\large \Cref{sec:prelim}};
\node[box,above,text width=3.3cm, align=center, fill=\ColRound!\Brightness] (prelimtext) at (prelim.north) {Preliminaries: Rounding Lemmas};

\node[box,text width=2.6cm, align=center] (smooth) at (2.5,0) {\large \Cref{sec:smoothing}};
\node[box,above,text width=2.6cm, align=center, fill=\ColSmooth!\Brightness] (smoothtext) at (smooth.north) {Transformation to low-degree};

\node[box,text width=2.6cm, align=center] (smoothapx) at ([shift={(0,4)}]smooth) {\large \Cref{apx:smoothing}}
edge[thick, -latex, dashed] (smoothtext);
\node[box,above,text width=2.6cm, align=center, fill=\ColSmooth!\Brightness] (smoothapxtext) at (smoothapx.north) {Proofs};

\node[box,text width=2.6cm, align=center] (multilin) at (5.7,0) {\large \Cref{sec:non-multilinear}};
\node[box,above,text width=2.6cm, align=center, fill=\ColMultiLin!\Brightness] (multilintext) at (multilin.north) {Transformation to multi-linear};

\node[box,text width=2.6cm, align=center] (multilinapx) at ([shift={(0,4)}]multilin) {\large \Cref{apx:non-multilinear}}
edge[thick, -latex, dashed] (multilintext);
\node[box,above,text width=2.6cm, align=center, fill=\ColMultiLin!\Brightness] (multilinapxtext) at (multilinapx.north) {Proofs};

\node[box,text width=4cm, align=center] (dimred) at (9.6,0) {\large \Cref{sec:dim-reduction}};
\node[box,above,text width=4cm, align=center, fill=\ColDimRed!\Brightness] (dimredtext) at (dimred.north) {Dimension Reduction for low-degree multi-linear polynomials \\ (\Cref{th:dim_red_Gauss_space})};

\node[box,text width=4cm, align=center] (dimredapx) at ([shift={(0,4)}]dimred) {\large \Cref{sec:app_dim_red}}
edge[thick, -latex, dashed] (dimredtext);
\node[box,above,text width=4cm, align=center, fill=\ColDimRed!\Brightness] (dimredapxtext) at (dimredapx.north) {Proofs};

\draw[black!50, thick] ([shift={(-0.3,0.3)}]smoothapxtext.north west) rectangle ([shift={(0.3,-0.3)}]dimredapx.south east);

\draw[thick, -latex] (prelimtext) -- ([shift={(-0.3,-0.3)}]smoothapx.south west);

\node[box,text width=3.3cm, align=center] (reg) at (-1,-8) {\large \Cref{sec:regularity}};
\node[box,above,text width=3.3cm, align=center, fill=\ColRegularity!\Brightness] (regtext) at (reg.north) {Regularity Lemma};

\node[box,text width=3.5cm, align=center] (inv) at (3,-8) {\large \Cref{sec:invariance}};
\node[box,above,text width=3.5cm, align=center, fill=\ColInvariance!\Brightness] (invtext) at (inv.north) {Invariance Principle};

\node[box,text width=4.8cm, align=center] (GaussNIS) at (8,-4.25) {\large \Cref{sec:noise-stability}};
\node[box,above,text width=4.8cm, align=center, upper right=\ColDimRed!\Brightness,upper left=\ColSmooth!\Brightness, lower right=\ColMultiLin!\Brightness,lower left=\ColRound!\Brightness] (GaussNIStext) at (GaussNIS.north) {NIS from Gaussian Sources \\ (\Cref{th:NIS_Gaussian_src}, and as corollary, \Cref{th:noise-stability-informal})}
edge[thick, latex-] (prelim)
edge[thick, latex-] (smooth)
edge[thick, latex-] (multilin)
edge[thick, latex-] (dimred);
%edge[thick, latex-, out=140, in=270] (prelim)
%edge[thick, latex-, out=120, in=270] (smooth)
%edge[thick, latex-, out=100, in=270] (multilin)
%edge[thick, latex-, out=80, in=270] (dimred);

\node[box,text width=4.8cm, align=center] (GenNIS) at (1,-4) {\large \Cref{sec:non-int-sim}}
edge[thick, latex-] (regtext)
edge[thick, latex-] (invtext);
\node[box,above,text width=4.8cm, align=center, upper right=\ColDimRed!\Brightness,upper left=\ColSmooth!\Brightness, lower right=\ColMultiLin!\Brightness,lower left=\ColRound!\Brightness] (GenNIStext) at (GenNIS.north) {NIS from Discrete Sources \\ (\Cref{th:non-int-sim})}
edge[thick, latex-] (prelim)
edge[thick, latex-] (smooth)
edge[thick, latex-] (multilin)
edge[thick, latex-] (GaussNIStext);

\node[box,dashed,text width=4.5cm, align=center] (decide) at (7.8,-8) {\large \Cref{sec:decidability}};
\node[box,above,dashed,text width=4.5cm, align=center, fill=Gred!\Brightness] (decidetext) at (decide.north) {Decidability of $\ANIS$\\ (\Cref{thm:decidability})}
edge[thick, latex-] (GenNIS);

\end{tikzpicture}
\caption{Organization of the paper, where arrows indicate the dependencies of the Sections/Appendices on each other (dashed lines indicate the dependencies that could be skipped on first reading; dashed box indicates an optional appendix)}
\label{fig:outline}
\end{center}
\end{figure}
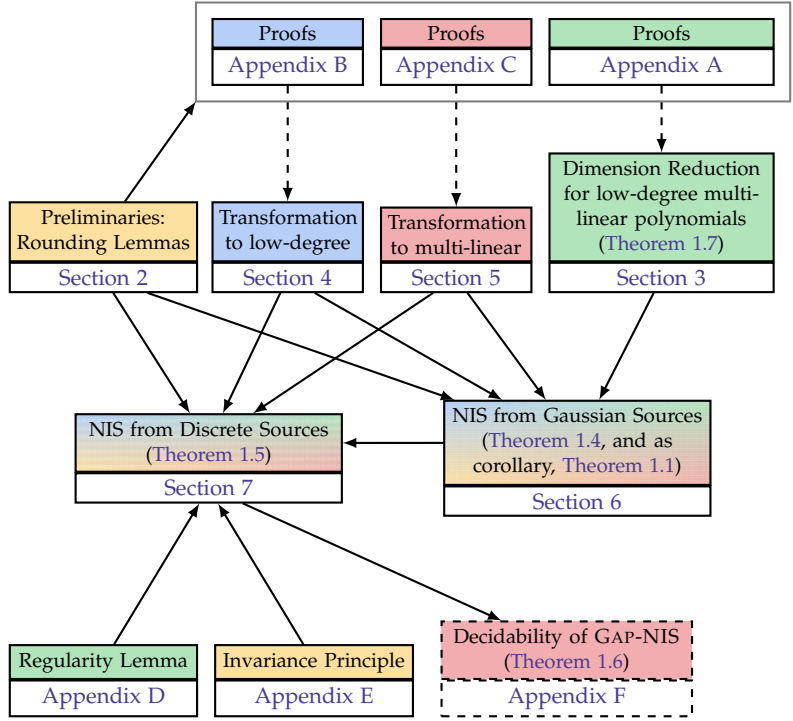

%% file: sec_prelim.tex
\section{Preliminaries} \label{sec:prelim}

\subsection{Probability Spaces : Discrete and Gaussian}

%\paragraph{Discrete Probability spaces \& Fourier Analysis.}
We will mostly use script letter $\calZ$ to denote a finite set of size $q$, and $\mu$ will usually denote a probability distribution. We use small letters $x$, $y$, etc. to denote elements of $\calZ$, and bold small letters $\bx$, $\by$, etc. to denote elements in $\calZ^n$. We use $x_i$, $y_i$ to denote individual coordinates of $\bx$, $\by$ respectively. For a probability space $(\calZ, \mu)$, we will use the following definitions and notations,
\begin{itemize}
\item The pair $(\calZ^n, \mu^{\otimes n})$ denotes the product space $\calZ \times \calZ \times \cdots \times \calZ$ endowed with the product distribution.
\item The support of $\mu$ is $\Supp(\mu) := \setdef{x}{\mu(x) >0}$. We assume w.l.o.g. that $\Supp(\mu) = \calZ$.
\item $\alpha(\mu)$ denotes the minimum non-zero probability atom in $\mu$.
\item $L^2(\calZ,\mu)$ denotes the space of all functions from $\calZ$ to $\bbR$.
\item The inner product on $L^2(\calZ,\mu)$ is denoted by $\inangle{f,g}_{\mu} := \Ex\limits_{x \sim \mu}[f(x) g(x)]$.
\item The $\ell_p$-norm by $\norm{p}{f} := \insquare{\Ex\limits_{x \sim \mu}|f(x)|^p}^{1/p}$. Also, $ \norm{\infty}{f} := \max_{\mu(x) >0} |f(x)|$. %It is easy to verify that $\norm{p}{f} \le \norm{q}{f}$ for $1 \le p \le q$. 
\item For two distributions $\mu$ and $\nu$, $\dTV(\mu, \nu)$ is the total variation distance between $\mu$ and $\nu$.
\end{itemize}
$(\calZ \times \calZ, \mu)$ denotes a joint probability space. We use $\mu_A$ and $\mu_B$ to denote the marginal distributions of $\mu$. The correlation between functions acting on parts of a joint distribution is defined as follows.
\begin{defn}[Correlation between strategies] \label{def:corr_strategies}
	Let $(\calZ \times \calZ, \mu)$ be any joint probability space. For functions $A \in L^2(\calZ, \mu_A)$ and $B \in L^2(\calZ, \mu_B)$, the correlation between $A$ and $B$ over distribution $\mu$ is defined as,
	\[ \inangle{A, B}_{\mu} ~=~ \Ex\limits_{(x, y) \sim \mu} \insquare{A(x) \cdot B(y)}\;.\]
	More generally, if we have functions $A \in L^2(\calZ^n, \mu_A^{\otimes n})$ and $B \in L^2(\calZ^n, \mu_B^{\otimes n})$, the correlation between $A$ and $B$ over distribution $\mu^{\otimes n}$ is defined as,
	\[ \inangle{A, B}_{\mu^{\otimes n}} ~=~ \Ex\limits_{(\bx, \by) \sim \mu^{\otimes n}} \insquare{A(\bx) \cdot B(\by)}\;.\]
\end{defn}
\noindent {\bf Remark.} While $\inangle{A, B}_{\mu^{\otimes n}}$ is the correlation over the joint distribution, the term $\inangle{A, A'}_{\mu_A^{\otimes n}}$ is the correlation as defined earlier over the marginal space. To make this distinction clear, from now on, $\mu$ always refers to the joint distribution on $\calZ \times \calZ$, and we will use $\mu_A$ or $\mu_B$ to indicate distributions over $\calZ$.

\noindent An important quantity associated to any joint distribution is that of the {\em maximal correlation coefficient}, which was first introduced by Hirschfeld \cite{hirschfeld1935connection} and Gebelein \cite{gebelein1941statistische} and then studied by R{\'e}nyi \cite{renyi1959measures}.

\begin{defn}[Maximal correlation] \label{def:max_corr}
	Given a joint probability space $(\calZ \times \calZ, \mu)$, we define its {\em maximal correlation} $\rho(\calZ \times \calZ; \mu)$ (or simply $\rho(\mu)$) as follows,
	\[ \rho(\calZ \times \calZ; \mu) ~=~ \sup_{f, g} \set{\inangle{f,g}_{\mu} \ \ \bigg | \ \ \inmat{f : \calZ \to \bbR, & \Ex[f] = \Ex[g]= 0\\ g : \calZ \to \bbR, & \Var(f) = \Var(g) = 1}} \]
	%For ease of notation, we will simply use $\rho(\mu)$ to denote the maximal correlation.
\end{defn}

Although the above definitions were stated for distributions over finite sets, they extend naturally to the case where $\calZ = \bbR$, equipped with the Gaussian measure $\calN(0,1)$ (also denoted as $\gamma_1$). To distinguish between discrete and Gaussian spaces, we will use capital letters $X$, $Y$, etc. to denote elements of $\bbR$, and bold letters $\bX$, $\bY$, etc. to denote elements in $\bbR^n$. In this case, we use $X_i$, $Y_i$ to denote individual coordinates of $\bX$, $\bY$ respectively. The pair $(\bbR^n, \gamma_n)$ denotes the product space $\bbR^n$ endowed with the standard $n$-dimensional Gaussian measure. Unless explicitly mentioned otherwise, all the functions with domain $\bbR^n$ that we consider will be in $L^2(\bbR^n, \gamma_n)$, which is the space of $\ell_2$-integrable functions with respect to the $\gamma_n$ measure.

Over the space of reals, we will primarily consider the joint distribution of $\rho$-correlated Gaussians $(\bbR \times \bbR, \calG_\rho)$. This is a 2-dimensional Gaussian distributions $(X,Y)$, where $X$ and $Y$ are marginally distributed according to $\gamma_1$, with $\Ex[XY] = \rho$. It is well-known that the maximal correlation of $\calG_\rho$ is $\rho$.

%\Red{we typically use capital letters $A$ and $B$ for such functions, corresponding to ``Alice'' and ``Bob''.}

\subsection{Fourier \& Hermite Analysis} \label{subsec:analysis_prelim}

\paragraph{Fourier analysis for discrete product spaces.} We recall some background in Fourier analysis that will be useful to us. Let $(\calZ,\mu_A)$ be a finite probability space with $|\calZ| = q$. Let $\calX_0, \cdots, \calX_{q-1}: \calZ \to \bbR$ be an orthonormal basis for the space $L^2(\calZ,\mu_A)$ with respect to the inner product $\inangle{.,.}_{\mu_A}$. Furthermore, we require that this basis has the property that $\calX_0 = {\bf 1}$, i.e., the function that is identically $1$ on every element of $\calZ$.

For $\bsigma = (\sigma_1, \cdots, \sigma_n) \in \bbZ_q^n$, define $\calX_{\bsigma} : \calZ^n \to \bbR^n$ as follows,
\[ \calX_{\bsigma}(x_1,\dots,x_n) = \prod_{i \in [n]} \calX_{\sigma_i}(x_i) \]
It is easily seen that the functions $\setdef{\calX_{\bsigma}}{\bsigma \in \bbZ_q^n}$ form an orthonormal basis for the product space $L^2(\calZ^n, \mu_A^{\otimes n})$. Thus, every function $A \in L^2(\calZ^n, \mu_A^{\otimes n})$ has a {\em Fourier expansion} given by
$$A(\bx) = \sum\limits_{\bsigma \in \bbZ_q^n} \what{A}(\bsigma) \calX_{\bsigma}(\bx)\;,$$
where $\what{A}(\bsigma)$'s are the {\em Fourier coefficients} of $A$, which can be obtained as $\what{A}(\bsigma) = \inangle{A,\calX_{\bsigma}}_{\mu_A}$. Although we will work with an arbitrary (albeit fixed) basis, many of the important properties of the Fourier transform are basis-independent. For example, Parseval's identity states that $\norm{2}{A}^2 = \sum_{\bsigma \in \bbZ_{\ge 0}^n} \what{A}(\bsigma)^2$.
%where $\what{A}: \bbZ_q^n \to \bbR$ can be obtained as $\what{A}(\bsigma) = \inangle{A,\calX_{\bsigma}}_{\mu_A}$. The function $\what{A}$ is the Fourier transform of $A$ with respect to the basis $\set{\calX_i}_{i \in \bbZ_q}$. Although we will work with an arbitrary (albeit fixed) basis, many of the important properties of the Fourier transform are basis-independent.

For a joint probability space $(\calZ \times \calZ, \mu)$, we let $\calX_0, \cdots, \calX_{q-1}: \calZ \to \bbR$ be an orthonormal basis for the space $L^2(\calZ,\mu_A)$, and $\calY_0, \cdots, \calY_{q-1}: \calZ \to \bbR$ be an orthonormal basis for the space $L^2(\calZ,\mu_B)$. Although we could choose these basis independently, it is helpful to choose the basis such that, $\inangle{\calX_i, \calY_j}_{\mu} = \rho_i \cdot \mathbf{1}_{i = j}$, where $\rho_{q-1} \le \cdots \le \rho_1 = \rho(\mu)$ (here, $\rho(\mu)$ is the maximal correlation of $\mu$ as in \Cref{def:max_corr}).

For $\bsigma \in \mathbb{Z}_q^n$, the {\em degree} of $\bsigma$ is denoted by $\inabs{\bsigma} \defeq \inabs{\setdef{i \in [n]}{\sigma_i \neq 0}}$. We say that the degree of a function\footnote{we will interchangeably use the word {\em polynomial} to talk about any function in $L^2(\calA^n, \mu^{\otimes n})$.} $A \in L^2(\calZ^n, \mu_A^{\otimes n})$, denoted by $\deg(A)$, is the largest value of $|\bsigma|$ such that $\what{A}(\bsigma) \ne 0$.

\paragraph{Hermite Analysis for Gaussian space.} Analogous to discrete spaces, the set of Hermite polynomials $\setdef{H_r : \bbR \to \bbR}{r \in \bbZ_{\ge 0}}$ form an orthonormal basis for functions in $L^2(\bbR, \gamma_1)$ with respect to the inner product $\inangle{\cdot, \cdot}_{\gamma_1}$. The Hermite polynomial $H_r : \bbR \to \bbR$ (for $r \in \bbZ_{\ge 0}$) is defined as,
\[ H_0(x) = 1; \quad H_1(x) = x; \quad H_r(x) = \frac{(-1)^r}{\sqrt{r!}} e^{x^2/2} \cdot \frac{d^r}{dx^r}e^{-x^2/2}\;. \]
Hermite polynomials can also be obtained via the generating function, $e^{xt - \frac{t^2}{2}} = \sum_{r=0}^{\infty} \frac{H_r(x)}{\sqrt{r!}} \cdot t^r$.

\noindent For any $\bsigma = (\sigma_1, \ldots, \sigma_n) \in \bbZ_{\ge 0}^n$, define $H_\bsigma : \bbR^n \to \bbR$ as
\[ H_\bsigma(\bX) ~=~ \prod_{i=1}^{n} H_{\sigma_i}(X_i).\]
It is easily follows that the set $\setdef{H_\bsigma}{\bsigma \in \bbZ_{\ge 0}^n}$ forms an orthonormal basis for $L^2(\bbR^n, \gamma_n)$. Thus, every $A \in L^2(\bbR^n, \gamma_n)$ has a {\em Hermite expansion} given by
\[ A(\bX) = \sum_{\sigma \in \bbZ_{\ge 0}^n} \what{A}(\bsigma) \cdot H_{\bsigma}(\bX)\;, \]
where the $\what{A}(\bsigma)$'s are the {\em Hermite coefficients} of $A$, which can be obtained as $\what{A}(\bsigma) = \inangle{A,H_{\bsigma}}_{\gamma_n}$. The degree of $\bsigma$ is defined as $|\bsigma| := \sum_{i \in [n]} \sigma_i$, and the degree of $A$ is the largest $|\bsigma|$ for which $\what{A}(\bsigma) \ne 0$. Analogous to Boolean functions, we have Parseval's identity, that is, $\norm{2}{A}^2 = \sum_{\bsigma \in \bbZ_{\ge 0}^n} \what{A}(\bsigma)^2$. We say that $A \in L^2(\bbR^n, \gamma_n)$ is {\em multilinear} if $\what{A}(\bsigma)$ is non-zero only if $\sigma_i \in \bit$ for all $i \in [n]$.

\subsection{Vector-valued functions}

We will extensively work with vector-valued functions. For any function $A : \calD \to \bbR^k$ (for any domain $\calD$), we will write $A = (A_1, \cdots, A_k)$, where $A_i : \calD \to \bbR$ is the $i$-th coordinate of the output of $A$. That is, $A_i(x) = (A(x))_i$ for any $x \in \calD$.

The definitions of Fourier analysis and Hermite analysis extend naturally to vector-valued functions. For $A : \bbR^n \to \bbR^k$, we use $\what{A}(\bsigma)$ to denote the vector $\inparen{\what{A}_1(\bsigma), \ldots, \what{A}_k(\bsigma)}$. In this setting, $\norm{2}{A}^2 := \Ex_{\bX \sim \gamma_n} \norm{2}{A(\bX)}^2 = \norm{2}{A_1}^2 + \cdots + \norm{2}{A_k}^2 = \sum_{\bsigma \in \bbZ_{\ge 0}^n} \norm{2}{\what{A}(\bsigma)}^2$. Also, $\deg(A)$ is defined as $\max_{i \in [k]} \deg(A_i)$. Again, unless explicitly mentioned otherwise, all the vector-valued functions with domain $\bbR^n$ that we consider will be such that the function in each coordinate is in $L^2(\bbR^n, \gamma_n)$.

For $k \in \bbN$, and $i \in [k]$, let $\be_i$ be the unit vector along coordinate $i$ in $\bbR^k$. The simplex $\Delta_k$ is defined as the convex hull formed by $\set{\be_i}_{i \in [k]}$. Equivalently, $\Delta_k = \setdef{\bv \in \bbR^k}{\norm{1}{\bv} = 1}$ is the set of probability distributions over $[k]$. While we consider vector-valued functions mapping to $\bbR^k$, we are primarily interested in functions which map to $\Delta_k$. We use the rouding operator, defined as follows, in order to change the range of a function from $\bbR^k$ to $\Delta_k$.

\begin{defn}[Rounding operator]\label{def:rounding_op}
	The {\em Rounding operator} $\calR^{(k)} : \bbR^k \to \Delta_k$ maps any $\bv \in \bbR^k$ to its closest point in $\Delta_k$. In particular, it is the identity map on $\Delta_k$. We will drop the superscript, as $k$ is fixed throughout this paper.
	
	As for vector-valued functions, we use $\calR_i$ to denote the $i$-th coordinate of $\calR$. Thus, while the $i$-th coordinate of $A$ is denoted by $A_i$, the $i$-th coordinate of $\calR(A)$ is denoted by $\calR_i(A)$.
\end{defn}

\iffalse
\begin{fact}[Properties of Rounding operator]
	For any $k \in \bbN$, $\calR^{(k)}$ has the following properties.
	\begin{itemize}
		\item $\calR^{(k)}$ has Lipschitz constant of $1$. That is, $\calR^{(k)}(\bu - \bv) \le \norm{2}{\bu - \bv}$.
		\item The function $\Psi(\bu, \bv) := \calR^{(k)}_i(\bu) \cdot \calR^{(k)}_j(\bv)$ has Lipschitz constant of $O(\sqrt{k})$.
	\end{itemize}
\end{fact}
\fi

\subsubsection*{Useful lemmas for $\ell_2$-close strategies}

An important relaxation in our work is to consider strategies that do not map to $\Delta_k$, but instead map to $\bbR^k$. For such strategies to be meaningful, we will require that the outputs are {\em usually close} to $\Delta_k$. In this case, we will be {\em rounding} them to the simplex $\Delta_k$.

The following simple lemmas are going to be very useful. The first lemma says that if we modify the strategies of Alice and Bob such that they remain close in $\ell_2$-distance, then the correlation between their strategies does not change significantly.

\begin{lem}[Close strategies in $\ell_2$, have similar correlations]\label{lem:close-strategies}
	Given any joint probability space $(\calZ \times \calZ, \mu)$. Let $A, \wtilde{A} \in L^2(\calZ^n, \mu_A^{\otimes n})$ and $B, \wtilde{B} \in L^2(\calZ^n, \mu_B^{\otimes n})$ such that $\norm{2}{A}, \norm{2}{\wtilde{A}}, \norm{2}{B},\norm{2}{\wtilde{B}} \le 1$.\\
	If $\norm{2}{A - \wtilde{A}} \le \eps$ and $\norm{2}{B - \wtilde{B}} \le \eps$, then it holds that,
	\[ \inabs{\inangle{\wtilde{A}, \wtilde{B}}_{\mu^{\otimes n}} - \inangle{A, B}_{\mu^{\otimes n}}} ~\le~ 2\eps\]
	%\[ \inabs{\Ex\limits_{(x,y) \sim \mu^{\otimes n}} \wtilde{A}(x) \cdot \wtilde{B}(y) - \Ex\limits_{(x,y) \sim \mu^{\otimes n}} A(x) \cdot B(y)} ~\le~ 2\eps\]
\end{lem}
\begin{proof}%[Proof of \Cref{lem:close-strategies}]
	The proof follows very easily from the Cauchy-Schwarz inequality. In particular,
	\begin{align*}
	%		&\inabs{\Ex\limits_{(x,y) \sim \mu^{\otimes n}} \wtilde{A}(x) \cdot \wtilde{B}(y) - \Ex\limits_{(x,y) \sim \mu^{\otimes n}} A(x) \cdot B(y)}\\
	%		&~=~ \inabs{\Ex\limits_{(x,y) \sim \mu^{\otimes n}} \inparen{\wtilde{A}(x) - A(x)} \cdot \wtilde{B}(y) +  A(x) \cdot \inparen{\wtilde{B}(y) - B(y)}}\\
	%		&~\le~ \inabs{\Ex\limits_{(x,y) \sim \mu^{\otimes n}} \inparen{\wtilde{A}(x) - A(x)} \cdot \wtilde{B}(y)} +  \inabs{\Ex\limits_{(x,y) \sim \mu^{\otimes n}} A(x) \cdot \inparen{\wtilde{B}(y) - B(y)}}\\
	\inabs{\inangle{\wtilde{A}, \wtilde{B}}_{\mu^{\otimes n}} - \inangle{A, B}_{\mu^{\otimes n}}}
	&~=~ \inabs{\inangle{(\wtilde{A} - A), \ \wtilde{B}}_{\mu^{\otimes n}} + \inangle{A,\ (\wtilde{B} - B)}_{\mu^{\otimes n}}}\\
	&~\le~ \inabs{\inangle{(\wtilde{A} - A), \ \wtilde{B}}_{\mu^{\otimes n}}} + \inabs{\inangle{A,\ (\wtilde{B} - B)}_{\mu^{\otimes n}}}\\
	&~\le~ \norm{2}{\wtilde{A} - A} \cdot \norm{2}{\wtilde{B}} + \norm{2}{\wtilde{B} - B} \cdot \norm{2}{A} \qquad \ldots \text{(Cauchy-Schwarz inequality)}\\
	&~\le~ 2\eps\;. \qedhere
	\end{align*}
\end{proof}
%The proof of \Cref{lem:close-strategies} appears in \Cref{sec:app_prelim}.
\noindent The second lemma says that if we have two strategies which are close in $\ell_2$-distance, and one of them is {\em close} to the simplex $\Delta_k$, then so is the other. The proof follows by a straightforward triangle inequality.

\begin{lem}\label{lem:close-to-simplex}
	Given any joint probability space $(\calZ \times \calZ, \mu)$. Let $A : \calZ^n \to \bbR^k$ and $\wtilde{A} : \calZ^n \to \bbR^k$ such that $\norm{2}{A}, \norm{2}{\wtilde{A}} \le 1$. Then, for the rounding operator $\calR : \bbR^k \to \Delta_k$, it holds that,
	\[ \norm{2}{\calR(\wtilde{A}) - \wtilde{A}} ~\le~ \norm{2}{\calR(A) - A} + \norm{2}{A - \wtilde{A}}\;.\]
	%\[ \inabs{\Ex\limits_{(x,y) \sim \mu^{\otimes n}} A_1(x) \cdot B_1(y) - \Ex\limits_{(x,y) \sim \mu^{\otimes n}} A(x) \cdot B(y)} ~\le~ 2\eps\]
\end{lem}
\begin{proof}%[Proof of \Cref{lem:close-to-simplex}]
	The proof follows very easily from a triangle inequality. In particular,
	\begin{align*}
	\norm{2}{\calR(\wtilde{A}) - \wtilde{A}}
	&~\le~ \norm{2}{\calR(A) - \wtilde{A}} && \text{(since $\calR(\wtilde{A}(\bx))$ closest in $\Delta_k$ to $\wtilde{A}(\bx)$)}\\
	&~\le~ \norm{2}{\calR(A) - A} + \norm{2}{A - \wtilde{A}} && \text{(Triangle inequality)} \qedhere
	\end{align*}
\end{proof}

%% file: sec_dim-reduction.tex
\section{Dimension Reduction for Low-Degree Multilinear Polynomials} \label{sec:dim-reduction}

In this section, we present our main technique which is a dimension reduction for low-degree multilinear polynomials over Gaussian space, and prove \Cref{th:dim_red_Gauss_space}, which is obtained immediately as a combination of \Cref{thm:dim-reduction} and \Cref{prop:dim-red-close-to-simplex} stated below.

\begin{restatable}{theorem}{dim-reduction}\label{thm:dim-reduction}
	Given parameters $d \in \bbZ_{> 0}$, $\rho \in [0,1]$ and $\delta > 0$, there exists an \emph{explicitly computable} $D = D(d, \delta)$, such that the following holds:
	
	Let $A: \bbR^N \to \bbR$ and $B: \bbR^N \to \bbR$ be degree-$d$ multilinear polynomials, such that $\norm{2}{A}, \norm{2}{B} \le 1$.\\
	For column vectors $\ba, \bb \in \bbR^D$ and $M \in \bbR^{N \times D}$, define the functions $A_M : \bbR^D \to \bbR$ and $B_M : \bbR^D \to \bbR$ as
	\[A_M(\ba) ~=~ A\inparen{\frac{M \ba}{\norm{2}{\ba}}} \qquad \sAND \qquad B_M(\bb) = B\inparen{\frac{M \bb}{\norm{2}{\bb}}}\]
	
	\noindent Sample $M \sim \calN(0,1)^{\otimes (N \times D)}$. Then, with probability at least $1-\delta$ over the choice of $M$, it holds that,
	\[\inabs{\inangle{A_M, B_M}_{\calG_{\rho}^{\otimes D}} ~-~ \inangle{A, B}_{\calG_{\rho}^{\otimes N}}} ~<~ \delta\;.\]
	In particular, one may take $D = \frac{d^{O(d)}}{\delta^4}$.\\
	
	\noindent In other words, for a typical choice of $M \sim \calN(0,1)^{\otimes (N \times D)}$, the correlation between $A$ and $B$ is approximately preserved if we replace $(\bX,\bY) \sim \calG_{\rho}^{\otimes N}$ by $(M \ba/\norm{2}{\ba}, M \bb/\norm{2}{\bb})$, where $(\ba, \bb) \sim \calG_{\rho}^{\otimes D}$. Intuitively, $M$ can be thought of as a means to ``stretch'' $D$ coordinates of $\calG_{\rho}$ into effectively $N$ coordinates of $\calG_{\rho}$, while ``fooling'' correlations between degree-$d$ multilinear polynomials.
\end{restatable}

\noindent Before we prove the above theorem, we prove a simple proposition which shows that if this dimension reduction were applied to vector-valued functions whose outputs lie close to the simplex $\Delta_k$, then with high probability, even the dimension-reduced functions will also have outputs close to the simplex. More formally,

\begin{proposition}\label{prop:dim-red-close-to-simplex}
	Let $A : \bbR^N \to \bbR^k$ and $B: \bbR^N \to \bbR^k$, such that $\norm{2}{\calR(A) - A}, \norm{2}{\calR(B) - B} \le \delta$. Let $A_M : \bbR^D \to \bbR^k$ and $B_M : \bbR^D \to \bbR^k$ be defined analogously to \Cref{thm:dim-reduction}. For $M \sim \calN(0,1)^{\otimes (N \times D)}$, with probability at least $1-2\delta$, it holds that,
	\[ \norm{2}{\calR(A_M) - A_M} \le \sqrt{\delta} \qquad \sAND \qquad \norm{2}{\calR(B_M) - B_M} \le \sqrt{\delta}\;. \]
\end{proposition}
\begin{proof}
	We first observe that for any fixed $\ba \in \bbR^D$, the distribution of $\frac{M\ba}{\norm{2}{\ba}}$ is identical to that of a standard $N$-variate Gaussian distribution. Thus, we immediately have that,
	\[ \Ex_M \Ex_{\ba} \norm{2}{\calR\inparen{A\inparen{\frac{M\ba}{\norm{2}{\ba}}}} - A\inparen{\frac{M\ba}{\norm{2}{\ba}}}}^2 ~=~ \Ex_{\bX} \norm{2}{\calR\inparen{A\inparen{\bX}} - A\inparen{\bX}}^2\]
	Alternately,
	\[ \Ex_M \norm{2}{\calR(A_M) - A_M}^2 = \norm{2}{\calR(A) - A}^2 \le \delta^2\]
	Thus, using Markov's inequality, we get that with probability at least $1 - \delta$,
	\[ \norm{2}{\calR(A_M) - A_M} \le \sqrt{\delta}. \]
	We similarly argue for $B_M$, and a union bound completes the proof.
\end{proof}

\noindent To prove \Cref{thm:dim-reduction}, we primarily use the second moment method (i.e., Chebyshev's inequality). In particular, let $F(M)$ be defined as,
\[F(M) ~\defeq~ \inangle{A_M, B_M}_{\calG_{\rho}^{\otimes D}}\]
The most technical part of this work is to show sufficently good bounds on the mean and variance of $F(M)$ for a random choice of $M \sim \calN(0,1)^{\otimes (N \times D)}$, given by the following lemma.

\begin{restatable}{lem}{meanvarbound}{\em (Bound on Mean \& Variance).}\label{lem:mean_var_bound}
	Given parameters $d$ and $\delta$, there exists $D := D(d, \delta)$ such that the following holds: For $M \sim \calN(0,1)^{\otimes (N \times D)}$,
	\begin{align*}
	\inabs{\Ex\limits_{M} F(M) - \inangle{A, B}_{\calG_{\rho}^{\otimes N}}} &~\le~ \delta \qquad \text{\bf (Mean bound)}\\
	\Var_M \inparen{F(M)} &~\le~ \delta \qquad \text{\bf (Variance bound)}
	\end{align*}
	In particular, one may take $D = \frac{d^{O(d)}}{\delta^2}$.
\end{restatable}

%\begin{restatable}{lem}{variancebound}\label{lem:variance_bound} {\em (Bound on Variance).}
%	Given parameters $d$ and $\delta$, there exists $D := D(d, \delta)$ such that the following holds:
%	\[\Var_M \inparen{F(M)} ~\le~ \delta\]
%	In particular, one may take $D = \frac{d^{O(d)}}{\delta^2}$.
%\end{restatable}

\noindent The proof of \Cref{lem:mean_var_bound} appears in \Cref{sec:app_dim_red}. Assuming \Cref{lem:mean_var_bound}, we can easily prove \Cref{thm:dim-reduction}.

\begin{proof}[Proof of \Cref{thm:dim-reduction}]
	We invoke \Cref{lem:mean_var_bound} with parameters $d$ and $\delta^2/2$, and we get a choice of $D = \frac{d^{O(d)}}{\delta^4}$. Using Chebyshev's inequality and using the Variance bound in \Cref{lem:mean_var_bound}, we have that for any $\eta > 0$,
	\[ \Pr_M \insquare{\inabs{F(M) - \Ex_M F(M)} > \eta} ~\le~ \frac{\delta^2}{2\eta}\;.\]
	Using the triangle inequality, and the Mean bound in \Cref{lem:mean_var_bound}, we get
	\begin{align*}
		&\Pr_M \insquare{\inabs{F(M) - \inangle{A, B}_{\calG_{\rho}^{\otimes N}}} > \delta}\\
		& ~\le~ \Pr_M \insquare{\inabs{F(M) - \Ex_M F(M)} + \inabs{\Ex_M F(M) - \inangle{A, B}_{\calG_{\rho}^{\otimes N}}} > \delta}\\
		& ~\le~ \Pr_M \insquare{\inabs{F(M) - \Ex_M F(M)} > \delta - \delta^2}\\
		& ~\le~ \delta.\qedhere
	\end{align*}
	%This completes the proof of \Cref{thm:dim-reduction}.
\end{proof}

%% file: sec_smoothing.tex
\section{Reduction from General Polynomials to Low-Degree Polynomials} \label{sec:smoothing}

In this section, we state a lemma which says that functions in $L^2(\bbR^n, \gamma_n)$, and more generally in $L^2(\calZ^n, \mu^{\otimes n})$, can be converted to low-degree polynomials while approximately preserving correlations with other functions and also not deviating much from the simplex $\Delta_k$. This technique is considered quite standard and was also used in \cite{GKS_NIS_decidable, DMN_NoiseStabilityComputable, DMN_NIS_decidable} for the same reason. For completeness, we provide the proof in \Cref{apx:smoothing}.

%\noindent The following is the main lemma that allows us to convert functions to low-degree ones, . The lemma holds both for functions over correlated discrete hypercube as well as for functions over correlated Gaussian spaces.

\begin{lem}[Main Smoothing Lemma] \label{lem:smoothing_main}
	Let $\rho \in [0,1]$, $\delta > 0$, $k \in \bbN$ be any given constant parameters. There exists an explicit $d = d(\rho, k, \delta)$ such that the following holds:\\
	
	\noindent {\bf (Correlated Discrete Hypercube):} Let $(\calZ \times \calZ, \mu)$ be a joint probability space, with $\rho(\calZ, \calZ; \mu) = \rho$. Let $A : \calZ^n \to \bbR^k$ and $B : \calZ^n \to \bbR^k$, such that, for any $j \in [k]$ : $\Var(A_j), \Var(B_j) \le 1$.\\
	Then, there exist functions $A^{(1)} : \calZ^n \to \bbR^k$ and $B^{(1)} : \calZ^n \to \bbR^k$ such that statements 1-4 below hold.\\
	
	\noindent {\bf (Correlated Gaussian):} Let $A : \bbR^n \to \bbR^k$ and $B : \bbR^n \to \bbR^k$, such that, for any $j \in [k]$ : $\Var(A_j), \Var(B_j) \le 1$.\\
	Then, there exist functions $A^{(1)} : \bbR^n \to \bbR^k$ and $B^{(1)} : \bbR^n \to \bbR^k$ such that statements 1-4 below hold.

	\begin{enumerate}
	\item $A^{(1)}$ and $B^{(1)}$ have degree at most $d$.
	\item For any $i \in [k]$, it holds that $\Var(A^{(1)}_i) \le \Var(A_i) \le 1$ and $\Var(B^{(1)}_i) \le \Var(B_i) \le 1$.
	\item $\norm{2}{\calR(A^{(1)}) - A^{(1)}} ~\le~ \norm{2}{\calR(A) - A} + \delta$ and $\norm{2}{\calR(B^{(1)}) - B^{(1)}} ~\le~ \norm{2}{\calR(B) - B} + \delta$
	\item For every $i, j \in [k]$,
	\begin{align*}
		& \inabs{\inangle{A^{(1)}_i, B^{(1)}_j}_{\mu^{\otimes n}} - \inangle{A_i, B_j}_{\mu^{\otimes n}}} ~\le~ \frac{\delta}{\sqrt{k}} && \text{\bf (Correlated Discrete Hypercube)} \\
		& \inabs{\inangle{A^{(1)}_i, B^{(1)}_j}_{\calG_\rho^{\otimes n}} - \inangle{A_i, B_j}_{\calG_\rho^{\otimes n}}} ~\le~ \frac{\delta}{\sqrt{k}} && \text{\bf (Correlated Gaussian)}
	\end{align*}
%	$\inabs{\inangle{A^{(1)}_i, B^{(1)}_j}_{\mu^{\otimes n}} - \inangle{A_i, B_j}_{\mu^{\otimes n}}} ~\le~ \delta \qquad$ {\bf (Correlated Discrete Hypercube)} \\
%	$\inabs{\inangle{A^{(1)}_i, B^{(1)}_j}_{\calG_\rho^{\otimes n}} - \inangle{A_i, B_j}_{\calG_\rho^{\otimes n}}} ~\le~ \delta \qquad$ {\bf (Correlated Gaussian)}
	\end{enumerate}

	\noindent In particular, one may take $d = O\inparen{\frac{\sqrt{k}\log^2 (k/\delta)}{\delta (1 - \rho)}}$. %$d = O\inparen{\frac{\log (1/\delta)}{\log (1/\nu)}}$, where $\nu = 1 - C \frac{(1-\rho)\delta}{\log(1/\delta)}$ (for some absolute constant $C$).
\end{lem}

%% file: sec_non-multilinear.tex
\section{Reduction from General Polynomials to Multilinear Polynomials} \label{sec:non-multilinear}

In this section, we present a simple technique to convert low-degree (non-multilinear) polynomials into multilinear polynomials, without hurting the correlation, albeit increasing the number of variables slightly. This step is of a similar nature as \Cref{lem:smoothing_main}. In particular, note that the conditions 1, 2, 3, 5 in the following lemma are also present in \Cref{lem:smoothing_main}. This idea also appears in \cite{DMN_NoiseStabilityComputable, DMN_NIS_decidable}. Since the exact statement we desire is slightly different, we provide a proof for completeness in \Cref{apx:non-multilinear}.

\begin{lem}[Multi-linearization Lemma] \label{lem:multilin_main}
	Let $\rho \in [0,1]$, $\delta > 0$, $d, k \in \bbZ_{\ge 0}$ be any given constant parameters. There exists an explicit $t = t(k, d, \delta)$ such that the following holds:\\
	
	\noindent Let $A : \bbR^n \to \bbR^k$ and $B : \bbR^n \to \bbR^k$ be degree-$d$ polynomials, such that, for any $j \in [k]$ : $\Var(A_j), \Var(B_j) \le 1$. Then, there exist functions $A^{(1)} : \bbR^{nt} \to \bbR^k$ and $B^{(1)} : \bbR^{nt} \to \bbR^k$ such that the following hold
	
	\begin{enumerate}
		\item $A^{(1)}$ and $B^{(1)}$ are multilinear with degree $d$.
		\item For any $j \in [k]$, it holds that $\Var(A^{(1)}_j) \le \Var(A_j) \le 1$ and $\Var(B^{(1)}_j) \le \Var(B_j) \le 1$.
		\item $\norm{2}{\calR(A^{(1)}) - A^{(1)}} ~\le~ \norm{2}{\calR(A) - A} + \delta$ and $\norm{2}{\calR(B^{(1)}) - B^{(1)}} ~\le~ \norm{2}{\calR(B) - B} + \delta$
		\item For any $\ell \in [nt]$ and $j \in [k]$, it holds that $\Inf_\ell(A^{(1)}_j) \le \delta$ and $\Inf_\ell(B^{(1)}_j) \le \delta$.
		\item For every $i, j \in [k]$,
		\[ \inabs{\inangle{A^{(1)}_i, B^{(1)}_j}_{\calG_\rho^{\otimes nt}} - \inangle{A_i, B_j}_{\calG_\rho^{\otimes n}}} ~\le~ \frac{\delta}{\sqrt{k}} \]
	\end{enumerate}
	
	\noindent In particular, one may take $t = O\inparen{\frac{k^2 d^2}{\delta^2}}$. %$d = O\inparen{\frac{\log (1/\delta)}{\log (1/\nu)}}$, where $\nu = 1 - C \frac{(1-\rho)\delta}{\log(1/\delta)}$ (for some absolute constant $C$).
\end{lem}

%\noindent We defer the proof of \Cref{lem:multilin_main} to \Cref{apx:non-multilinear}.

%% file: sec_noise-stability.tex
\section{Non-Interactive Simulation from Correlated Gaussian Sources} \label{sec:noise-stability}

In this section, we show our main theorem regarding non-interactive simulation from Correlated Gaussian sources. That is, we show \Cref{th:NIS_Gaussian_src} (restated below as \Cref{thm:gaussian-non-int-sim}), and \Cref{th:noise-stability-informal} (which follows immediately as \Cref{cor:noise-stability}). %that the main theorem for proving approximate decidability of non-interactive simulation in the special case where the source distribution is correlated Gaussians.

\begin{theorem}\label{thm:gaussian-non-int-sim}
	Given parameters $k \ge 2$, $\rho \in [0,1]$ and $\eps > 0$, there exists an explicitly computable $n_0 = n_0(\rho, k, \eps)$ such that the following holds:

	\noindent For any $N$, and any $A : \bbR^N \to \Delta_k$ and $B : \bbR^N \to \Delta_k$, there exist functions $\wtilde{A} : \bbR^{n_0} \to \Delta_k$ and $\wtilde{B} : \bbR^{n_0} \to \Delta_k$ such that,
	\[\dTV\inparen{(A(\bX),B(\bY))_{(\bX,\bY)\sim \calG_\rho^{\otimes N}},\ (\wtilde{A}(\ba), \wtilde{B}(\bb))_{(\ba, \bb) \sim \calG_\rho^{\otimes n_0}}} \le \eps\;.\]
	Moreover, there exists $d_0 = d_0(\rho, k, \eps)$ for which there are degree-$d_0$ polynomials $A_0 : \bbR^{n_0} \to \bbR^k$ and $B_0 : \bbR^{n_0} \to \bbR^k$, such that, $\wtilde{A}(\ba) = \calR\inparen{A_0\inparen{\frac{\ba}{\|\ba\|_2}}}$ and $\wtilde{B}(\bb) = \calR\inparen{B_0\inparen{\frac{\bb}{\|\bb\|_2}}}$.\\
	In particular, one may take $n_0 = \exp \inparen{\poly\inparen{k, \frac{1}{\eps}, \frac{1}{1-\rho}}}$ and $d_0 = \poly\inparen{k, \frac{1}{\eps}, \frac{1}{1-\rho}}$.\footnote{the details of the exact value of $n_0$ could be inferred from combining the bounds across various lemmas used. We skip it for brevity, and instead stress on the qualitative nature of the bound.}\\%In particular, one may take $n_0 = \exp \inparen{\wtilde{O}\inparen{\frac{k^{4.5}}{\eps^2 (1-\rho)}}}$ and $d_0 = \wtilde{O}\inparen{\frac{k^{4.5}}{\eps^2 (1-\rho)}}$.\footnote{here, we $\wtilde{O}(\cdot)$ is suppressing poly-logarithmic factors in $k$, $1/\eps$ and $1/(1-\rho)$.}\\
	
	\noindent In fact, the transformation satisfies a stronger property that there exists an ``oblivious'' randomized transformation (with a shared random seed) to go from $A$ to $\wtilde{A}$ and from $B$ to $\wtilde{B}$, which works with probability at least $1-\eps$. Since the same transformation is applied on $A$ and $B$ simultaneously with the same random seed, if $A = B$, then the transformation gives $\wtilde{A}=\wtilde{B}$ as well.
	
%	\noindent In fact, the reduction satisfies two more stronger properties:
%	\begin{enumerate}
%		\item There exists an oblivious randomized transformation to go from $A$ to $A_0$ and from $B$ to $B_0$, which works with probability at least $1-\eps$. Since the same transformation is applied on $A$ and $B$ simultaneously with the same random seed, if $A = B$, then the transformation gives $A_0=B_0$ as well.
%		\item If the functions $A$, $B$ map to $\bbR^k$ instead of $\Delta_k$, while satisfying $\norm{2}{\calR(A) - A}, \norm{2}{\calR(B) - B} \le \delta$ then the same reduction has the property that $\norm{2}{\calR(A_0) - A_0} \le \sqrt{\delta}$. Similarly for $B_0$. This means that if the original strategies were close the $\Delta_k$, then the resulting (un-rounded) strategies will also be so. However, the reduction now succeeds with probability only $1 - \eps - O(\sqrt{\delta})$.
%	\end{enumerate}
\end{theorem}

\noindent Before proving the theorem, we remark that it immediately implies the desired statement needed to prove a dimension bound on $\eps$-approximate noise stable function (i.e. \Cref{th:noise-stability-informal}).

\begin{cor}\label{cor:noise-stability}
	Given parameters $k \ge 2$, $\rho \in [0,1]$ and $\eps > 0$, there exists an explicitly computable $n_0 = n_0(\rho, k, \eps)$ such that the following holds:
	
	Let $f : \bbR^N \to [k]$. Then, there exists a function $\wtilde{f} : \bbR^{n_0} \to \Delta_k$ such that
	\begin{enumerate}
		\item $\norm{1}{\Ex[f] - \Ex[\wtilde{f}]} \le \eps$.
		\item $\Ex\insquare{\inangle{\wtilde{f}, U_\rho \wtilde{f}}} \ge \Ex[\inangle{f, U_\rho f}] - \eps$.
	\end{enumerate}
	Moreover, there exists $d_0 = d_0(\rho, k, \eps)$ for which there is a degree-$d_0$ polynomial $g : \bbR^{n_0} \to \bbR^k $, such that, $\wtilde{f}(\ba) = \calR\inparen{g\inparen{\frac{\ba}{\|\ba\|_2}}}$. In particular, one may take $n_0 = \exp \inparen{\poly\inparen{k, \frac{1}{\eps}, \frac{1}{1-\rho}}}$ and $d_0 = \poly\inparen{k, \frac{1}{\eps}, \frac{1}{1-\rho}}$.%$n_0 = \exp \inparen{\wtilde{O}\inparen{\frac{k^{4.5}}{\eps^2 (1-\rho)}}}$ and $d_0 = \wtilde{O}\inparen{\frac{k^{4.5}}{\eps^2 (1-\rho)}}$.
\end{cor}
\begin{proof}
	We invoke \Cref{thm:gaussian-non-int-sim} with both $A$ and $B$ as $f$ and with parameter $\eps/2$, thereby obtaining functions $\wtilde{A}$, $\wtilde{B}$ which map $\bbR^{n_0} \to \Delta_k$. Note that since $A = B = f$, we also have $\wtilde{A} = \wtilde{B} = \wtilde{f}$. We get both our desired goals by observing that both $\norm{1}{\Ex[f] - \Ex[\wtilde{f}]}$ and $\inabs{\Ex\insquare{\inangle{\wtilde{f}, U_\rho \wtilde{f}}} - \Ex[\inangle{f, U_\rho f}]}$ are upper bounded by at most twice the total variation distance between the distributions $(A,B)_{\bX,\bY}$ and $(\wtilde{A}, \wtilde{B})_{\ba, \bb}$.
\end{proof}

%\noindent {\bf Proof overview}: %We in fact consider the more general setting, where were have two functions $A : \bbR^n \to \Delta_q$ and $B : \bbR^n \to \Delta_q$ corresponding to Alice and Bob respectively. We interpret the ``Stability'' as the following game: Alice and Bob get access to multivariate Gaussians, which are $\rho$-correlated in every coordinate. They use their respective functions $A$ and $B$, to output an element in $[q]$. ``Stability'' of the pair of functions $(A,B)$ is the probability that Alice and Bob output the same element in $[q]$. More generally, we could care about the entire joint distribution of $(A(X), B(Y))$ supported over $[q] \times [q]$.

The rest of this section is dedicated to proving \Cref{thm:gaussian-non-int-sim}. We first provide the main intuition behind the proof. Starting with functions $A : \bbR^N \to \Delta_k$ and $B: \bbR^N \to \Delta_k$, we would have liked to directly apply our dimension reduction. That would have entailed having $\wtilde{A}(\ba) = A(M\ba/\norm{2}{\ba})$ and $\wtilde{B}(\bb) = B(M\bb/\norm{2}{\bb})$, where $(\ba, \bb) \sim \calG_\rho^{\otimes n_0}$ and $M$ is a $N \times n_0$ matrix with entries sampled randomly from $\calN(0,1)$. This already gives us that the range of $\wtilde{A}$ and $\wtilde{B}$ is $\Delta_k$, since that was the range of $A$ and $B$ as well. Thus, if our dimension reduction were to approximately preserve correlations, i.e. $\inangle{A_i, B_j}_{\calG_\rho^{\otimes N}} \approx \inangle{A_i, B_j}_{\calG_\rho^{\otimes n_0}}$ for all $i, j \in [k]$ with high probability over $M$, we would have been done! However our actual dimension reduction (\Cref{thm:dim-reduction}) works only for low-degree multilinear polynomials $A$ and $B$. To get around this, we first apply the Smoothing (\Cref{lem:smoothing_main}) and Multilinearization (\Cref{lem:multilin_main}) transformations that make $A$ and $B$ both low-degree and multilinear, and then subsequently apply our dimension reduction (\Cref{thm:dim-reduction}). Unfortunately, this creates a new problem, that after these transformations, the range is no longer $\Delta_k$, but is instead $\bbR^k$. Nevertheless, we do have that these transformations ensure that the functions still output something ``close'' to the simplex $\Delta_k$. This allows us to use the standard rounding operation to get the range as $\Delta_k$ again (using \Cref{lem:close-strategies}). An overview of the transformations done is presented in \Cref{fig:Gaussian_NIS}.

\begin{proof}[Proof of \Cref{thm:gaussian-non-int-sim}]
For any $i, j \in [k]$, we focus on the quantity $\inangle{A_i, B_j}_{\calG_\rho^{\otimes n}}$ which is the probability of the event that [{\em Alice outputs $i$ and Bob outputs $j$}]. Through several steps, we modify Alice's and Bob's strategies, while approximately preserving this quantity for every $i, j$. Note that if we preserve the probability that Alice outputs $i$ and Bob outputs $j$ for every $i, j$ upto an additive $\eps/k^2$, this implies that we would preserve the joint distribution of Alice and Bob's outputs up to $\ell_1$-distance of $\eps$.% The stronger properties of our reduction will be clear from our transformations.

\input{fig_GaussianNIS}

We transform $A$ and $B$ through each of the following steps, as illustrated in \Cref{fig:Gaussian_NIS}. At each step, we approximately preserve the correlation $\inangle{A_i, B_j}$ for every $i, j \in [k]$. Additionally, in each step $\norm{2}{\calR(A) - A}$ and $\norm{2}{\calR(B) - B}$ also doesn't increase significantly. Note that, to begin with the range of $A$ and $B$ is $\Delta_k$ and hence $\norm{2}{\calR(A) - A} = \norm{2}{\calR(B) - B} = 0$.

\begin{enumerate}
	\item {\bf Smoothing}: We apply \Cref{lem:smoothing_main} with parameter $\delta$, setting $d = d(\rho, k, \delta)$ as required, on $A$ and $B$ to get the smoothened versions $A^{(1)} : \bbR^N \to \bbR^k$ and $B^{(1)} : \bbR^N \to \bbR^k$. This guarantees that $A^{(1)}$ and $B^{(1)}$ have degree at most $d$. Moreover, we have that for every $i, j \in [k]$,
	\begin{equation}
	\inabs{\inangle{A^{(1)}_i, B^{(1)}_j}_{\calG_\rho^{\otimes N}} - \inangle{A_i, B_j}_{\calG_\rho^{\otimes N}}} \le \delta \label{eqn:GNIS-1}
	\end{equation}
	Additionally,
	\[ \norm{2}{\calR(A^{(1)}) - A^{(1)}} ~\le~ \norm{2}{\calR(A) - A} + \delta ~\le~ \delta\;.\]
	Similarly, we also have that,
	\[\norm{2}{\calR(B^{(1)}) - B^{(1)}} ~\le~ \delta\]
	
	\item {\bf Multilinearization}: We apply \Cref{lem:multilin_main} with parameter $\delta$, setting $t = t(d, k, \delta)$ as required, on $A^{(1)}$ and $B^{(1)}$ to get the multilinearized versions $A^{(2)} : \bbR^{Nt} \to \bbR^k$ and $B^{(2)} : \bbR^{Nt} \to \bbR^k$. This guarantees that both $A^{(2)}$ and $B^{(2)}$ are multilinear and have degree at most $d$, albeit over a slightly larger number of variables. We get,
	\begin{equation}
	\inabs{\inangle{A^{(2)}_i, B^{(2)}_j}_{\calG_\rho^{\otimes Nt}} - \inangle{A^{(1)}_i, B^{(1)}_j}_{\calG_\rho^{\otimes N}}} \le \delta \label{eqn:GNIS-2}
	\end{equation}
	%\[ \norm{2}{A^{(2)} - A^{(1)}} ~\le~ \eps \qquad \sAND \qquad \norm{2}{B^{(2)} - B^{(1)}} ~\le~ \eps\]
	%Combining this with \Cref{lem:close-to-simplex}, we have that,
	Additionally,
	\[ \norm{2}{\calR(A^{(2)}) - A^{(2)}} ~\le~ \norm{2}{\calR(A^{(1)}) - A^{(1)}} + \delta ~\le~ 2\delta\]
	Similarly, we also have that,
	\[ \norm{2}{\calR(B^{(2)}) - B^{(2)}} ~\le~ 2 \delta\]
	
	\item {\bf Dimension reduction}: We apply \Cref{thm:dim-reduction} with parameter $\delta/k^2$, setting $D = D(d, \rho, \delta/k^2)$ as required, on individual coordinates of $A^{(2)}$ and $B^{(2)}$ to obtain functions $A^{(3)} : \bbR^D \to \bbR^k$ and $B^{(3)} : \bbR^D \to \bbR^k$. Taking a union bound, we have that with probability at least $1-\delta$, it holds for every $i, j \in [k]$ that, %Moreover, we have that, %Although \Cref{thm:dim-reduction} is stated only for polynomials mapping to $\bbR$, we can apply the same transformation on the individual coordinates of $A^{(2)}$ and $B^{(2)}$.
	\begin{equation}
	\inabs{\inangle{A^{(3)}_i, B^{(3)}_j}_{\calG_\rho^{\otimes D}} - \inangle{A^{(2)}_i, B^{(2)}_j}_{\calG_\rho^{\otimes Nt}}} \le \delta \label{eqn:GNIS-3}
	\end{equation}
	From \Cref{prop:dim-red-close-to-simplex}, we have that with probability $1-4\delta$,
	\[ \norm{2}{\calR(A^{(3)}) - A^{(3)}} ~\le~ \sqrt{\norm{2}{\calR(A^{(2)}) - A^{(2)}}} ~\le~ \sqrt{2\delta}\]
	\[\norm{2}{\calR(B^{(3)}) - B^{(3)}} ~\le~ \sqrt{\norm{2}{\calR(B^{(1)}) - B^{(1)}}} ~\le~ \sqrt{2\delta}\]
	
	Note that this is the only randomized procedure in the entire transformation. This reduction succeeds in obtaining the three constraints above with probability at least $1 - 5\delta$.
	
	\item {\bf Rounding to $\Delta_k$}: We obtain our final functions as $\wtilde{A} = \calR(A^{(3)})$ and $\wtilde{B} = \calR(B^{(3)})$. %Here we use \Cref{lem:close-strategies} to conclude that the rounding operation doesn't change the correlation significantly.
	Note that $\norm{2}{\calR(A^{(3)}) - A^{(3)}}, \norm{2}{\calR(B^{(3)}) - B^{(3)}} \le \sqrt{2\delta}$, and hence we have for any $i, j \in [k]$ that \[\norm{2}{\wtilde{A}_i - A_i^{(3)}} \le \sqrt{2 \delta} \qquad \sAND \qquad \norm{2}{\wtilde{B}_j - B_j^{(3)}} \le \sqrt{2\delta}\;.\]
	Hence we can invoke \Cref{lem:close-strategies}, to conclude that for $\wtilde{A} = \calR(A^{(3)})$ and $\wtilde{B} = \calR(B^{(3)})$,
	\begin{equation}
	\inabs{\inangle{\wtilde{A}_i, \wtilde{B}_j}_{\calG_\rho^{\otimes D}} - \inangle{A^{(3)}_i, B^{(3)}_j}_{\calG_\rho^{\otimes D}}} ~\le~ 2 \sqrt{\delta}. \label{eqn:GNIS-4}
	\end{equation}
\end{enumerate}

\noindent Finally, we choose $n_0 = D$ and $d_0 = d$ as obtained above. Note that we started with functions $A: \bbR^N \to \Delta_k$ and $B : \bbR^N \to \Delta_k$ and we ended with functions $\wtilde{A} : \bbR^{n_0} \to \Delta_k$ and $\wtilde{B} : \bbR^{n_0} \to \Delta_k$ such that for every $i, j \in [k]$, we have by combining \Cref{eqn:GNIS-1,eqn:GNIS-2,eqn:GNIS-3,eqn:GNIS-4} that,
\[ \inabs{\inangle{\wtilde{A}_i, \wtilde{B}_j}_{\calG_\rho^{\otimes D}} - \inangle{A_i, B_j}_{\calG_\rho^{\otimes N}}} ~\le~ O(\sqrt{\delta})\]

\noindent Thus, more strongly, if we instantiate $\delta = O(\eps^2/k^4)$, then we get that our entire transformation succeeds with probability $1-\eps$ in obtaining $\wtilde{A}$ and $\wtilde{B}$ such that,
\[\dTV((A(\bX),B(\bY))_{\bX,\bY}, (\wtilde{A}(\ba), \wtilde{B}(\bb))_{\ba, \bb}) \le \eps\;,\]
where recall that $(\bX, \bY) \sim \calG_\rho^{\otimes N}$ and $(\ba, \bb) \sim \calG_\rho^{\otimes n_0}$. It is easy to see that the parameters work out to be
\[ d_0 = d = \wtilde{O}\inparen{\frac{k^{4.5}}{\eps^2 (1-\rho)}}\;, \]
\[ n_0 = D = \frac{d^{O(d)}}{\delta^4} = \exp\inparen{\wtilde{O}\inparen{\frac{k^{4.5}}{\eps^2 (1-\rho)}}}\;.\]
\end{proof}

%% file: fig_GaussianNIS.tex
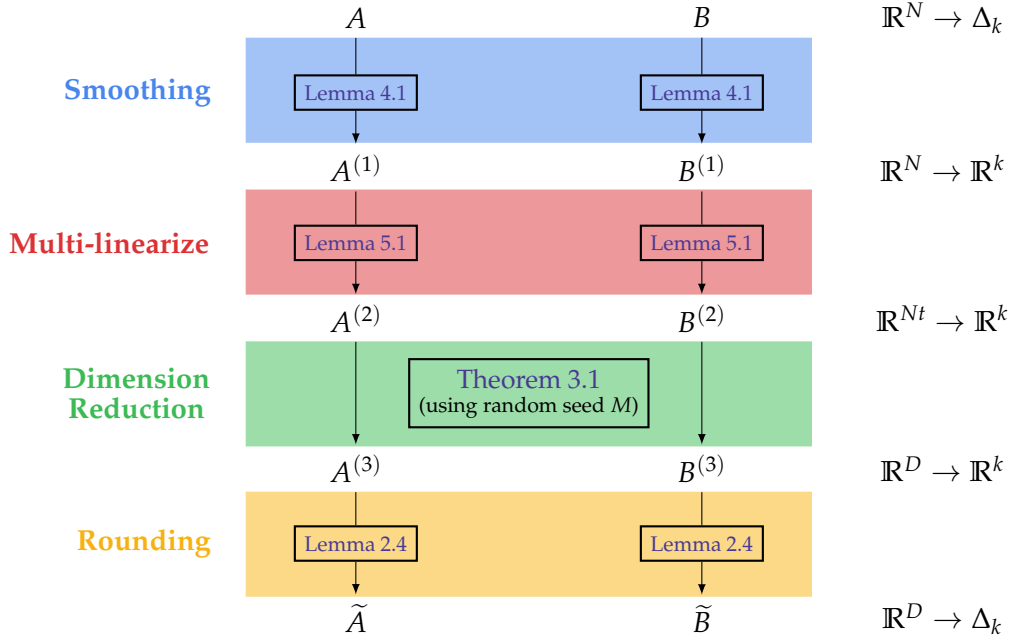
\begin{figure}
\begin{center}
\begin{tikzpicture}[scale=1, transform shape]

\def\ColSmooth{Gblue}
\def\ColMultiLin{Gred}
\def\ColDimRed{Ggreen}
\def\ColRound{Gyellow}
\def\Brightness{50}

\def \stepname{-1.8}
\def \alice{0}
\def \bob{4.6}
\def \midd{2.3}
\def \descrip{7.8}
\def \colorwidth{7.5cm}

\def \ycorr{0}
\node (A) at (\alice,\ycorr) {$A$};
\node (B) at (\bob, \ycorr) {$B$};
\node at(\descrip, \ycorr) {$\bbR^N \to \Delta_k$};

\def \ycorr{-1}
\node[box, below, minimum width=\colorwidth, minimum height=1.37cm, fill=\ColSmooth!\Brightness, draw=\ColSmooth!\Brightness] at (\midd, \ycorr+0.73) {};
\node[left] at (\stepname,\ycorr) {\textcolor{\ColSmooth}{\bf Smoothing}};
\node[box] (SmoothA) at (\alice,\ycorr) {\scriptsize  \Cref{lem:smoothing_main}} edge[-] (A);
\node[box] (SmoothB) at (\bob,\ycorr) {\scriptsize  \Cref{lem:smoothing_main}} edge[-] (B);

\def \ycorr{-2}
\node (A1) at (\alice,\ycorr) {$A^{(1)}$} edge[latex-] (SmoothA);
\node (B1) at (\bob,\ycorr) {$B^{(1)}$} edge[latex-] (SmoothB);
\node at (\descrip,\ycorr) {$\bbR^N \to \bbR^k$};

\def \ycorr{-3}
\node[box, below, minimum width=\colorwidth, minimum height=1.37cm, fill=\ColMultiLin!\Brightness, draw=\ColMultiLin!\Brightness] at (\midd, \ycorr+0.71) {};
\node[left] at (\stepname,\ycorr) {\textcolor{\ColMultiLin}{\bf Multi-linearize}};
\node[box] (MultiLinA) at (\alice,\ycorr) {\scriptsize \Cref{lem:multilin_main}} edge[-] (A1);
\node[box] (MultiLinB) at (\bob,\ycorr) {\scriptsize \Cref{lem:multilin_main}} edge[-] (B1);

\def \ycorr{-4}
\node (A2) at (\alice,\ycorr) {$A^{(2)}$} edge[latex-] (MultiLinA);
\node (B2) at (\bob,\ycorr) {$B^{(2)}$} edge[latex-] (MultiLinB);
\node at (\descrip,\ycorr) {$\bbR^{Nt} \to \bbR^k$};

\def \ycorr{-5}
\node[box, below, minimum width=\colorwidth, minimum height=1.37cm, fill=\ColDimRed!\Brightness, draw=\ColDimRed!\Brightness] at (\midd, \ycorr+0.69) {};
\node[left] at (\stepname,\ycorr) {\shortstack{\textcolor{\ColDimRed}{\bf Dimension} \\ \textcolor{\ColDimRed}{\bf Reduction}}};
\node[box] (DimRed) at (\midd,\ycorr) {\shortstack{\small \Cref{thm:dim-reduction} \\ \scriptsize (using random seed $M$)}};
%\node[box] (DimRedA) at (\alice,\ycorr) {\scriptsize \Cref{thm:dim-reduction}} edge[-] (A2);
%\node[box] (DimRedB) at (\bob,\ycorr) {\scriptsize \Cref{thm:dim-reduction}} edge[-] (B2);

\def \ycorr{-6}
\node (A3) at (\alice,\ycorr) {$A^{(3)}$} edge[latex-] (A2);
\node (B3) at (\bob,\ycorr) {$B^{(3)}$} edge[latex-] (B2);
\node at (\descrip,\ycorr) {$\bbR^D \to \bbR^k$};

\def \ycorr{-7}
\node[box, below, minimum width=\colorwidth, minimum height=1.37cm, fill=\ColRound!\Brightness, draw=\ColRound!\Brightness] at (\midd, \ycorr+0.69) {};
\node[left] at (\stepname,\ycorr) {\textcolor{\ColRound}{\bf Rounding}};
\node[box] (RoundA) at (\alice,\ycorr) {\scriptsize \Cref{lem:close-strategies}} edge[-] (A3);
\node[box] (RoundB) at (\bob,\ycorr) {\scriptsize \Cref{lem:close-strategies}} edge[-] (B3);

\def \ycorr{-8}
\node (A4) at (\alice,\ycorr) {$\wtilde{A}$} edge[latex-] (RoundA);
\node (B4) at (\bob,\ycorr) {$\wtilde{B}$} edge[latex-] (RoundB);
\node at (\descrip,\ycorr) {$\bbR^D \to \Delta_k$};

\end{tikzpicture}
\caption{Transformations for Non-interactive simulation from Correlated Gaussian Sources}
\label{fig:Gaussian_NIS}
\end{center}
\end{figure}

%% file: sec_non-int-sim.tex
\section{Non-Interactive Simulation from Arbitrary Discrete Sources} \label{sec:non-int-sim}

In this section we prove our main theorem regarding non-interactive simulation from arbitrary discrete sources. That is, we prove \Cref{th:non-int-sim} (restated below as \Cref{thm:non-int-sim}). %that the problem of non-interactive simulation is approximately decidable. We state the main theorem, which immediately leads to the decidability result.

\begin{theorem}\label{thm:non-int-sim}
	Let $(\calZ \times \calZ, \mu)$ be a joint probability space. Given parameters $k \ge 2$ and $\eps > 0$, there exists an explicitly computable $n_0 = n_0(\mu, k, \eps)$ such that the following holds:
	
	Let $A : \calZ^N \to \Delta_k$ and $B : \calZ^N \to \Delta_k$. Then there exist functions $\wtilde{A} : \calZ^{n_0} \to \Delta_k$ and $\wtilde{B} : \calZ^{n_0} \to \Delta_k$ such that,
	\[\dTV\inparen{(A(\bx),B(\by))_{(\bx,\by)\sim \mu^{\otimes N}}, \ (\wtilde{A}(\ba), \wtilde{B}(\bb))_{\ba, \bb \sim \mu^{\otimes n_0}}} \le \eps\;.\]
	In particular, $n_0$ is an explicit function upper bounded by $\exp\inparen{\poly\inparen{k, \frac{1}{\eps}, \frac{1}{1-\rho}, \log\inparen{\frac{1}{\alpha}}}}$\footnote{the details of the exact value of $n_0$ could be inferred from combining the bounds across various lemmas used. We skip it for brevity, and instead stress on the qualitative nature of the bound.}, where $\alpha = \alpha(\mu)$ is the smallest atom in $\mu$ and $\rho = \rho(\mu)$ is the maximal correlation of $\mu$.
\end{theorem}

\noindent Note, that this theorem is, in a way, a generalization of \Cref{thm:gaussian-non-int-sim}, where $\calZ$ was $\bbR$ and the distribution $\mu$ was $\calG_\rho$. On the other hand, this theorem is only for the case when $\calZ$ is a finite set, so in this sense it is incomparable to \Cref{thm:gaussian-non-int-sim}.\\

\noindent {\bf Proof Overview}: The proof works by a reduction to \Cref{thm:gaussian-non-int-sim}. This reduction is done along the same framework as introduced in \cite{GKS_NIS_decidable}. We first apply a Smoothing operation (\Cref{lem:smoothing_main}) similar to the Gaussian case, to make the functions $A$ and $B$ have {\em low-degree}. Next, we will apply a Regularity Lemma (\Cref{lem:joint_regularity}), to identify a constant sized subset of coordinates, such that for a random fixing of these coordinates, the restricted function is low influential on the remaining coordinates. This allows us to apply the invariance principle (\Cref{lem:our_invariance}) to replace the coordinates of Alice and Bob on the remaining coordinates by $\rho$-correlated Gaussians. We now use \Cref{thm:gaussian-non-int-sim} to reduce the number of coordinates of $\rho$-correlated Gaussians needed. Finally, we wish to get the strategies to use samples from $\mu$ instead of $\calG_{\rho}$, by simulating the correlated multivariate Gaussians using a bounded number of samples of $\mu$. However, after the transformation of \Cref{thm:gaussian-non-int-sim}, the resulting function might be none of low-degree, multilinear or low-influential. To get around this we apply Smoothing (\Cref{lem:smoothing_main}) to make it low-degree and Multilinearization (\Cref{lem:multilin_main}) to make it multilinear and low-influential, after which we can apply the invariance principle (\Cref{lem:our_invariance}). An overview of the transformations done is presented in \Cref{fig:General_NIS}.\\

\begin{proofof}{\Cref{thm:non-int-sim}}
As in the proof of \Cref{thm:gaussian-non-int-sim}, for any $i, j \in [k]$, we focus on the quantity $\inangle{A_i, B_j}_{\mu^{\otimes n}}$ which is the probability of the event that [{\em Alice outputs $i$ and Bob outputs $j$}]. Through the several steps we modify Alice's and Bob's strategy, while preserving this quantity approximately for every $i, j$. If we preserve the probability that Alice outputs $i$ and Bob outputs $j$ for every $i, j$ upto an additive $\eps/k^2$, it implies that we preserve the joint distribution of Alice and Bob's outputs up to an $\ell_1$-distance of $\eps$.

\input{fig_dependency}

\paragraph{Choice of parameters \& bound on $n_0$.} As described in the overview, we are going to invoke several of the lemmas we have developed in our proof. We now describe the choice of parameters for which we invoke these lemmas, thereby obtaining our final explicit bound on $n_0$. We recommend consulting \Cref{fig:dependency} to follow the exact chain of dependencies among the parameters (the dependencies on $\mu$ and $k$ are suppressed in the figure for clarity).

Let $\delta$ be a running parameter, that we finalize at the end (in terms of $\mu$, $k$ and $\eps$). We wish to invoke the Smoothing operation over $\calZ^N$ (\Cref{lem:smoothing_main}) with parameter $\delta$. This dictates a value of $d = d(\rho, k, \delta)$, which is the degree of the polynomials obtained after smoothing. We will invoke the Invariance Principle (\Cref{lem:our_invariance}) with parameters $\delta$ and $d$ as obtained just now. This dictates a value of $\tau = \tau(\mu, k, d, \delta)$, which is the bound on the influence needed in order to apply the invariance principle. We will invoke the Regularity Lemma (\Cref{lem:joint_regularity}) with parameters $d$ and $\tau$ as obtained above. This dictates a value of $h = h(\mu,k,d,\tau)$, which is the bound on number of {\em head} coordinates that need to be fixed to get all influences less than $\tau$ on the remaining coordinates, for a random restriction of the head coordinates.

We will apply \Cref{thm:gaussian-non-int-sim} with the error parameter $\eps$ as $\delta$. This dictates a value of $D = n_0(k,\rho,\delta)$, which the number of coordinates of correlated Gaussians needed after dimension reduction. In order to go back from correlated Gaussians to samples from $(\calZ \times \calZ; \mu)$, we will again apply the Smoothing operation, this time over Gaussian space (\Cref{lem:smoothing_main}) with parameter $\delta$, which again dictates a value of $d = d(\rho, k, \delta)$. We will again apply the Invariance principle (\Cref{lem:our_invariance}) with parameters $\delta$ and $d$ as obtained just now. This dictates a value of $\tau = \tau(\mu, k, d, \delta)$, which is the bound on the influence needed in order to apply the invariance principle. In order to make the influences small, we will apply the Multilinearization operation (\Cref{lem:multilin_main}) with parameters $\delta \gets \tau$ and $d$ as obtained above. This dictates a value of $t = t(k,d,\delta)$, which is the blow up incurred while making the polynomials multilinear and to make them have all influences smaller than $\delta$.

Finally $n_0 = h + D \cdot t$. Recall that $h$ is the number of {\em head} coordinates in the Regularity Lemma. $D$ is the number of coordinates obtained after applying Gaussian NIS (\Cref{thm:gaussian-non-int-sim}). Finally $t$ is the blow up incurred while going back from correlated Gaussian space to $(\calZ \times \calZ; \mu)$. We will eventually choose $\delta = \eps^2/k^4$. It can be inferred by going through all the parameters carefully that $n_0(\mu, k, \eps)$ is an explicit function that can be upper bounded as $\exp \inparen{\poly \inparen{k, \frac{1}{\eps}, \frac{1}{1-\rho}, \log \inparen{\frac{1}{\alpha}}}}$. We skip this the details of this calculation for brevity.

\paragraph{Analysis of the transformations.} We now turn to the analysis of the above transformation. We wish to show that $\inangle{A_i, B_j}_{\mu^{\otimes N}} \approx \inangle{\wtilde{A}_i, \wtilde{B}_j}_{\calG_\rho^{\otimes n_0}}$ for every $i, j \in [k]$. We transform $A$ and $B$ through each of the following steps, as illustrated in \Cref{fig:General_NIS}. At each step, we approximately preserve the correlation $\inangle{A_i, B_j}$ for every $i, j \in [k]$.

\input{fig_GeneralNIS}

\begin{enumerate}
	\item {\bf Smoothing (over hypercube)}: We apply \Cref{lem:smoothing_main} on $A$ and $B$ to get the low-degree versions $A^{(1)} : \calZ^N \to \bbR^k$ and $B^{(1)} : \calZ^N \to \bbR^k$. This guarantees that $A^{(1)}$ and $B^{(1)}$ have degree at most $d$. Moreover, we have that for every $i, j \in [k]$,
	\[ \inabs{\inangle{A^{(1)}_i, B^{(1)}_j}_{\mu^{\otimes N}} - \inangle{A_i, B_j}_{\mu^{\otimes N}}} \le \delta \]
	Additionally,
	\[ \norm{2}{\calR(A^{(1)}) - A^{(1)}} ~\le~ \norm{2}{\calR(A) - A} + \delta ~\le~ \delta\;.\]
	Similarly, we also have that,
	\[\norm{2}{\calR(B^{(1)}) - B^{(1)}} ~\le~ \delta\]
	Using \Cref{lem:close-strategies}, we can conclude that for every $i, j \in [k]$,
	\begin{equation}
	\inabs{\inangle{\calR_i(A^{(1)}), \calR_j(B^{(1)})}_{\mu^{\otimes N}} - \inangle{A_i, B_j}_{\mu^{\otimes N}}} ~\le~ 3\delta \label{eqn:NIS-1}
	\end{equation}
	%\[\inabs{\inangle{A_i, B_j}_{\mu^{\otimes N}} - \inangle{A^{(1)}_i, B^{(1)}_j}_{\mu^{\otimes N}}} \le \eps\]
	%\[ \norm{2}{\calR(A^{(1)}) - A^{(1)}} \le \eps \qquad \sAND \qquad \norm{2}{\calR(B^{(1)}) - B^{(1)}} \le \eps\]
	
	\item {\bf Regularity Lemma}:  We apply \Cref{lem:joint_regularity} to identify a subset $H \subseteq [n]$ with $|H| = h$, such that, for a random restriction $(\bx_H, \by_H) \sim \mu^{\otimes h}$, it holds with probability at least $1-\tau$, that the restricted functions $(A^{(1)}_i)^{\bx_H} : \calZ^{N-h} \to \bbR^k$ and $(B^{(1)}_i)^{\by_H} : \calZ^{N-h} \to \bbR^k$ have all individual influences smaller than $\tau$, for any $i \in [k]$. We call a restriction $(\bx_H, \by_H)$ as ``{\em good}'' in this case, and ``{\em bad}'' otherwise. Note that any such restriction (``good'' or ``bad'') has degree at most $d$.
	%Note that this step does not change the functions. We only gain structural knowledge about the function.
	
	For the rest of the steps, we will focus on a {\em good} $(\bx_H, \by_H)$. For convenience, define $N_1 = N-h$.
	
	\item {\bf Invariance Principle (from $\calZ$ to $\bbR$)}: For a {\em good} $(\bx_H, \by_H)$, we apply \Cref{lem:our_invariance} on $(A^{(1)})^{\bx_H}$ and $(B^{(1)})^{\by_H}$ to get functions $(A^{(2)})^{\bx_H} : \bbR^{N_2} \to \Delta_k$ and $(B^{(2)})^{\by_H} : \bbR^{N_2} \to \Delta_k$ (where $N_2 = N_1 \cdot (q-1)$), such that for every $i, j \in [k]$,
	%\[\inabs{\Ex\limits_{(\bx_H, \by_H) \sim \mu^{\otimes h}} \inparen{\inangle{(A^{(2)}_i)_{\bx_H}, (B^{(2)}_j)_{\by_H}}_{\calG_\rho^{\otimes (N - h)(q-1)}} - \inangle{(A^{(1)}_i)_{\bx_H}, (B^{(1)}_j)_{\by_H}}_{\mu^{\otimes (N - h)}}}} \le \eps\]
	\begin{equation}
		\inabs{\inangle{(A^{(2)}_i)^{\bx_H}, (B^{(2)}_j)^{\by_H}}_{\calG_\rho^{\otimes N_2}} - \inangle{\calR_i\inparen{(A^{(1)})^{\bx_H}}, \calR_j\inparen{(B^{(1)})^{\by_H}}}_{\mu^{\otimes N_1}}} ~\le~ \delta \label{eqn:NIS-2}
	\end{equation}
	Note that strictly speaking \Cref{lem:our_invariance}, as stated, gives us functions mapping to $\bbR^k$ and not $\Delta_k$. However, we consider their rounded versions, which exactly gives us the statement above.
	
	\item {\bf Dimension Reduction}: We apply \Cref{thm:dim-reduction} on $(A^{(2)})^{\bx_H}$ and $(B^{(2)})^{\by_H}$, to get functions $(A^{(3)})^{\bx_H} : \bbR^D \to \Delta_k$ and $(B^{(3)})^{\by_H} : \bbR^D \to \Delta_k$, such that, for every $i, j \in [k]$,
	\begin{equation}
		\inabs{\inangle{(A^{(3)}_i)^{\bx_H}, (B^{(3)}_j)^{\by_H}}_{\calG_\rho^{\otimes D}} - \inangle{(A^{(2)}_i)^{\bx_H}, (B^{(2)}_j)^{\by_H}}_{\calG_\rho^{\otimes N_2}}} ~\le~ \delta \label{eqn:NIS-3}
	\end{equation}
	Note that this is the only randomized step in the entire transformation. This reduction succeeds with probability at least $1 - 4\delta$ for every {\em good} $(\bx_H, \by_H)$. For a fixed choice of the random seed, we call a {\em good} $(\bx_H, \by_H)$ as ``{\em lucky}''\footnote{rather poor choice of terminology, given that most $(\bx_H, \by_H)$ end up being {\em lucky}!} if the reduction succeeds for that $(\bx_H, \by_H)$. In expectation over the choice of random seeds, a $(1-4\delta)$ fraction of the {\em good} $(\bx_H, \by_H)$ are going to be {\em lucky}. We can hence choose a choice of random seed for which indeed a $(1-4\delta)$ fraction of the {\em good} $(\bx_H, \by_H)$ are {\em lucky}.
	
	On the other hand, Regularity Lemma ensures that with probability $1-\delta$, a sampled $(\bx_H, \by_H)$ will be {\em good}. This gives that for the said choice of random seed in \Cref{thm:gaussian-non-int-sim}, with probability at least $1-5\delta$, the sampled $(\bx_H, \by_H)$ is both {\em good} and {\em lucky}. For the rest of the steps, we will focus on a {\em good} and {\em lucky} $(\bx_H, \by_H)$.
	
	\item {\bf Smoothing (over Gaussian space)}: We again apply \Cref{lem:smoothing_main} on $(A^{(3)})^{\bx_H}$ and $(B^{(3)})^{\by_H}$ to get the low-degree versions $(A^{(4)})^{\bx_H} : \bbR^D \to \bbR^k$ and $(B^{(4)})^{\by_H} : \bbR^D \to \bbR^k$. This guarantees that $(A^{(4)})^{\bx_H}$ and $(B^{(4)})^{\by_H}$ have degree at most $d$. Moreover, we have that for every $i, j \in [k]$,
	\begin{equation}
	\inabs{\inangle{(A^{(4)}_i)^{\bx_H}, (B^{(4)}_j)^{\by_H}}_{\calG_\rho^{\otimes D}} - \inangle{(A^{(3)}_i)^{\bx_H}, (B^{(3)}_j)^{\by_H}}_{\calG_\rho^{\otimes D}}} ~\le~ \delta \label{eqn:NIS-4}
	\end{equation}
	Additionally,
	\[ \norm{2}{\calR(A^{(4)}) - A^{(4)}} ~\le~ \delta \qquad \sAND \qquad \norm{2}{\calR(B^{(4)}) - B^{(4)}} ~\le~ \delta \]
%	Using \Cref{lem:close-strategies}, we can conclude that for every $i, j \in [k]$,
%	\begin{equation}
%	\inabs{\inangle{\calR_i\inparen{(A^{(4)})^{\bx_H}}, \calR_j\inparen{(B^{(4)})^{\by_H}}}_{\calG_\rho^{\otimes D}} - \inangle{(A^{(3)}_i)^{\bx_H}, (B^{(3)}_j)^{\by_H}}_{\calG_\rho^{\otimes D}}} ~\le~ \delta 
%	\end{equation}
	
	\item {\bf Multilinearization}: We apply \Cref{lem:multilin_main} on $(A^{(4)})^{\bx_H}$ and $(B^{(4)})^{\by_H}$ to get the multilinearized and low-influential versions $(A^{(5)})^{\bx_H} : \bbR^{D_1} \to \bbR^k$ and $(B^{(5)})^{\by_H} : \bbR^{D_1} \to \bbR^k$ (where $D_1 = Dt$). Thus, we have for every $i, j \in [k]$,
	\[ \inabs{\inangle{(A^{(5)}_i)^{\bx_H}, (B^{(5)}_j)^{\by_H}}_{\calG_\rho^{\otimes D_1}} - \inangle{(A^{(4)}_i)^{\bx_H}, (B^{(4)}_j)^{\by_H}}_{\calG_\rho^{\otimes D}}} ~\le~ \delta \]
	Additionally, combining with \Cref{lem:close-to-simplex},
	\[ \norm{2}{\calR(A^{(5)}) - A^{(5)}} ~\le~ \norm{2}{\calR(A^{(4)}) - A^{(4)}} + \norm{2}{A^{(5)} - A^{(4)}} ~\le~ 2\delta \]
	Similarly,
	\[ \norm{2}{\calR(B^{(5)}) - B^{(5)}} ~\le~ 2\delta \]
	Combining all this with \Cref{lem:close-strategies} we get that,
	\begin{equation}
	\inabs{\inangle{\calR_i\inparen{(A^{(5)})^{\bx_H}}, \calR_j\inparen{(B^{(5)})^{\by_H}}}_{\calG_\rho^{\otimes D_1}} - \inangle{(A^{(4)}_i)^{\bx_H}, (B^{(4)}_j)^{\by_H}}_{\calG_\rho^{\otimes D}}} ~\le~ 5\delta \label{eqn:NIS-5}
	\end{equation}
	Additionally, note that we also have that $\Inf_\ell((A^{(5)}_i))^{\bx_H} \le \tau$ and $\Inf_\ell((B^{(5)}_i))^{\by_H} \le \tau$ for all $i \in [k]$ and $\ell \in [n]$. This is helpful for us to apply the invariance principle next.
	
	\item {\bf Invariance Principle (from $\bbR$ to $\calZ$)}: We apply \Cref{lem:our_invariance} on $(A^{(5)})^{\bx_H}$ and $(B^{(5)})^{\by_H}$ to get functions $(A^{(6)})^{\bx_H} : \calZ^{D_1} \to \Delta_k$ and $(B^{(6)})^{\by_H} : \calZ^{D_1} \to \Delta_k$, such that for every $i, j \in [k]$,
	%\[\inabs{\Ex\limits_{(\bx_H, \by_H) \sim \mu^{\otimes h}} \inparen{\inangle{(A^{(2)}_i)_{\bx_H}, (B^{(2)}_j)_{\by_H}}_{\calG_\rho^{\otimes (N - h)(q-1)}} - \inangle{(A^{(1)}_i)_{\bx_H}, (B^{(1)}_j)_{\by_H}}_{\mu^{\otimes (N - h)}}}} \le \eps\]
	\begin{equation}
	\inabs{\inangle{(A^{(6)}_i)^{\bx_H}, (B^{(6)}_j)^{\by_H}}_{\mu^{\otimes D_1}} - \inangle{\calR_i\inparen{(A^{(5)})^{\bx_H}}, \calR_j\inparen{(B^{(5)})^{\by_H}}}_{\calG_\rho^{\otimes D_1}}} ~\le~ \delta \label{eqn:NIS-6}
	\end{equation}
	Note again that strictly speaking \Cref{lem:our_invariance} gives us functions mapping to $\bbR^k$ and not $\Delta_k$. However, we consider their rounded versions, which exactly gives us the statement above.
	
%	From \Cref{lem:multilin_main}, we also have that for any fixing of $\bx_H$, $(A^{(5)})_{\bx_H}$ has influences at most $\tau$. Thus, using \Cref{lem:our_invariance}, we can substitute these coordinates by $\calX_1$ and $\calY_1$ respectively, to get,
%	\[\inabs{\Ex\limits_{(\bx_H, \by_H) \sim \mu^{\otimes h}} \inparen{\inangle{(A^{(6)}_i)_{\bx_H}, (B^{(6)}_j)_{\by_H}}_{\mu^{\otimes O(Dd^2/\eps^2)}} - \inangle{(A^{(5)}_i)_{\bx_H}, (B^{(5)}_j)_{\by_H}}_{\calG_\rho^{\otimes O(Dd^2/\eps^2)}}}} \le \eps\]
%	\[ \norm{2}{\calR(A^{(6)}) - A^{(6)}} \le O(\sqrt{\eps}) \qquad \sAND \qquad \norm{2}{\calR(B^{(6)}) - B^{(6)}} \le O(\sqrt{\eps})\]
%	
%	\item {\bf Rounding to $\Delta_k$}: Here we use Lemma~\ref{lem:close-strategies}, to conclude that since, $\norm{2}{\calR(A^{(6)}) - A^{(6)}}, \norm{2}{\calR(B^{(6)}) - B^{(6)}} \le O(\sqrt{\eps})$, we have that,
%	\[ \inabs{\inangle{\calR_i(A^{(6)}), \calR_j(B^{(6)})}_{\mu^{\otimes O(Dd^2/\eps^2)}} - \inangle{A^{(6)}_i, B^{(6)}_j}_{\mu^{\otimes O(Dd^2/\eps^2)}}} ~\le~ O(\sqrt{\eps})\]
\end{enumerate}

\paragraph{Putting it together.} We now show how to put together \Cref{eqn:NIS-1,eqn:NIS-2,eqn:NIS-3,eqn:NIS-4,eqn:NIS-5,eqn:NIS-6} to get our final conclusion. We now define our final functions $\wtilde{A} : \calZ^{n_0} \to \Delta_k$ and $\wtilde{B} : \calZ^{n_0} \to \Delta_k$ as follows. Firstly, we interpret the $n_0 = h + D_1$ coordinates of $\bx$ as two parts: {\em head} coordinates $\bx_H \in \calZ^h$ and the remaining coordinates $\bx_R \in \calZ^{D_1}$. Similarly for $\by$.
\[ \wtilde{A}(\bx) = \wtilde{A}(\bx_H, \bx_R) = (A^{(6)})^{\bx_H}(\bx_R) \qquad \sAND \qquad \wtilde{B}(\by) = \wtilde{B}(\by_H, \by_R) = (B^{(6)})^{\by_H}(\by_R)\;.\]
We now show that for all $i, j \in [k]$, it holds that,
\[ 	\inabs{\inangle{\wtilde{A}_i, \wtilde{B}_j}_{\mu^{\otimes n_0}} - \inangle{A_i, B_j}_{\mu^{\otimes N}}} ~\le~ O(\delta) \]
We note that,
\begin{align*}
& \inangle{\wtilde{A}_i, \wtilde{B}_j}_{\mu^{\otimes n_0}} \\
&~=~ \Ex_{(\bx_H, \by_H) \sim \mu^{\otimes h}} \inangle{(A^{(6)}_i)^{\bx_H}, (B^{(6)}_j)^{\by_H}}_{\mu^{\otimes D_1}}\\
&~=~ \Pr\inbmat{\text{$(\bx_H, \by_H)$ is} \\ \text{({\em good} and {\em lucky})}} \cdot \Ex_{\substack{(\bx_H, \by_H) \sim \mu^{\otimes h}\\ | \ \text{{\em good} \& {\em lucky}}}} \inangle{(A^{(6)}_i)^{\bx_H}, (B^{(6)}_j)^{\by_H}}_{\mu^{\otimes D_1}}\\
& \quad ~+~ \Pr\inbmat{\text{$(\bx_H, \by_H)$ is not} \\ \text{({\em good} and {\em lucky})}} \cdot \Ex_{\substack{(\bx_H, \by_H) \sim \mu^{\otimes h}\\ | \ \text{not ({\em good} \& {\em lucky})}}} \inangle{(A^{(6)}_i)^{\bx_H}, (B^{(6)}_j)^{\by_H}}_{\mu^{\otimes D_1}}\\
&~=~ \Pr\inbmat{\text{$(\bx_H, \by_H)$ is} \\ \text{({\em good} and {\em lucky})}} \cdot \Ex_{\substack{(\bx_H, \by_H) \sim \mu^{\otimes h}\\ | \ \text{{\em good} \& {\em lucky}}}} \inangle{(A^{(6)}_i)^{\bx_H}, (B^{(6)}_j)^{\by_H}}_{\mu^{\otimes D_1}} ~\pm~ O(\delta)\\
& \qquad \qquad \qquad \qquad \qquad \qquad \qquad \qquad \ldots \inparen{\text{since } \Pr\inbmat{\text{$(\bx_H, \by_H)$ is not} \\ \text{({\em good} and {\em lucky})}} \le O(\delta)}\\
&~=~ \Pr\inbmat{\text{$(\bx_H, \by_H)$ is} \\ \text{({\em good} and {\em lucky})}} \cdot \Ex_{\substack{(\bx_H, \by_H) \sim \mu^{\otimes h}\\ | \ \text{{\em good} \& {\em lucky}}}} \inangle{\calR_i\inparen{(A^{(1)})^{\bx_H}}, \calR_j\inparen{(B^{(1)})^{\by_H}}}_{\mu^{\otimes N_1}} ~\pm~ O(\delta)\\
& \qquad \qquad \qquad \qquad \qquad \qquad \qquad \qquad \ldots \inparen{\inmat{\text{combining \Cref{eqn:NIS-2,eqn:NIS-3,eqn:NIS-4,eqn:NIS-5,eqn:NIS-6}}\\ \text{which hold for {\em good} and {\em lucky} } (\bx_H,\by_H)}}\\
&~=~ \Ex_{(\bx_H, \by_H) \sim \mu^{\otimes h}} \inangle{\calR_i\inparen{(A^{(1)})^{\bx_H}}, \calR_j\inparen{(B^{(1)})^{\by_H}}}_{\mu^{\otimes N_1}} ~\pm~ O(\delta)\\
& \qquad \qquad \qquad \qquad \qquad \qquad \qquad \qquad \ldots \inparen{\text{since } \Pr\inbmat{\text{$(\bx_H, \by_H)$ is} \\ \text{({\em good} and {\em lucky})}} \in [1-O(\delta), \ 1]}\\
&~=~ \inangle{\calR_i(A^{(1)}), \calR_j(B^{(1)})}_{\mu^{\otimes N}} ~\pm~ O(\delta)\\
&~=~ \inangle{A_i, B_j}_{\mu^{\otimes N}} ~\pm~ O(\delta) \qquad \qquad \quad \ldots\text{(using \Cref{eqn:NIS-1})}
\end{align*}

%\noindent Thus, we started with functions $A: \calZ^n \to \Delta_q$ and $B : \calZ^n \to \Delta_q$ and we end with functions $\wtilde{A} : \calZ^{n_0} \to \Delta_q$ and $\wtilde{B} : \calZ^{n_0} \to \Delta_q$ such that for every $i, j \in [q]$, we have
%\[ \inabs{\Ex_{xy} A_i(x) B_j(y) - \Ex_{x,y} \wtilde{A}_i(x) \wtilde{B}_j(y)} ~\le~ O(\delta) \]
\noindent Thus, more strongly, if we instantiate $\delta$ as $O(\eps/k^2)$, then we get that
\[\dTV\inparen{(A(\bx),B(\by))_{(\bx,\by)\sim \mu^{\otimes N}}, \ (\wtilde{A}(\ba), \wtilde{B}(\bb))_{\ba, \bb \sim \mu^{\otimes n_0}}} \le \eps\;.\]
\end{proofof}

%% file: fig_dependency.tex
\begin{figure}
\begin{center}
\begin{tikzpicture}[scale=0.9, transform shape]

\def\ColSmooth{Gblue}
\def\ColMultiLin{Gred}
\def\ColDimRed{Ggreen}
\def\ColRound{Gyellow}
\def\ColRegularity{Ggreen}
\def\ColInvariance{Gyellow}
\def\Brightness{50}

\node[box, fill=black!20] (param) at (-3,0) {\large $(\calZ \times \calZ, \mu), \ k, \ \eps$};

\node[box, fill=\ColSmooth!\Brightness] (smooth1) at (3, 2.3) {\shortstack{Smoothing : 1\\ (\Cref{lem:smoothing_main})}}
edge[in=20, out=190, latex-]  node[above, midway] {$\delta \quad$} (param);

\node[box, fill=\ColInvariance!\Brightness] (inv1) at (-3, 3) {\shortstack{Invariance : 1\\ (\Cref{lem:our_invariance})}}
edge[in=170, out=0, latex-]  node[above, midway] {$d \gets d(\delta)$} (smooth1)
edge[in=90, out=270, latex-]  node[left, midway] {$\delta$} (param);

%\node [draw, thick, fill=black!30, circle] (minnode1) at (-2,5) {}
%edge[in=150, out=180, latex-]  node[left, midway] {$\frac{\eps}{3}$} (param)
%edge[in=90, out=270, latex-]  node[right, midway] {$\tau = \tau(d,\delta)$} (inv1)
%;

\node[box, fill=\ColRegularity!\Brightness] (reg) at (3, 5) {\shortstack{Regularity \\ (\Cref{lem:joint_regularity})}}
edge[in=90, out=270, latex-]  node[right, midway] {$d \gets d(\delta)$} (smooth1)
edge[in=90, out=180, latex-]  node[above, midway] {$\tau \gets \tau(d,\delta)$} (inv1);

\node[box, fill=\ColSmooth!\Brightness] (smooth2) at (3, -2.3) {\shortstack{Smoothing : 2\\ (\Cref{lem:smoothing_main})}}
edge[in=-20, out=-190, latex-]  node[below, midway] {$\delta \quad$} (param);

\node[box, fill=\ColInvariance!\Brightness] (inv2) at (-3, -3) {\shortstack{Invariance : 2 \\ (\Cref{lem:our_invariance})}}
edge[in=-170, out=0, latex-]  node[below, midway] {$d \gets d(\delta)$} (smooth2)
edge[in=270, out=90, latex-]  node[left, midway] {$\delta$} (param);

%\node [draw, thick, fill=black!30, circle] (minnode2) at (-2,-5) {}
%edge[in=-150, out=180, latex-]  node[left, midway] {$\frac{\eps}{3}$} (param)
%edge[in=270, out=90, latex-]  node[right, midway] {$\tau = \tau(d,\delta)$} (inv2)
%;

\node[box, fill=\ColMultiLin!\Brightness] (multilin) at (3, -5) {\shortstack{Multilinearize \\ (\Cref{lem:multilin_main})}}
edge[in=270, out=90, latex-]  node[right, midway] {$d \gets d(\delta)$} (smooth2)
edge[in=270, out=180, latex-]  node[below, midway] {$\delta \gets \tau(d,\delta)$} (inv2)
;

\node[box, upper right=\ColDimRed!\Brightness, upper left=\ColSmooth!\Brightness, lower right=\ColMultiLin!\Brightness,lower left=\ColRound!\Brightness] (GNIS) at (2.5, 0) {\shortstack{Gaussian NIS \\ (\Cref{thm:gaussian-non-int-sim})}}
edge[in=0, out=180, latex-]  node[above, midway] {$\eps \gets \delta$} (param);

\node[box, fill=Sepia!20] (final) at (8, 0) {$n_0 = h + D \cdot t$}
edge[in=0, out=90, latex-]  node[right, midway] {$h \gets h(d,\tau)$} (reg)
edge[in=0, out=180, latex-]  node[above, midway] {$D \gets n_0(\delta)$} (GNIS)
edge[in=0, out=270, latex-]  node[right, midway] {$t \gets t(d,\delta)$} (multilin);
\end{tikzpicture}
\caption{Dependency of parameters in the proof of \Cref{thm:non-int-sim}}
\label{fig:dependency}
\end{center}
\end{figure}

%% file: fig_GeneralNIS.tex
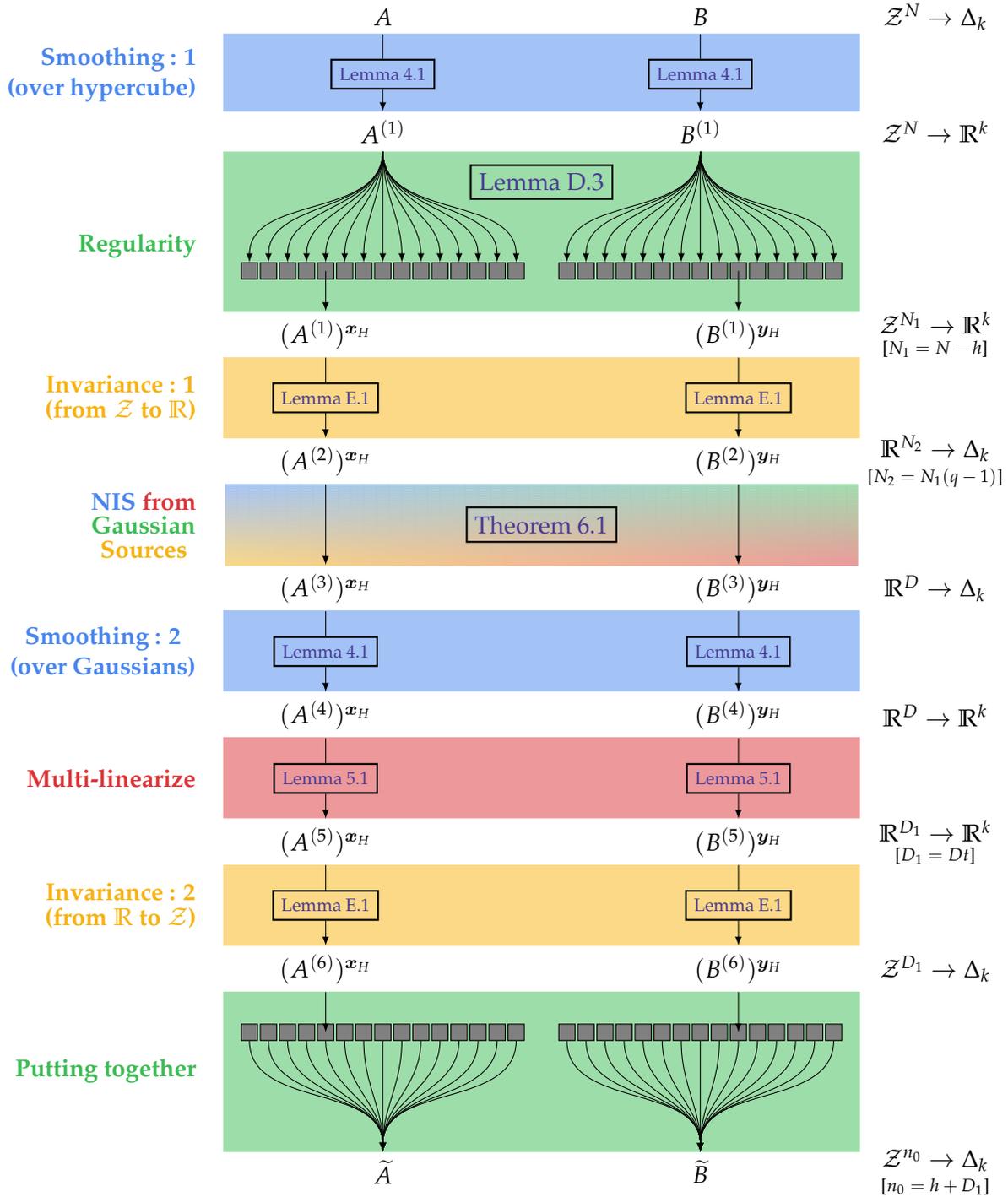
\begin{figure}
\begin{center}
\begin{tikzpicture}[scale=1, transform shape]

\def\ColSmooth{Gblue}
\def\ColMultiLin{Gred}
\def\ColDimRed{Ggreen}
\def\ColRound{Gyellow}
\def\ColRegularity{Ggreen}
\def\ColInvariance{Gyellow}
\def\Brightness{50}

\def \stepname{-2.8}
\def \alice{0}
\def \bob{5}
\def \midd{2.5}
\def \descrip{8.7}

\def \ycorr{0}
\node (A) at (\alice,\ycorr) {$A$};
\node (B) at (\bob, \ycorr) {$B$};
\node at(\descrip, \ycorr) {$\calZ^N \to \Delta_k$};

\def \ycorr{-0.9}
\node[box, below, minimum width=10cm, minimum height=1.2cm, fill=\ColSmooth!\Brightness, draw=\ColSmooth!\Brightness] at (\midd, \ycorr+0.64) {};
%\node[box, below right, minimum width=1.2cm, minimum height=1.2cm, fill=blue!20, draw=blue!20] at (\stepname-0.2,\ycorr+0.6) {Smoothing};
\node[left] at (\stepname,\ycorr) {\shortstack[r]{\textcolor{\ColSmooth}{\bf Smoothing : 1}\\ \textcolor{\ColSmooth}{\bf (over hypercube)}}};
\node[box] (SmoothA1) at (\alice,\ycorr) {\scriptsize  \Cref{lem:smoothing_main}} edge[-] (A);
\node[box] (SmoothB1) at (\bob,\ycorr) {\scriptsize  \Cref{lem:smoothing_main}} edge[-] (B);

\def \ycorr{-1.8}
\node (A1) at (\alice,\ycorr) {$A^{(1)}$} edge[latex-] (SmoothA1);
\node (B1) at (\bob,\ycorr) {$B^{(1)}$} edge[latex-] (SmoothB1);
\node at (\descrip,\ycorr) {$\calZ^N \to \bbR^k$};

\def \ycorr{-2.8}
\node[box, below, minimum width=10cm, minimum height=2.5cm, fill=\ColRegularity!\Brightness, draw=\ColRegularity!\Brightness] at (\midd, \ycorr+0.68) {};
\node[left] at (\stepname,\ycorr-0.8) {\textcolor{\ColRegularity}{\bf Regularity}};
\node[box] (Reg) at (\midd,\ycorr+0.2) {\Cref{lem:joint_regularity}};
%\node[box] (RegA) at (\alice,\ycorr) {\scriptsize  \Cref{lem:joint_regularity}} edge[-] (A1);
%\node[box] (RegB) at (\bob,\ycorr) {\scriptsize \Cref{lem:joint_regularity}} edge[-] (B1);

\def \ycorr{-4}
\def \boxwidth{0.3}
\foreach \num in {-7,...,7} {
	\node[box, fill=black!50, thin, minimum width=0.2cm, minimum height=0.2cm] at (\alice+\num*\boxwidth,\ycorr) {} edge[latex-,out=90,in=270] (A1);
}
\foreach \num in {-7,...,7} {
	\node[box, fill=black!50, thin, minimum width=0.2cm, minimum height=0.2cm] at (\bob+\num*\boxwidth,\ycorr) {} edge[latex-,out=90,in=270] (B1);
}

\def \oldycorr{-4}
\def \oldalice{0}
\def \oldbob{5}
\def \ycorr{-5}
\def \alice{\oldalice - 3*\boxwidth}
\def \bob{\oldbob+2*\boxwidth}
\node (A2) at (\alice,\ycorr) {$(A^{(1)})^{\bx_H}$} edge[latex-] (\alice,\oldycorr) {};
\node (B2) at (\bob,\ycorr) {$(B^{(1)})^{\by_H}$} edge[latex-] (\bob,\oldycorr);
\node at (\descrip,\ycorr) {\shortstack{$\calZ^{N_1} \to \bbR^k$  \\ \scriptsize [$N_1 = N - h$]}};

\def\ycorr{-6}
\node[box, below, minimum width=10cm, minimum height=1.25cm, fill=\ColInvariance!\Brightness, draw=\ColInvariance!\Brightness] at (\midd, \ycorr+0.64) {};
\node[left] at (\stepname,\ycorr) {\shortstack{\textcolor{\ColInvariance}{\bf Invariance : 1} \\ \textcolor{\ColInvariance}{\bf (from $\calZ$ to $\bbR$)}}};
%\node[box] (Inv) at (\midd,\ycorr) {\Cref{lem:our_invariance}};
\node[box] (InvA1) at (\alice,\ycorr) {\scriptsize  \Cref{lem:our_invariance}} edge[-] (A2);
\node[box] (InvB1) at (\bob,\ycorr) {\scriptsize \Cref{lem:our_invariance}} edge[-] (B2);

\def \ycorr{-7}
\node (A3) at (\alice,\ycorr) {$(A^{(2)})^{\bx_H}$} edge[latex-] (InvA1);
\node (B3) at (\bob,\ycorr) {$(B^{(2)})^{\by_H}$} edge[latex-] (InvB1);
\node at (\descrip,\ycorr) {\shortstack{$\bbR^{N_2} \to \Delta_k$ \\ \scriptsize [$N_2 = N_1 (q-1)$]}};

\def\ycorr{-8}
%\shade[upper right=red!20,upper left=blue!20, lower right=green!20,lower left=black!20] (0,0) rectangle (2,2);
\node[box, upper right=\ColDimRed!\Brightness,upper left=\ColSmooth!\Brightness, lower right=\ColMultiLin!\Brightness,lower left=\ColRound!\Brightness, below, minimum width=10cm, minimum height=1.32cm, draw=white] at (\midd, \ycorr+0.67) {};
\node[left] at (\stepname,\ycorr) {\shortstack{\textcolor{\ColSmooth}{\bf NIS} \textcolor{\ColMultiLin}{\bf from} \\ \textcolor{\ColDimRed}{\bf Gaussian} \\ \textcolor{\ColRound}{\bf Sources}}};
\node[box] (GNIS) at (\midd,\ycorr) {\Cref{thm:gaussian-non-int-sim}};
%\node[box] (GNISA) at (\alice,\ycorr) {\scriptsize \Cref{thm:gaussian-non-int-sim}} edge[-] (A3);
%\node[box] (GNISB) at (\bob,\ycorr) {\scriptsize \Cref{thm:gaussian-non-int-sim}} edge[-] (B3);
%\draw[snake, segment length=2mm] ([shift={(0,0)}]GNISA.east) -- ([shift={(0,0)}]GNISB.west);

\def \ycorr{-9}
\node (A4) at (\alice,\ycorr) {$(A^{(3)})^{\bx_H}$} edge[latex-] (A3);
\node (B4) at (\bob,\ycorr) {$(B^{(3)})^{\by_H}$} edge[latex-] (B3);
\node at (\descrip,\ycorr) {$\bbR^{D} \to \Delta_k$};

\def\ycorr{-10}
\node[box, below, minimum width=10cm, minimum height=1.25cm, fill=\ColSmooth!\Brightness, draw=\ColSmooth!\Brightness] at (\midd, \ycorr+0.65) {};
\node[left] at (\stepname,\ycorr) {\shortstack{\textcolor{\ColSmooth}{\bf Smoothing : 2}\\ \textcolor{\ColSmooth}{\bf (over Gaussians)}}};
\node[box] (SmoothA2) at (\alice,\ycorr) {\scriptsize \Cref{lem:smoothing_main}} edge[-] (A4);
\node[box] (SmoothB2) at (\bob,\ycorr) {\scriptsize \Cref{lem:smoothing_main}} edge[-] (B4);

\def \ycorr{-11}
\node (A5) at (\alice,\ycorr) {$(A^{(4)})^{\bx_H}$} edge[latex-] (SmoothA2);
\node (B5) at (\bob,\ycorr) {$(B^{(4)})^{\by_H}$} edge[latex-] (SmoothB2);
\node at (\descrip,\ycorr) {$\bbR^{D} \to \bbR^k$};

\def \ycorr{-12}
\node[box, below, minimum width=10cm, minimum height=1.25cm, fill=\ColMultiLin!\Brightness, draw=\ColMultiLin!\Brightness] at (\midd, \ycorr+0.65) {};
\node[left] at (\stepname,\ycorr) {\textcolor{\ColMultiLin}{\bf Multi-linearize}};
\node[box] (MultiLinA) at (\alice,\ycorr) {\scriptsize \Cref{lem:multilin_main}} edge[-] (A5);
\node[box] (MultiLinB) at (\bob,\ycorr) {\scriptsize \Cref{lem:multilin_main}} edge[-] (B5);

\def \ycorr{-13}
\node (A6) at (\alice,\ycorr) {$(A^{(5)})^{\bx_H}$} edge[latex-] (MultiLinA);
\node (B6) at (\bob,\ycorr) {$(B^{(5)})^{\by_H}$} edge[latex-] (MultiLinB);
\node at (\descrip,\ycorr) {\shortstack{$\bbR^{D_1} \to \bbR^k$  \\ \scriptsize [$D_1 = D t$]}};

\def\ycorr{-14}
\node[box, below, minimum width=10cm, minimum height=1.25cm, fill=\ColInvariance!\Brightness, draw=\ColInvariance!\Brightness] at (\midd, \ycorr+0.64) {};
\node[left] at (\stepname,\ycorr) {\shortstack{\textcolor{\ColInvariance}{\bf Invariance : 2} \\ \textcolor{\ColInvariance}{\bf (from $\bbR$ to $\calZ$)}}};
%\node[box] (Inv2) at (\midd,\ycorr) {\Cref{lem:our_invariance}};
\node[box] (InvA2) at (\alice,\ycorr) {\scriptsize  \Cref{lem:our_invariance}} edge[-] (A6);
\node[box] (InvB2) at (\bob,\ycorr) {\scriptsize \Cref{lem:our_invariance}} edge[-] (B6);

\def \ycorr{-15}
\node[box, below, minimum width=10cm, minimum height=2.5cm, fill=\ColRegularity!\Brightness, draw=\ColRegularity!\Brightness] at (\midd, \ycorr-0.36) {};
\node (A7) at (\alice,\ycorr) {$(A^{(6)})^{\bx_H}$} edge[latex-] (InvA2);
\node (B7) at (\bob,\ycorr) {$(B^{(6)})^{\by_H}$} edge[latex-] (InvB2);
\node at (\descrip,\ycorr) {$\calZ^{D_1} \to \Delta_k$};

\def \ycorr{-18.2}
\node (A8) at (\oldalice,\ycorr) {$\wtilde{A}$};
\node (B8) at (\oldbob,\ycorr) {$\wtilde{B}$};
\node at (\descrip,\ycorr) {\shortstack{$\calZ^{n_0} \to \Delta_k$ \\ \scriptsize [$n_0 = h + D_1$]}};

\def \ycorr{-16}
\def \boxwidth{0.3}
\foreach \num in {-7,...,7} {
	\node[box, fill=black!50, thin, minimum width=0.2cm, minimum height=0.2cm] at (\oldalice+\num*\boxwidth,\ycorr) {} edge[-latex,out=270,in=90] (A8);
}
\foreach \num in {-7,...,7} {
	\node[box, fill=black!50, thin, minimum width=0.2cm, minimum height=0.2cm] at (\oldbob+\num*\boxwidth,\ycorr) {} edge[-latex,out=270,in=90] (B8);
}
\node[left] at (\stepname,\ycorr-0.6) {\textcolor{\ColRegularity}{\bf Putting together}};

\draw[-latex] (A7) -- (\alice,\ycorr);
\draw[-latex] (B7) -- (\bob,\ycorr);

\end{tikzpicture}
\caption{Transformations for Non-interactive simulation from Arbitrary Discrete Sources}
\label{fig:General_NIS}
\end{center}
\end{figure}

%% file: ack.tex
\section{Acknowledgments}
The authors are extremely grateful to Madhu Sudan for very many helpful conversations throughout all the stages of this project. The authors would also like to thank Anindya De, Elchanan Mossel and Joe Neeman for clarifying explanations of their papers and helpful discussions.

%% file: main.bbl
\newcommand{\etalchar}[1]{$^{#1}$}

%% file: apx_dim-reduction.tex
\section{Proofs of Mean and Variance Bounds in Dimension Reduction}\label{sec:app_dim_red}

In this section, we provide the proof of \Cref{lem:mean_var_bound}. This is the main new technical component introduced in this paper. Even though the calculations might seem cumbersome, they involve mostly elementary steps. The proof breaks down into three modular steps, which we describe first.

Recall that we are given degree $d$ multilinear polynomials $A : \bbR^N \to \bbR$ and $B : \bbR^N \to \bbR$. For a matrix $M$ sampled from $\calN(0,1)^{\otimes (N \times D)}$, we defined functions $A_M : \bbR^D \to \bbR$ and $B_M : \bbR^D \to \bbR$ as
\[A_M(\ba) ~=~ A\inparen{\frac{M \ba}{\norm{2}{\ba}}} \qquad \sAND \qquad B_M(\bb) = B\inparen{\frac{M \bb}{\norm{2}{\bb}}}\]
and we defined their correlation as $F(M) ~\defeq~ \inangle{A_M, B_M}_{\calG_{\rho}^{\otimes D}}$. \Cref{lem:mean_var_bound} proves bounds on the mean and variance of $F(M)$, which we restate below for convenience.

\meanvarbound*

\noindent The proof of the above Lemma proceeds in three main steps.
\begin{enumerate}
	\item In \Cref{subsec:meta-lemma}, we first prove a {\em meta-lemma} (\Cref{le:norm_Gauss}) that will help us prove both the mean and variance bounds; indeed this meta-lemma is at the heart of why \Cref{thm:dim-reduction} holds. Morally, this lemma says that if we have an expectation of a product of a small number of inner products of normalized correlated Gaussian vectors, then, we can exchange the product and the expectations while incurring only a small additive error.
	\item In \Cref{subsec:bds_for_multilinear_monom}, we prove strong enough bounds on the mean and co-variances of degree-$d$ multilinear monomials, under the above transformation of replacing $\bX, \bY \in \bbR^N$ (inputs to $A$ and $B$) by $\frac{M\ba}{\|\ba\|_2}$ and $\frac{M\bb}{\|\bb\|_2}$ respectively. %\todo{What are $\bX$ and $\bY$ here?}
	\item In \Cref{subec:bds_for_gen_multilinear_polys}, we finally use the above bounds on mean and co-variances of degree-$d$ multilinear monomials in order to prove \Cref{lem:mean_var_bound}.
\end{enumerate}
\paragraph{Remark.} To make our notations convenient, we will often write equations such as $\alpha = \beta \pm \eps$ which is to be interpreted as $\inabs{\alpha - \beta} \le \eps$.

\subsection{Product of Inner Products of Normalized Correlated Gaussian Vectors}\label{subsec:meta-lemma}

%We will reduce the task of obtaining bounds on the mean and variance for multilinear polynomials to the task of obtaining similar bounds for individual monomials that show up in the Hermite expansion of $A$ and $B$. 

The following is the main lemma in this subsection (this is the {\em meta-lemma} alluded to earlier).
\begin{lem}\label{le:norm_Gauss}
	Given parameter $d$ and $\delta$, there exists an explicitly computable $D := D(d, \delta)$ such that the following holds:\\
	Let $(\bu_1,\ldots,\bu_d, \bv_1,\ldots, \bv_d)$ be a multivariate Gaussian distribution such that,
	\begin{itemize}
	\item each $\bu_i,\bv_i \in \bbR^D$ are marginally distributed as $D$-dimensional standard Gaussians, i.e. $\gamma_D$.
	\item for each $j \in [D]$, the joint distribution of the $j$-th coordinates, i.e., $(u_{1,j}, \ldots, u_{d,j}, v_{1,j}, \ldots, v_{d,j})$, is independent across different values of $j$.
	\end{itemize}
	Then,
	$$\inabs{\Ex\limits_{\set{\bu_i, \bv_i}_i} \ \insquare{\prod_{i=1}^d \frac{\inangle{\bu_i, \bv_i}}{\norm{2}{\bu_i} \norm{2}{\bv_i}}} - \prod_{i=1}^d \ \Ex\limits_{\set{\bu_i, \bv_i}_i} \ \insquare{\frac{\inangle{\bu_i, \bv_i}}{D}}} ~\le~ \delta.$$
	In particular, one may take $D = \frac{d^{O(d)}}{\delta^2}$.
\end{lem}
\noindent We point out that there are two steps taking place in \Cref{le:norm_Gauss}:\\
(i) the replacement of $\norm{2}{\bu_i}$ (and $\norm{2}{\bv_i}$) by $\sqrt{D}$ (around which it is tightly concentrated), and\\
(ii) the interchanging of the expectation and the product.\\

\noindent We will handle each of these changes one by one.

\subsubsection*{Product of Negative Moments of $\ell_2$-norm of Correlated Gaussian vectors}

In order to handle the replacement of $\norm{2}{\bu_i}$ (and $\norm{2}{\bv_i}$) by $\sqrt{D}$, we will show the following lemma which gives us useful bounds on the mean and variance of products of negative powers of the $\ell_2$-norm of a standard Gaussian vector. %In the rest of this section, all the asymptotic $o(\cdot)$ and $O(\cdot)$ terms denote \emph{explicitly computable} functions.

\begin{lem}\label{le:corr_norms_rec_ex}
	Let $\bw_1$, $\bw_2$, $\dots$, $\bw_{\ell}$ be (possibly correlated) multivariate Gaussians where each $\bw_i \in \mathbb{R}^D$ is marginally distributed as a $D$-dimensional standard Gaussian (i.e., $\gamma_D$), and let $d_1, d_2, \dots, d_{\ell}$ be non-negative integers with $d := \sum_{i=1}^{\ell} d_i$. Then,
	\begin{align*}
	\inabs{\Ex\bigg[\prod\limits_{i=1}^{\ell} \frac{1}{\|\bw_i\|_2^{d_i}}\bigg] - \frac{1}{D^{d/2}}} &~\le~ \frac{1}{D^{\frac{d}{2}}} \cdot O\bigg(\frac{d^5}{D} \bigg),\\
	\Var\insquare{\prod\limits_{i=1}^{\ell} \frac{1}{\|\bw_i\|_2^{d_i}}} &~\le~ \frac{1}{D^d} \cdot O\inparen{\frac{d^5}{D}}.
	\end{align*}
\end{lem}

\begin{remark}
	It is conceivable that the bounds in \Cref{le:corr_norms_rec_ex} could be improved in terms of the dependence on $d$. However, this was not central to our application, so we go ahead with the stated bounds. The main point to note in the above lemma is the extra factor of $D$ in the denominator.
\end{remark}

\noindent We start out by first proving the base case where we have a single vector $\bw$, that is, $\ell = 1$.

\begin{proposition}\label{prop:inv_chi_mom_dist}
	There exists an absolute constant $C$ such that for sufficiently large $d, D \in \bbZ_{> 0}$, such that $D > Cd^2$, we have that for $\bw \sim \gamma_D$,
	\begin{align}
	\bigg| \Ex_{\bw}\bigg[\frac{1}{\|\bw\|_2^d}\bigg] - \frac{1}{D^{d/2}} \bigg| &~\le~ C \cdot \inparen{\frac{d^2}{D^{\frac{d}{2}+1}}}, \label{eq:mean_univ_bd}\\
	\Var_{\bw}\bigg[ \frac{1}{\| \bw \|_2^d}\bigg] &~\le~ 8 C \cdot \inparen{\frac{d^2}{D^{d+1}}}.\label{eq:var_univ_bd}
	\end{align}
\end{proposition}
\begin{proof}%[Proof of Fact~\ref{prop:inv_chi_mom_dist}]
	It is well-known that the distribution of $\norm{2}{\bw}$ follows a $\chi$-distribution with parameter $D$, and whose probability density function is given by
	\begin{equation*}
		f_D(x) = \frac{x^{D-1} \cdot e^{- \frac{x^2}{2}}}{2^{\frac{D}{2} - 1} \cdot \Gamma(\frac{D}{2})},
	\end{equation*}
	for every $x \geq 0$ (and where $\Gamma(\cdot)$ denotes the Gamma function). Thus, we have that
	\begin{align*}
		\Ex_{\bw}\bigg[\frac{1}{\|\bw\|^d}\bigg] &= \int_{0}^{\infty} \frac{1}{x^d} \cdot f_D(x) dx\\ 
		&=  \int_{0}^{\infty} \frac{x^{D-d-1} \cdot e^{- \frac{x^2}{2}}}{2^{\frac{D}{2} - 1} \cdot \Gamma(\frac{D}{2})} dx\\ 
		&= \frac{2^{\frac{D-d-1}{2}} \cdot \Gamma\inparen{\frac{D-d}{2}}}{2^{\frac{D}{2} - 1} \cdot \Gamma(\frac{D}{2})}\\ 
		&= \frac{1}{D^{d/2}} \cdot \inparen{1 \pm O\inparen{\frac{d^2}{D}}},
	\end{align*}
	where the last equality follows from the following Stirling's approximation of the Gamma function, which holds for every real number $z > 0$:
	\[ \Gamma(z+1) = \sqrt{ 2 \pi z } \cdot \inparen{\frac{z}{e}}^z \cdot \inparen{1 \pm O\inparen{\frac{1}{z}}}. \]
	This completes the proof of \Cref{eq:mean_univ_bd}, for the explicit constant $C$ that can be derived from the Stirling's approximation. Now, \Cref{eq:var_univ_bd} immediately follows as:
	\begin{align*}
		\Var_{\bw}\insquare{\frac{1}{\| \bw \|^d}} &~=~ \Ex_{\bw} \insquare{\frac{1}{\|\bw\|^{2d}}} -  \Ex_{\bw}\insquare{\frac{1}{\|\bw\|^d}}^2\\ 
		&~=~ \inparen{\frac{1}{D^d} \pm C \cdot \inparen{\frac{(2d)^2}{D^{d+1}}}} - \inparen{\frac{1}{D^{d/2}} \pm C \cdot \inparen{\frac{d^2}{D^{d/2+1}}}}^2\\ 
		&~\le~ 8C \cdot \inparen{\frac{d^2}{D^{d+1}}},
	\end{align*}
	where, we use that $D$ is sufficiently large that $C^2 \inparen{\frac{d^4}{D^{d+2}}} < 2C \cdot \inparen{\frac{d^2}{D^{d+1}}}$, i.e. $D > Cd^2$.
\end{proof}

\noindent We now show how to generalize the above to prove \Cref{le:corr_norms_rec_ex}.

\begin{proof}[Proof of \Cref{le:corr_norms_rec_ex}]
	More specifically, we will show that,
	\begin{align}
	\inabs{\Ex\bigg[\prod\limits_{i=1}^{\ell} \frac{1}{\|\bw_i\|_2^{d_i}}\bigg] - \frac{1}{D^{d/2}}} &~\le~ C \cdot \ell^3 \cdot \bigg(\frac{d^2}{D^{\frac{d}{2}+1}} \bigg) \label{eq:exp_prod_dep}\\
	\Var\insquare{\prod\limits_{i=1}^{\ell} \frac{1}{\|\bw_i\|_2^{d_i}}} &~\le~ 8C \cdot \ell^3 \cdot \inparen{\frac{d^2}{D^{d+1}}}\label{eq:exp_prod_dep_var}
	\end{align}
	where $C$ is the absolute constant (as obtained in \Cref{prop:inv_chi_mom_dist}). This implies the lemma since $\ell \le d$.
	
	We proceed by induction on $\ell$ (more specifically on $\log \ell$). For $\ell = 1$, the bound immediately follows from \Cref{prop:inv_chi_mom_dist}. For the inductive step, we assume that the bound in \Cref{eq:exp_prod_dep,eq:exp_prod_dep_var} holds for $\ell$, and we prove that the bound also holds for $2\ell$. While it may seem that our bounds are being proven only when $\ell$ is a power of $2$, it is not hard to see that our proof could be done for non powers of $2$ as well, giving a bound that is monotonically increasing in $\ell$ and hence it suffices having proved it for $\ell$ that are powers of $2$.	Let $d_1, d_2, \dots, d_{2\ell}$ be non-negative integers with $d:= \sum_{i=1}^{2\ell} d_i$. For notational convenience, let $s_1 = \sum_{i=1}^{\ell} d_i$ and $s_2 = \sum_{i=\ell+1}^{2\ell} d_i$, and so $d = s_1 + s_2$. %\todo{It might be better to formalize this footnote in-text?}
	
	We will first prove \Cref{eq:exp_prod_dep}. The main idea that we use to prove this inductively is: for any two random variables $X$ and $Y$, it holds that $\Ex[XY] - \Ex[X] \Ex[Y] = \Cov[X,Y]$, and that the covariance satisfies $|\Cov[X,Y]| \le \sqrt{\Var[X] \cdot \Var[Y]}$ (by Cauchy-Schwarz inequality). Thus, we get,
	\begin{align}\label{eq:cov_manip}
		&\Ex\bigg[\prod_{i=1}^{2\ell} \frac{1}{\norm{2}{\bw_i}^{d_i}} \bigg]
		~=~ \Ex\bigg[\prod_{i=1}^{\ell} \frac{1}{\norm{2}{\bw_i}^{d_i}} \ \cdot \  \prod_{i=\ell+1}^{2\ell} \frac{1}{\norm{2}{\bw_i}^{d_i}} \bigg]\nonumber\\
		%&~=~ \Ex\bigg[\frac{1}{\norm{2}{\bw_{\ell+1}}^{d_{\ell+1}}} \cdot \prod_{i=1}^{\ell} \frac{1}{\norm{2}{\bw_i}^{d_i}}\bigg] \nonumber\\ 
		&~=~ \Ex\bigg[\prod_{i=1}^{\ell} \frac{1}{\norm{2}{\bw_i}^{d_i}} \bigg] \cdot \Ex\bigg[ \prod_{i=\ell+1}^{2\ell} \frac{1}{\norm{2}{\bw_i}^{d_i}}\bigg] ~+~ \Cov\bigg[\prod_{i=1}^{\ell} \frac{1}{\norm{2}{\bw_i}^{d_i}} , \prod_{i=\ell+1}^{2\ell} \frac{1}{\norm{2}{\bw_i}^{d_i}} \bigg] \nonumber\\ 
		&~=~ \Ex\bigg[\prod_{i=1}^{\ell} \frac{1}{\norm{2}{\bw_i}^{d_i}} \bigg] \cdot \Ex\bigg[ \prod_{i=\ell+1}^{2\ell} \frac{1}{\norm{2}{\bw_i}^{d_i}}\bigg] ~\pm~ \sqrt{\Var\bigg[\prod_{i=1}^{\ell} \frac{1}{\norm{2}{\bw_i}^{d_i}} \bigg] \cdot \Var\bigg[ \prod_{i=\ell+1}^{2\ell} \frac{1}{\norm{2}{\bw_i}^{d_i}}\bigg]}\;.
	\end{align}
	Using the inductive assumption w.r.t. $\ell$, we get that,
	\begin{align}
	\Ex\bigg[ \prod_{i=1}^{\ell}       \frac{1}{\norm{2}{\bw_i}^{d_i}}\bigg] &~=~ \frac{1}{D^{s_1/2}} \inparen{1 \pm C \cdot \ell^3 \cdot \bigg(\frac{s_1^2}{D} \bigg)} \label{eq:exp_ind1}\\
	\Ex\bigg[ \prod_{i=\ell+1}^{2\ell} \frac{1}{\norm{2}{\bw_i}^{d_i}}\bigg] &~=~ \frac{1}{D^{s_2/2}} \inparen{1 \pm C \cdot \ell^3 \cdot \bigg(\frac{s_2^2}{D} \bigg)} \label{eq:exp_ind2}
	\end{align}
	and
	\begin{align}
		\Var\bigg[ \prod_{i=1}^{\ell}       \frac{1}{\norm{2}{\bw_i}^{d_i}}\bigg] &~\le~ \frac{1}{D^{s_1}} \cdot 8C \cdot \ell^3 \cdot \bigg(\frac{s_1^2}{D} \bigg) \label{eq:var_ind1}\\
		\Var\bigg[ \prod_{i=\ell+1}^{2\ell} \frac{1}{\norm{2}{\bw_i}^{d_i}}\bigg] &~\le~ \frac{1}{D^{s_2}} \cdot 8C \cdot \ell^3 \cdot \bigg(\frac{s_2^2}{D} \bigg) \label{eq:var_ind2}
	\end{align}
	Plugging \Cref{eq:exp_ind1,eq:exp_ind2,eq:var_ind1,eq:var_ind2} in \Cref{eq:cov_manip}, it is not hard to see that,
	\begin{equation*}
		\Ex\bigg[\prod_{i=1}^{2\ell} \frac{1}{\norm{2}{\bw_i}^{d_i}} \bigg] = \frac{1}{D^{d/2}} \inparen{1 \pm C \cdot (2\ell)^3 \cdot \bigg(\frac{d^2}{D} \bigg)}\;.
	\end{equation*}
	This completes the proof of \Cref{eq:exp_prod_dep}. Now, \Cref{eq:exp_prod_dep_var} follows easily as,
	\begin{align*}
	\Var\insquare{\prod_{i=1}^{2\ell} \frac{1}{\norm{2}{\bw_i}^{d_i}}}
	&~=~ \Ex \insquare{\prod_{i=1}^{2\ell} \frac{1}{\norm{2}{\bw_i}^{2d_i}}} -  \Ex_{\bw}\insquare{\prod_{i=1}^{2\ell} \frac{1}{\norm{2}{\bw_i}^{d_i}}}^2\\ 
	&~=~ \inparen{\frac{1}{D^d} \pm C \cdot (2\ell)^3 \inparen{\frac{(2d)^2}{D^{d+1}}}} - \inparen{\frac{1}{D^{d/2}} \pm C \cdot (2\ell)^3 \cdot \inparen{\frac{d^2}{D^{d/2+1}}}}^2\\ 
	&~\le~ 8C \cdot (2\ell)^3 \cdot \inparen{\frac{d^2}{D^{d+1}}}\;.
	\end{align*}
\end{proof}

\subsubsection*{Interchanging Product and Expectation}

In order to handle the interchanging of the product and expectation operations, we will show the following lemma.

\begin{lem}\label{le:swap_exp_prod}
	Let $(\bu_1,\ldots,\bu_d, \bv_1,\ldots, \bv_d)$ be a multivariate Gaussian distribution such that,
	\begin{itemize}
		\item each of $\bu_i,\bv_i \in \bbR^D$ is marginally distributed as a $D$-dimensional standard Gaussian, i.e., $\gamma_D$.
		\item for each $j \in [D]$, the joint distribution of the $j$-th coordinates, i.e., $(u_{1,j}, \ldots, u_{d,j}, v_{1,j}, \ldots, v_{d,j})$, is independent across different values of $j$.
	\end{itemize}
	Then,
	\begin{equation*}
	\inabs{\Ex\limits_{\set{\bu_i, \bv_i}_i} \ \insquare{\prod_{i=1}^d \inangle{\bu_i, \bv_i}} - \prod_{i=1}^d \ \Ex\limits_{\set{\bu_i, \bv_i}_i} \ \insquare{\inangle{\bu_i, \bv_i}}} ~\le~ d^{O(d)} \cdot D^{d-1}.
	\end{equation*}
	%\change{Stress the explicitness of the $O(d)$ exponent?}
	%where $\lambda_d$ is an \emph{explicitly computable} function of $d$.
\end{lem}

\begin{remark}
	The $d^{O(d)}$ term has an explicit expression, although we only highlight its qualitative nature for clarity. Again, it is conceivable that the bounds in \Cref{le:swap_exp_prod} could be improved in terms of the dependence on $d$, although we suspect that it is tight upto constant factors in the exponent. Anyhow, this was not central to our application, so we go ahead with the stated bounds. The main point to note in the above lemma is that the exponent of $D$ is $(d-1)$ instead of $d$.
\end{remark}

\noindent To prove the lemma, we first obtain the following proposition on moments of a multivariate Gaussian.

\begin{proposition}\label{prop:mag_mom_multi_Gauss}
	Let $\bw \in \mathbb{R}^{\ell}$ be any multivariate Gaussian vector with each coordinate marginally distributed according to $\gamma_1$. Let $d_1, d_2, \dots, d_{\ell}$ be non-negative integers such that $d := \sum_{i=1}^{\ell} d_i$. Then,
	\begin{equation}
	\bigg|  \Ex \bigg[ \prod_{i=1}^{\ell} w_i^{d_i} \bigg] \bigg| \le (2d)^{3d}.
	\end{equation}
\end{proposition}
\begin{proof}
	More specifically we will show that when $\ell$ is a power of $2$,
	\begin{equation}\label{eq:mag_bd_mult_Gauss}
	\bigg|  \Ex \bigg[ \prod_{i=1}^{\ell} w_i^{d_i} \bigg] \bigg| \le 2^{\ell-1} (\ell d)^{d}.
	\end{equation}
	It is easy to see that this immediately implies the bound of $2^d \cdot d^{2d}$ in the main lemma, since $\ell \le d$. However if $\ell$ is not a power of $2$ we can round it up to the nearest power of $2$, which amounts to substituting $\ell \le 2d$ in the above, obtaining a bound of $2^{3d} \cdot d^{2d} \le (2d)^{3d}$.\\% losing the other factor of $2$. \change{Don't we get $2^{2d-2} (2 d^2)^d$ which is a bit larger than $(2d)^{2d}$?}\\
	
	\noindent We proceed by induction on $\ell$ (more specifically on $\log \ell$). For $\ell = 1$, we use the well-known fact that for $w \sim \gamma_1$,
	$$\inabs{\Ex[w^d]} ~=~ \inbrace{\inmat{0 & \sIF d \text{ is odd} \\ (d-1)!! & \sIF d \text{ is even}}} ~\le~ d^{d},$$
	where $(d-1)!!$ denotes the double factorial of $(d-1)$, i.e., the product of all integers from $1$ to $d-1$ that have the same parity as $d-1$. For the inductive step, we assume that the bound in (\ref{eq:mag_bd_mult_Gauss}) holds for $\ell$ and we show that it also holds for $2\ell$. For notational convenience, let $s_1 = \sum_{i=1}^{\ell} d_i$ and $s_2 = \sum_{i=\ell+1}^{2\ell} d_i$, and so $d = s_1 + s_2$.
	
	The main idea that we use to prove the inductive step is: for any two random variables $X$ and $Y$, it holds that $\Ex[XY] = \Ex[X] \Ex[Y] + \Cov[X,Y]$, and that the covariance satisfies $|\Cov[X,Y]| \le \sqrt{\Var[X] \cdot \Var[Y]}$ (by Cauchy-Schwarz inequality). Additionally, we use that $\Var[X] \le \Ex[X^2]$. Thus, we get,
	\begin{align*}
	&\bigg|  \Ex \bigg[ \prod_{i=1}^{2\ell} w_i^{d_i} \bigg] \bigg|%\\
	~=~ \bigg|  \Ex \bigg[ \prod_{i=1}^{\ell} w_i^{d_i} \cdot \prod_{i=\ell+1}^{2\ell} w_i^{d_i} \bigg] \bigg|\\
	&~=~ \bigg|  \Ex\bigg[ \prod_{i=1}^{\ell} w_i^{d_i} \bigg] \cdot \Ex \bigg[ \prod_{i=\ell+1}^{2\ell} w_i^{d_i} \bigg] ~+~ \Cov\bigg[ \prod_{i=1}^{\ell} w_i^{d_i}\ , \prod_{i=\ell+1}^{2\ell} w_i^{d_i}\bigg]  \bigg| \\ 
	&~\le~ \bigg|  \Ex\bigg[ \prod_{i=1}^{\ell} w_i^{d_i} \bigg] \cdot \Ex \bigg[ \prod_{i=\ell+1}^{2\ell} w_i^{d_i} \bigg] \bigg|  ~+~ \sqrt{\Var\bigg[ \prod_{i=1}^{\ell} w_i^{d_i} \bigg] \cdot \Var\bigg[ \prod_{i=\ell+1}^{2\ell} w_i^{d_i}\bigg]}\\
	&~\le~ \bigg|  \Ex\bigg[ \prod_{i=1}^{\ell} w_i^{d_i} \bigg] \cdot \Ex \bigg[ \prod_{i=\ell+1}^{2\ell} w_i^{d_i} \bigg] \bigg| ~+~ \sqrt{\Ex\bigg[ \prod_{i=1}^{\ell} w_i^{2d_i} \bigg] \cdot \Ex\bigg[ \prod_{i=\ell+1}^{2\ell} w_i^{2d_i}\bigg]}\\
	&~\le~ 2^{\ell-1} (\ell s_1)^{s_1} \cdot 2^{\ell-1} (\ell s_2)^{s_2} + \sqrt{2^{\ell - 1} (2\ell s_1)^{2s_1} \cdot 2^{\ell-1} (2\ell s_2)^{2s_2}}\\
	&~\le~ 2^{2\ell-1} (2\ell d)^{d}\;,
	\end{align*}
	where, the second last inequality uses the inductive assumption regarding product of $\ell$ terms. The last inequality follows from $s_1^{s_1} \cdot s_2^{s_2} \le d^{s_1} \cdot d^{s_2} = d^d$.
\end{proof}
%\change{Might be better to mention the general inequality that is used in the last step in the above sequence of inequalities..}
%Note that the bound in (\ref{eq:unorm_ineq_norms}) immediately follows from Lemma~\ref{le:swap_exp_prod}.
\noindent Using the above proposition, we are now able to prove \Cref{le:swap_exp_prod}.
\begin{proof}[Proof of \Cref{le:swap_exp_prod}]
	Let $S \subseteq [D]^d$ be the set of all tuples $\bc \in [D]^d$ such that $c_j \neq c_k$ for all $j \neq k \in [d]$. Let $\overline{S}$ denote the complement of $S$ in $[D]^d$. Note that $|\overline{S}| \le d^2 \cdot D^{d-1}$. We have that
	\begin{align}
	\Ex\bigg[ \prod_{i=1}^d \langle \bu_i, \bv_i \rangle \bigg] &~=~ \Ex\bigg[ \prod_{i=1}^d \sum_{k=1}^D u_{i,k} v_{i,k}  \bigg] \nonumber \\ 
	&~=~ \Ex \bigg[ \sum_{\bc \in [D]^d} \prod_{i=1}^d u_{i, c_i} v_{i, c_i} \bigg] \nonumber \\ 
	&~=~ \sum_{\bc \in [D]^d} \Ex \bigg[ \prod_{i=1}^d u_{i, c_i} v_{i, c_i} \bigg] \nonumber \\ 
	&~=~ \sum_{\bc \in S} \Ex \bigg[ \prod_{i=1}^d u_{i, c_i} v_{i, c_i} \bigg] + \sum_{\bc \in \overline{S}} \Ex \bigg[ \prod_{i=1}^d u_{i, c_i} v_{i, c_i} \bigg] \nonumber\\ 
	&~=~ \sum_{\bc \in S} \prod_{i=1}^d \Ex[ u_{i, c_i} v_{i, c_i}] + \sum_{\bc \in \overline{S}} \Ex \bigg[ \prod_{i=1}^d u_{i, c_i} v_{i, c_i} \bigg], \label{eq:exp_of_prod}
	\end{align}
	where the last equality follows from the assumption that the distribution of $(u_{1,j}, \ldots, u_{d,j}, v_{1,j}, \ldots, v_{d,j})$ is independent across $j \in [D]$. On the other hand, we have that
	\begin{align}
	\prod_{i=1}^d \ \Ex [ \langle \bu_i, \bv_i \rangle ] &~=~ \prod_{i=1}^d \ \Ex \bigg[ \sum_{k=1}^D u_{i,k} v_{i,k} \bigg] \nonumber\\
	&~=~ \sum_{\bc \in [D]^d} \prod_{i=1}^d \Ex[ u_{i, c_i} v_{i, c_i}] \nonumber\\
	&~=~ \sum_{\bc \in S} \prod_{i=1}^d \Ex[ u_{i, c_i} v_{i, c_i}] + \sum_{\bc \in \overline{S}} \prod_{i=1}^d \Ex[ u_{i, c_i} v_{i, c_i}] \label{eq:prod_of_exp}
	\end{align}
	Combining \Cref{eq:exp_of_prod,eq:prod_of_exp}, we get
	\begin{align*}
	\bigg| \Ex\bigg[ \prod_{i=1}^d \langle \bu_i, \bv_i \rangle \bigg] - \prod_{i=1}^d \ \Ex [ \langle \bu_i, \bv_i \rangle ] \bigg|
	&~=~  \bigg|  \sum_{\bc \in \overline{S}} \bigg( \Ex \bigg[ \prod_{i=1}^d u_{i, c_i} v_{i, c_i}\bigg] - \prod_{i=1}^d \Ex[ u_{i, c_i} v_{i, c_i}]  \bigg) \bigg|\\ 
	&~\le~ |\overline{S}| \cdot \max_{\bc \in \overline{S}} \bigg|  \Ex \bigg[ \prod_{i=1}^d u_{i, c_i} v_{i, c_i}\bigg] - \prod_{i=1}^d \Ex[ u_{i, c_i} v_{i, c_i}]  \bigg|\\ 
	&~\le~ d^2 \cdot D^{d-1} \cdot \inparen{(2d)^{3d} + 1}\\
	&~\le~ d^{O(d)} \cdot D^{d-1},
	\end{align*}
	where the second last inequality follows from the fact that $|\overline{S}| \le d^2 \cdot D^{d-1}$ and from \Cref{prop:mag_mom_multi_Gauss}.
\end{proof}

\subsubsection*{Putting things together to prove \Cref{le:norm_Gauss}}

\begin{proof}[Proof of \Cref{le:norm_Gauss}]
	We prove this lemma in two steps. We show the following two bounds,
	\begin{align}
	\inabs{\Ex\limits_{\set{\bu_i, \bv_i}_i} \ \insquare{\prod_{i=1}^d \frac{\inangle{\bu_i, \bv_i}}{\norm{2}{\bu_i} \norm{2}{\bv_i}}} - \Ex\limits_{\set{\bu_i, \bv_i}_i} \ \insquare{\prod_{i=1}^d \frac{\inangle{\bu_i, \bv_i}}{D}}} &~\le~ O\inparen{\frac{d^{2.5}}{\sqrt{D}} + \frac{d^{O(d)}}{D}}. \label{eq:norm_ineq_norms}\\
	\inabs{\Ex\limits_{\set{\bu_i, \bv_i}_i} \ \insquare{\prod_{i=1}^d \frac{\inangle{\bu_i, \bv_i}}{D}} - \prod_{i=1}^d \ \Ex\limits_{\set{\ba_i, \bb_i}_i} \ \insquare{\frac{\inangle{\bu_i, \bv_i}}{D}}} &~\le~ \frac{d^{O(d)}}{D}.\label{eq:unorm_ineq_norms}
	\end{align}
	\Cref{le:norm_Gauss} now follows immediately by the triangle inequality, for an explicit choice of $D$ that is upper bounded by $\frac{d^{O(d)}}{\delta^2}$.\\
	
	\noindent Note that \Cref{eq:unorm_ineq_norms} is simply a restatement of \Cref{le:swap_exp_prod}. To prove \Cref{eq:norm_ineq_norms}, we define the random variables
	\[ W ~:=~ \prod_{i=1}^d \inangle{\bu_i, \bv_i} \qquad \sAND \qquad Z ~:=~ \prod_{i=1}^d \frac{1}{\norm{2}{\bu_i} \norm{2}{\bv_i}} ~-~ \frac{1}{D^d}\;. \]
	Note that \Cref{eq:norm_ineq_norms} is equivalent to showing bounds on $|\Ex[W \cdot Z]|$. In order to do so, we show the following four bounds:
	\begin{enumerate}
		\item\label{it:W_mean} $\inabs{\Ex[W]} = D^d + d^{O(d)} \cdot D^{d-1}$.
		\item\label{it:W_var} $\Var[W] = D^{2d} + d^{O(d)} \cdot D^{2d-1}$.
		\item\label{it:Z_mean} $\inabs{\Ex[Z]} = O\inparen{\frac{d^5}{D^{d+1}}}$. 
		\item\label{it:Z_var} $\Var[Z] = O\inparen{\frac{d^5}{D^{2d+1}}}$. 
	\end{enumerate}
	We now prove each of these four bounds. We start by proving \Cref{it:W_mean}. By \Cref{le:swap_exp_prod}, we have that
	\begin{align*}
		|\Ex[W]|
		&~\le~ \bigg| \prod_{i=1}^d  \Ex [ \langle \bu_i, \bv_i \rangle ] \bigg| + d^{O(d)} \cdot D^{d-1}\\
		&~=~ \bigg| \prod_{i=1}^d \sum_{j=1}^D \Ex[\bu_{i,j} \bv_{i,j}] \bigg| + d^{O(d)} \cdot D^{d-1}\\
		&~\le~ D^d + d^{O(d)} \cdot D^{d-1}
	\end{align*}
	To prove \Cref{it:W_var}, we can again apply \Cref{le:swap_exp_prod} as follows:
	\begin{align*}
		\Var[W] &~\le~ \Ex[W^2]\\
		&~=~ \Ex \insquare{\prod_{i=1}^{d} \inangle{\bu_i, \bv_i}^2}\\
		&~\le~ D^{2d} + d^{O(d)} \cdot D^{2d-1}
		% \le \prod_{i=1}^d  \Ex [ \langle \bu_i, \bv_i \rangle ]^2 + \lambda_{2d} \cdot D^{2d-1} \le D^{2d} + \lambda_{2d} \cdot D^{2d-1} = O(D^{2d}).
	\end{align*}
	Finally, we note that \Cref{it:Z_mean,it:Z_var} follows exactly from \Cref{le:corr_norms_rec_ex}.\\
	
	\noindent To put this together, recall the definition of covariance $\Cov[W,Z] := \Ex[W \cdot Z] - \Ex[W] \cdot \Ex[Z]$, and that by Cauchy-Schwarz inequality, $|\Cov[W,Z]| \le \sqrt{\Var[W] \cdot \Var[Z]}$. Thus,
	\begin{align*}
	\inabs{\Ex[W \cdot Z]} &= \inabs{\Ex[W] \cdot \Ex[Z] + \Cov(W,Z)}\\ 
	&~\le~ \inabs{\Ex[W] \cdot \Ex[Z]} + \inabs{\Cov(W,Z)}\\ 
	&~\le~ \inabs{\Ex[W]} \cdot \inabs{\Ex[Z]} + \sqrt{\Var[W] \cdot \Var[Z]}\\
	&~\le~ O\inparen{\frac{d^{2.5}}{\sqrt{D}} + \frac{d^{O(d)}}{D}}
	\end{align*}
	where the last line combines \Cref{it:W_mean,it:W_var,it:Z_mean,it:Z_var}.
\end{proof}

\subsection{Mean \& Variance Bounds for Multilinear Monomials}\label{subsec:bds_for_multilinear_monom}

For the rest of this section, we simplify our notations as follows:
\begin{itemize}
	\item For vectors $\ba, \bb \in \bbR^D$, we will use $\wtilde{\ba}$ and $\wtilde{\bb}$ to denote the normalized vectors $\frac{\ba}{\|\ba\|_2}$ and $\frac{\bb}{\|\bb\|_2}$ respectively.
	\item We will use $\bU \in \bbR^N$ to denote $M\wtilde{\ba}$ and similarly $\bV \in \bbR^N$ to denote $M\wtilde{\bb}$. We will also have independent variables $\ba'$ and $\bb'$, for which we use $\bU' = M\wtilde{\ba}'$ and $\bV' = M\wtilde{\bb}'$.
	\item $U_i$ denotes the $i$-th coordinate of $U$. Similarly, $\m_i$ denotes the vector corresponding to the $i$-th row of $M$. Note that $U_i = \inangle{\m_i, \wtilde{\ba}}$.  For $S \subseteq [N]$, let $\bU_S$ denote $\prod_{i \in S} U_i = \prod_{i \in S} \inangle{\m_i, \wtilde{\ba}}$. Similarly for $\bV_S$.
	\item We will take expectations over random variables $M$, $\ba$, $\bb$, $\ba'$, $\bb'$. It will be understood that we are sampling $M \sim \calN(0,1)^{\otimes (N \times D)}$. Also, $(\ba, \bb)$ and $(\ba', \bb')$ are independently sampled from $\calG_{\rho}^{\otimes D}$.
\end{itemize}

\begin{lem}[Mean bounds for monomials]\label{lem:mean_bound_monomial}
	Given parameter $d$ and $\delta$, there exists an explicitly computable $D := D(d, \delta)$ such that the following holds:	
	For any subsets $S, T \subseteq [N]$ satisfying $|S|, |T| \le d$, it holds that,
	\begin{align*}
		\text{if } S \ne T : & \qquad \Ex_{M} \Ex_{\ba, \bb} \bU_S \bV_T ~=~ 0\,.\\
		\text{if } S = T : & \qquad \inabs{\Ex\limits_{M} \Ex\limits_{\ba, \bb} \bU_S \bV_T - \rho^{|S|}} ~\le~ \delta\,.
	\end{align*}
	In particular, one may take $D = \frac{d^{O(d)}}{\delta^2}$.
\end{lem}
\begin{proof}
	%Recall that $\m_1, \m_2, \dots, \m_N \in \mathbb{R}^D$ denotes the vectors corresponding to the rows of $M$.
	We have that
	\begin{align}\label{eq:mean_manip}
		\Ex_{M} \Ex_{\ba, \bb} \bU_S \bV_T &~=~ \Ex_{M} \Ex_{\ba, \bb} \bigg[ \prod_{i \in S} U_i  \cdot \prod_{i \in T} V_i \bigg]\nonumber\\ 
		&~=~ \Ex_{\ba, \bb} \Ex_{M} \bigg[  \prod_{i \in S \cap T} U_i V_i \cdot \prod_{i \in S \setminus T} U_i \cdot \prod_{i \in T \setminus S} V_i  \bigg]\nonumber\\ 
		&~=~ \Ex_{\ba, \bb} \Ex_{M} \bigg[  \prod_{i \in S \cap T} \inangle{\m_i, \wtilde{\ba}} \inangle{\m_i , \wtilde{\bb}} \cdot \prod_{i \in S \setminus T} \inangle{\m_i, \wtilde{\ba}} \cdot \prod_{i \in T \setminus S} \inangle{\m_i , \wtilde{\bb}} \bigg]\nonumber\\ 
		&~=~ \Ex_{\ba, \bb} \bigg[\prod_{i \in S \cap T} \Ex_{\m_i} \inangle{\m_i, \wtilde{\ba}} \inangle{\m_i , \wtilde{\bb}} \cdot \prod_{i \in S \setminus T} \Ex_{\m_i} \inangle{\m_i, \wtilde{\ba}} \cdot \prod_{i \in T \setminus S} \Ex_{\m_i} \inangle{\m_i , \wtilde{\bb}} \bigg],
	\end{align}
	where the last equality follows from the independence of the $\m_i$'s.\\
	
	\noindent If $S \neq T$, then either $S \setminus T$ is non-empty in which case $\prod_{i \in S \setminus T} \Ex_{\m_i}[\inangle{\m_i, \wtilde{\ba}}] = 0$ or $T \setminus S$ is non-empty in which case $\prod_{i \in T \setminus S} \Ex_{\m_i}[\inangle{\m_i , \wtilde{\bb}}] = 0$.  This is because for any fixed vector $\ba$ and for each $i \in [N]$, the random variable $\inangle{\m_i, \wtilde{\ba}}$ has zero-mean (and similarly for $\inangle{\m_i, \wtilde{\bb}}$). The first part of the lemma now follows from \Cref{eq:mean_manip}.\\
	
	\noindent If $S = T$, \Cref{eq:mean_manip} becomes
	\begin{align*}
		\Ex_{M} \Ex_{\ba, \bb} \bU_S \bV_T
		&~=~ \Ex_{\ba, \bb} \bigg[ \prod_{i \in S} \Ex_{\m_i}\frac{\langle \m_i, \ba \rangle}{\|\ba\|_2} \frac{\langle \m_i , \bb \rangle}{\|\bb\|_2} \bigg]\\
		&~=~ \Ex_{\ba, \bb} \bigg[\prod_{i \in S} \frac{ \langle \ba, \bb \rangle}{\|\ba\|_2 \|\bb\|_2} \bigg] \qquad \insquare{\text{since $\Ex_{\m_i} \m_i \cdot \m_i^T = I_{D \times D}$.}}\\
		&~=~ \prod_{i \in S} \bigg[ \frac{ \Ex_{\ba, \bb}\langle \ba, \bb \rangle}{D} \bigg] \pm \delta\\
		&~=~ \rho^{|S|} \pm \delta,
	\end{align*}
	where the penultimate equality above follows from \Cref{le:norm_Gauss} for an explicit choice of $D$ that is upper bounded by $\frac{d^{O(d)}}{\delta^2}$.
\end{proof}

\begin{lem}[Covariance bounds for monomials]\label{lem:var_bound_monomial}
	Given parameters $d$ and $\delta$, there exists an explicitly computable $D := D(d, \delta)$ such that the following holds:
	For any subsets $S, T, S', T' \subseteq [N]$ satisfying $|S|, |T|, |S'|, |T'| \le d$, it holds that,
	\begin{align*}
		\text{if } \ S \triangle T \triangle S' \triangle T' \ne \emptyset ~:~
		& \ \inabs{\Ex\limits_{M} \Ex\limits_{\ba, \bb} \Ex\limits_{\ba', \bb'} \insquare{\bU_S \bV_T \bU'_{S'} \bV'_{T'}} - \inparen{\Ex\limits_{M} \Ex\limits_{\ba, \bb} \insquare{\bU_S \bV_T}} \cdot \inparen{\Ex\limits_{M} \Ex\limits_{\ba', \bb'} \insquare{\bU'_{S'} \bV'_{T'}}}} ~=~ 0\;,\\
		\text{if }\ S \triangle T \triangle S' \triangle T' = \emptyset ~:~
		& \ \inabs{\Ex\limits_{M} \Ex\limits_{\ba, \bb} \Ex\limits_{\ba', \bb'} \insquare{\bU_S \bV_T \bU'_{S'} \bV'_{T'}} - \inparen{\Ex\limits_{M} \Ex\limits_{\ba, \bb} \insquare{\bU_S \bV_T}} \cdot \inparen{\Ex\limits_{M} \Ex\limits_{\ba', \bb'} \insquare{\bU'_{S'} \bV'_{T'}}}} ~\le~ \delta.%\frac{O_d(1)}{\sqrt{D}}\;.\\
	\end{align*}
	Here, $S \triangle T \triangle S' \triangle T'$ is the symmetric difference of the sets $S, T, S', T'$, equivalently, the set of all $i \in [N]$ which appear an odd number of times in the multiset $S \sqcup T \sqcup S' \sqcup T'$.\\
	In particular, one may take $D = \frac{d^{O(d)}}{\delta^2}$.
\end{lem}

\noindent In order to prove \Cref{lem:var_bound_monomial}, we need the following lemma.

\begin{lem}\label{le:square}
	%Let $(\ba, \bb) \sim \calG_\rho^{\otimes D}$ and $(\ba', \bb') \sim \calG_\rho^{\otimes D}$ be independently sampled $\rho$-correlated $D$-dimensional Gaussian vectors. Let $\tilde{\ba} := \frac{\ba}{\|\ba\|_2}$,  $\tilde{\bb} := \frac{\bb}{\|\bb\|_2}$, $\tilde{\ba'} := \frac{\ba'}{\|\ba'\|_2}$ and $\tilde{\bb'} := \frac{\bb'}{\|\bb'\|_2}$.
	Let $\m$ be distributed as $\calN(0,1)^{\otimes D}$. Then,
	\begin{equation*}
		\Ex_{\ba, \bb, \ba', \bb'}\bigg[ \bigg( \Ex_{\m}[ \inangle{ \m, \wtilde{\ba} } \inangle{ \m, \wtilde{\bb} } \inangle{ \m, \wtilde{\ba}' } \inangle{ \m, \wtilde{\bb}' }] - \Ex_{\m}[\inangle{ \m, \wtilde{\ba} } \inangle{ m , \wtilde{\bb} } ] \cdot \Ex_{\m}[\inangle{ m, \wtilde{\ba}' } \inangle{ \m, \wtilde{\bb}' }] \bigg)^2 \bigg] \le O\bigg(\frac{1}{D^2}\bigg)
	\end{equation*}
	and
	\begin{equation*}
		\Ex_{\ba, \ba'}\bigg[ \bigg( \Ex_{\m}[\inangle{ \m, \wtilde{\ba} } \inangle{ \m , \wtilde{\ba}' }] - \Ex_{\m}[\inangle{ \m, \wtilde{\ba} }] \cdot \Ex_{\m}[\inangle{ \m, \wtilde{\ba}' }]\bigg)^2
		\bigg] \le O\bigg(\frac{1}{D}\bigg).
	\end{equation*}
\end{lem}
\begin{proof}%[Proof of \Cref{le:square}]
	To prove the first part of the lemma, consider the quantity
	\begin{align*}
		T(\ba, \bb, \ba', \bb') &~:=~  \Ex_{\m}[ \inangle{ \m, \wtilde{\ba} } \inangle{ \m, \wtilde{\bb} } \inangle{ \m, \wtilde{\ba}' } \inangle{ \m, \wtilde{\bb}'}] - \Ex_{\m}[\inangle{ \m, \wtilde{\ba} } \inangle{ m , \wtilde{\bb} } ] \cdot \Ex_{\m}[\inangle{ m, \wtilde{\ba}' } \inangle{ \m, \wtilde{\bb}' }]\\ 
		&~=~ \inangle{ \wtilde{\ba} , \wtilde{\bb} } \inangle{ \wtilde{\ba}', \wtilde{\bb}' } + \inangle{ \wtilde{\ba} , \wtilde{\ba}' } \inangle{ \wtilde{\bb}, \wtilde{\bb}' } + \inangle{ \wtilde{\ba} , \wtilde{\bb}' } \inangle{ \wtilde{\ba}', \wtilde{\bb} } - \inangle{ \wtilde{\ba} , \wtilde{\bb} } \inangle{ \wtilde{\ba}' , \wtilde{\bb}' }\\ 
		&~=~ \inangle{ \wtilde{\ba} , \wtilde{\ba}' } \inangle{ \wtilde{\bb}, \wtilde{\bb}' } + \inangle{ \wtilde{\ba} , \wtilde{\bb}' } \inangle{ \wtilde{\ba}', \wtilde{\bb} }.
	\end{align*}
	where we use that for any $i \in [D]$, it holds that $\Ex_{\m} [m_i^4] = 3$ and $\Ex_{\m} [m_i^2] = 1$.
	Thus,
	\begin{align*}
		\Ex_{\ba, \bb, \ba', \bb'} \insquare{T(\ba, \bb, \ba', \bb')^2} &~=~\Ex_{\ba, \bb, \ba', \bb'} \insquare{ \insquare{\inangle{ \wtilde{\ba} , \wtilde{\ba}' } \inangle{ \wtilde{\bb}, \wtilde{\bb}' } + \inangle{ \wtilde{\ba} , \wtilde{\bb}' } \inangle{ \wtilde{\ba}', \wtilde{\bb} } }^2 }\\ 
		&~\le~ 2 \cdot \Ex_{\ba, \bb, \ba', \bb'} \insquare{ \inangle{ \wtilde{\ba} , \wtilde{\ba}' }^2 \inangle{ \wtilde{\bb}, \wtilde{\bb}' }^2} ~+~ 2 \cdot \Ex_{\ba, \bb, \ba', \bb'} \insquare{ \inangle{ \wtilde{\ba} , \wtilde{\bb}' }^2 \inangle{ \wtilde{\ba}', \wtilde{\bb} }^2 }\\ 
		&~\le~ O \inparen{\frac{1}{D^2}},
	\end{align*}
	where the last step follows by two applications of \Cref{le:norm_Gauss} (with $d = 4$). This completes the proof of the first part of the lemma. The second part of the lemma similarly follows from \Cref{le:norm_Gauss} (with $d = 2$) along with the fact that $\Ex_{\m}[\inangle{ \m, \wtilde{\ba} }] = 0 $.
\end{proof}

\begin{proof}[Proof of \Cref{lem:var_bound_monomial}]	
	Let $\mathbb{1}(E)$ denote the $0/1$ indicator function of an event $E$. We have that
	\begin{align}\label{eq:joint}
		\Ex\limits_{M} \Ex\limits_{\ba, \bb} \Ex\limits_{\ba', \bb'} \insquare{\bU_S \bV_T \bU'_{S'} \bV'_{T'}} &= \Ex\limits_{M} \Ex\limits_{\ba, \bb} \Ex\limits_{\ba', \bb'} \insquare{ \prod_{i \in S \cup T \cup S' \cup T'} U_i^{\mathbb{1}(i \in S)} V_i^{\mathbb{1}(i \in T)} {U'}_i^{\mathbb{1}(i \in S')} {V'}_i^{\mathbb{1}(i \in T')}} \nonumber \\ 
		&= \Ex\limits_{\ba, \bb} \Ex\limits_{\ba', \bb'} \insquare{ \prod_{i \in S \cup T \cup S' \cup T'} \Ex\limits_{\m_i} \insquare{ U_i^{\mathbb{1}(i \in S)} V_i^{\mathbb{1}(i \in T)} {U'}_i^{\mathbb{1}(i \in S')} {V'}_i^{\mathbb{1}(i \in T')}}}.
	\end{align}
	On the other hand, we have that
	\begin{align}\label{eq:first_ind}
		\Ex\limits_{M} \Ex\limits_{\ba, \bb} \insquare{\bU_S \bV_T} &= \Ex\limits_{\ba, \bb} \Ex\limits_{M} \insquare{ \prod_{i \in S \cup T} U_i^{\mathbb{1}(i \in S)} V_i^{\mathbb{1}(i \in T)}} = \Ex\limits_{\ba, \bb} \insquare{ \prod_{i \in S \cup T} \Ex_{\m_i} \insquare{U_i^{\mathbb{1}(i \in S)} V_i^{\mathbb{1}(i \in T)}}},
	\end{align}
	and similarly
	\begin{align}\label{eq:sec_ind}
		\Ex\limits_{M} \Ex\limits_{\ba', \bb'} \insquare{\bU'_{S'} \bV'_{T'}} &=  \Ex\limits_{\ba', \bb'} \insquare{ \prod_{i \in S' \cup T'} \Ex_{\m_i} \insquare{U_i^{\mathbb{1}(i \in S')} V_i^{\mathbb{1}(i \in T')}}}.
	\end{align}
	If there exists $i \in S \cup T \cup S' \cup T'$ that appears in an odd number of $S$, $T$, $S'$ and $T'$, then it can be seen that the expectation in \Cref{eq:joint} is equal to $0$, and that at least one of the expectations in \Cref{eq:first_ind,eq:sec_ind} is equal to $0$. This already handles the case that $S \triangle T \triangle S' \triangle T' \ne \emptyset$.
	
	Henceforth, we assume that each $i \in S \cup T \cup S' \cup T'$ appears in an even number of $S$, $T$, $S'$ and $T'$.
	Assume for ease of notation that $S \cup T \cup S' \cup T' \subseteq [4d]$. Define
	\begin{align}
	g_i(\ba, \bb, \ba', \bb') &~:=~ \Ex_{\m_i} \insquare{U_i^{\mathbb{1}(i \in S)} V_i^{\mathbb{1}(i \in T)} {U'}_i^{\mathbb{1}(i \in S')} {V'}_i^{\mathbb{1}(i \in T')}} \label{eq:g_function}\\
	h_i(\ba, \bb) &~:=~ \Ex_{\m_i} \insquare{U_i^{\mathbb{1}(i \in S)} V_i^{\mathbb{1}(i \in T)}}.\label{eq:h_function}\\
	h_i'(\ba', \bb') &~:=~ \Ex_{\m_i} \insquare{U_i'^{\mathbb{1}(i \in S')} V_i'^{\mathbb{1}(i \in T')}}.\label{eq:h'_function}
	\end{align}
	Combining \Cref{eq:joint,eq:first_ind,eq:sec_ind} along with the definitions in \Cref{eq:g_function,eq:h_function,eq:h'_function}, we get
	\begin{align*}
		& \bigg| \Ex\limits_{M} \Ex\limits_{\ba, \bb} \Ex\limits_{\ba', \bb'} \insquare{\bU_S \bV_T \bU'_{S'} \bV'_{T'}} - \Ex\limits_{M} \Ex\limits_{\ba, \bb} \insquare{\bU_S \bV_T} \cdot \Ex\limits_{M} \Ex\limits_{\ba', \bb'} \insquare{\bU'_{S'} \bV'_{T'}} \bigg|\\ 
		&~=~ \bigg|\Ex\limits_{\ba, \bb} \Ex\limits_{\ba', \bb'} \insquare{ \prod_{i = 1}^{4 d} g_i(\ba, \bb, \ba', \bb') - \prod_{i = 1}^{4 d} h_i(\ba, \bb) \cdot h_i'(\ba', \bb')}\bigg|\\ 
		&~=~ \bigg|\Ex\limits_{\ba, \bb} \Ex\limits_{\ba', \bb'} \insquare{ \sum_{j=1}^{4d} \insquare{ \prod_{i = 1}^{j-1} h_i(\ba, \bb) \cdot h_i'(\ba', \bb') \prod_{i = j}^{4d} g_i(\ba, \bb, \ba', \bb') - \prod_{i = 1}^{j} h_i(\ba, \bb) \cdot h_i'(\ba', \bb') \prod_{i = j+1}^{4d} g_i(\ba, \bb, \ba', \bb')}}\bigg|\\ 
		&~\le~ \sum_{j=1}^{4d} \ \bigg|\Ex\limits_{\ba, \bb} \Ex\limits_{\ba', \bb'} \insquare{ \prod_{i = 1}^{j-1} h_i(\ba, \bb) \cdot h_i'(\ba', \bb') \prod_{i = j+1}^{4d} g_i(\ba, \bb, \ba', \bb') \cdot \insquare{g_j(\ba, \bb, \ba', \bb') - h_j(\ba, \bb) \cdot h_j'(\ba', \bb')}}\bigg|\\ 
		&~\le~ 4 \cdot d \cdot \sqrt{\tau \cdot \kappa},
	\end{align*}
	where the last inequality follows from the Cauchy-Schwarz inequality with
	\begin{align*}
		\tau &~:=~ \max_{j \in [4d]} ~\Ex\limits_{\ba, \bb} ~\Ex\limits_{\ba', \bb'} ~\insquare{
			\prod_{i = 1}^{j-1} h_i(\ba, \bb)^2 \cdot h_i(\ba', \bb')^2 \prod_{i = j+1}^{4d} g_i(\ba, \bb, \ba', \bb')^2 }\\
		\kappa &~:=~ \max_{j \in [4d]} ~\Ex\limits_{\ba, \bb} ~\Ex\limits_{\ba', \bb'} ~\insquare{g_j(\ba, \bb, \ba', \bb') - h_j(\ba, \bb) \cdot h_j(\ba', \bb')}^2
	\end{align*}
	\Cref{le:square} implies that $\kappa \le O(1/D)$. We now show that $\tau \le 2^{O(d)}$. Note that for any $i \in [D]$, it holds that,
	\[
		h_i(\ba, \bb) ~=~ \infork{\inangle{\wtilde{\ba}, \wtilde{\bb}} & \sIF i \in S \sAND i \in T\\ 1 & \sIF i \notin S \sAND i \notin T \\ 0 & \text{otherwise}}
		\qquad \sAND \qquad
		h_i'(\ba', \bb') ~=~ \infork{\inangle{\wtilde{\ba}', \wtilde{\bb}'} & \sIF i \in S' \sAND i \in T'\\ 1 & \sIF i \notin S' \sAND i \notin T' \\ 0 & \text{otherwise}}
	\]
	\[
	g_i(\ba, \bb, \ba', \bb') ~=~ \infork{
		\inangle{\wtilde{\ba}, \wtilde{\bb}}\inangle{\wtilde{\ba}', \wtilde{\bb}'} + \inangle{\wtilde{\ba}, \wtilde{\ba}'}\inangle{\wtilde{\bb}, \wtilde{\bb}'} + \inangle{\wtilde{\ba}, \wtilde{\bb}'}\inangle{\wtilde{\ba}', \wtilde{\bb}} & \sIF i \in S \cap T \cap S' \cap T' \\
		\inangle{\wtilde{\ba}, \wtilde{\bb}} & \sIF i \in S \cap T \sAND i \notin S' \cup T' \\
		\inangle{\wtilde{\ba}, \wtilde{\ba}'} & \sIF i \in S \cap S' \sAND i \notin T \cup T' \\
		\inangle{\wtilde{\ba}, \wtilde{\bb}'} & \sIF i \in S \cap T' \sAND i \notin S' \cup T \\
		\inangle{\wtilde{\ba}', \wtilde{\bb}} & \sIF i \in S' \cap T \sAND i \notin S \cup T' \\
		\inangle{\wtilde{\ba}', \wtilde{\bb}'} & \sIF i \in S' \cap T' \sAND i \notin S \cup T \\
		\inangle{\wtilde{\bb}, \wtilde{\bb}'} & \sIF i \in T \cap T' \sAND i \notin S \cup S' \\
		1 & \text{ otherwise}
		}
	\]
	Thus, if we expand out a single term $\prod_{i = 1}^{j-1} h_i(\ba, \bb)^2 \cdot h_i(\ba', \bb')^2 \prod_{i = j+1}^{4d} g_i(\ba, \bb, \ba', \bb')^2$, we get at most $3^{8d}$ terms (since each $g_i$ can multiply the number of terms by at most $3$). Each of these terms is the expectation of the product of inner product of some correlated Gaussian vectors. We thus have from \Cref{le:norm_Gauss} that each such term is at most $1 + \delta$. Thus, we have that $\tau \le 2^{O(d)}$. For an explicit choice of $D$ that is upper bounded by $d^{O(d)}/\delta^2$, we get that $4 d \sqrt{\tau \kappa} \le \delta$, which concludes the proof of the lemma.
\end{proof}

\subsection{Mean \& Variance Bounds for Multilinear Polynomials}\label{subec:bds_for_gen_multilinear_polys}

We are now ready to prove \Cref{lem:mean_var_bound}. Recall that,
\[ F(M) = \Ex\limits_{\ba, \bb} \insquare{A(\bU) \cdot B(\bV)} \qquad \text{where, }\ \bU = \frac{M\ba}{\norm{2}{\ba}} \ \sAND \ \bV = \frac{M\bb}{\norm{2}{\bb}}.\]
We wish to bound the mean and variance of $F(M)$. These proofs work by considering the Hermite expansions of $A$ and $B$ given by,
\[A(\bX) = \sum_{S \subseteq [N]} \what{A}_S \bX_S \qquad \sAND \qquad B(\bX) = \sum_{T \subseteq [N]} \what{B}_T \bY_T.\]
The basic definitions and facts related to Hermite polynomials were given in \Cref{sec:prelim}.

\begin{proofof}{\Cref{lem:mean_var_bound}}
	We start out by proving the bound on $\inabs{\Ex_{M} F(M) - \inangle{A, B}_{\calG_{\rho}^{\otimes N}}}$. To this end, we will use \Cref{lem:mean_bound_monomial} with parameters $d$ and $\delta$. Thus, for a choice of $D = d^{O(d)}/\delta^2$, we have that,
	\begin{align*}
		& \inabs{\Ex_{M} F(M) - \inangle{A, B}_{\calG_{\rho}^{\otimes N}}}\\
		&~=~ \inabs{\Ex\limits_M \ \Ex\limits_{\ba, \bb} \insquare{A(\bU) \cdot B(\bV)} - \Ex\limits_{\bX, \bY \sim \calG_{\rho}^{\otimes N}} \insquare{A(\bX) \cdot B(\bY)}}\\
		&~=~ \inabs{\sum\limits_{S, T \subseteq [N]} \what{A}_S \what{B}_T \cdot \inparen{\Ex\limits_M \ \Ex\limits_{\ba, \bb} \insquare{\bU_S \cdot \bV_T} - \Ex_{\bX, \bY \sim \calG_{\rho}^{\otimes N}} \insquare{\bX_S \cdot \bY_T}}}\\
		&~=~ \inabs{\sum\limits_{S \subseteq [N]} \what{A}_S \what{B}_S \cdot \inparen{\Ex\limits_M \ \Ex\limits_{\ba, \bb} \insquare{\bU_S \cdot \bV_S} - \rho^{|S|}}} \qquad \ldots \ \text{(terms corresponding to $S \ne T$ are $0$.)}\\
		&~\le~ \sum\limits_{S \subseteq [N]} \inabs{\what{A}_S \what{B}_S} \cdot \delta \qquad \qquad \ldots \ldots \ \text{(using \Cref{lem:mean_bound_monomial})}\\
		&~\le~ \norm{2}{A} \cdot \norm{2}{B} \cdot \delta \qquad \qquad \ldots \ldots \ \text{(Cauchy-Schwarz inequality)}\\
		&~\le~ \delta \qquad \qquad \qquad \qquad \qquad \ldots \ldots \ \text{($\norm{2}{A}, \norm{2}{B} \le 1$)}\qedhere
	\end{align*}
	We now move to proving the bound on $\Var_M(F(M))$. To this end, we will use \Cref{lem:var_bound_monomial} with parameters $d$ and $\delta/9^d$. Thus, for a choice of $D = d^{O(d)}/\delta^2$, we have that,
	\begin{align*}
		& \Ex_M \inparen{\Ex_{\ba, \bb} A(\bU) \cdot B(\bV)}^2 - \inparen{\Ex_M \Ex_{\ba, \bb} A(\bU) \cdot B(\bV)}^2\\
		& ~=~\inabs{\Ex\limits_{M} \Ex\limits_{\ba, \bb} \Ex\limits_{\ba', \bb'} \insquare{A(\bU) B(\bV) A(\bU') B(\bV')} - \inparen{\Ex\limits_{M} \Ex\limits_{\ba, \bb} \insquare{A(\bU) B(\bV)}} \cdot \inparen{\Ex\limits_{M} \Ex\limits_{\ba', \bb'} \insquare{A(\bU') B(\bV')}}}\\
		& ~\le~ \sum_{\substack{S, T \subseteq [N] \\ S', T' \subseteq [N]}} \inabs{\what{A}_S \what{B}_T \what{A}_{S'} \what{B}_{T'}} \cdot \inabs{\Ex\limits_{M} \Ex\limits_{\ba, \bb} \Ex\limits_{\ba', \bb'} \insquare{\bU_S \bV_T \bU'_{S'} \bV'_{T'}} - \inparen{\Ex\limits_{M} \Ex\limits_{\ba, \bb} \insquare{\bU_S \bV_T}} \cdot \inparen{\Ex\limits_{M} \Ex\limits_{\ba', \bb'} \insquare{\bU'_{S'} \bV'_{T'}}}}\\
		& ~\le~ \frac{\delta}{9^d} \cdot \sum_{\substack{S, T, S', T' \subseteq [N] \\ S \triangle T \triangle S' \triangle T' = \emptyset}} \inabs{\what{A}_S \what{B}_T \what{A}_{S'} \what{B}_{T'}}\;.
	\end{align*}
	To finish the proof, we will show that,
	\[\sum_{\substack{S, T, S', T' \subseteq [N] \\ S \triangle T \triangle S' \triangle T' = \emptyset}} \inabs{\what{A}_S \what{B}_T \what{A}_{S'} \what{B}_{T'}} ~\le~ 9^d \cdot \norm{2}{A}^2 \cdot \norm{2}{B}^2\;.\]
	Define functions $f : \sbit^N \to \bbR$, $g : \sbit^N \to \bbR$ over the boolean hypercube as,
	\[f(x) = \sum_{\substack{S \subseteq [N] \\ |S| \le d}} \what{A}_S \calX_S(x) \quad \sAND \quad g(x) = \sum_{\substack{S \subseteq [N] \\ |S| \le d}} \what{B}_S \calX_S(x)\;.\]
	Hypercontractivity bounds \cite{wolff2007hypercontractivity} for degree-$d$ polynomials over the boolean hypercube imply that,
	\[\Ex_x \insquare{f(x)^4} \le 9^d \inparen{\Ex_x \insquare{f(x)^2}}^2 \quad \sAND \quad \Ex_x \insquare{g(x)^4} \le 9^d \inparen{\Ex_x \insquare{g(x)^2}}^2\;.\]
	We now finish the proof as follows,
	\begin{align*}
		\sum_{\substack{S, T, S', T' \subseteq [N] \\ S \triangle T \triangle S' \triangle T' = \emptyset}} \inabs{\what{A}_S \what{B}_T \what{A}_{S'} \what{B}_{T'}}
		&~=~ \Ex_{x} \insquare{f(x)^2 g(x)^2}\\
		&~\le~ \inparen{\Ex_x \insquare{f(x)^4}}^{1/2} \cdot \inparen{\Ex_x \insquare{g(x)^4}}^{1/2}\\
		&~\le~ 9^d \cdot \inparen{\Ex_x \insquare{f(x)^2}} \cdot \inparen{\Ex_x \insquare{g(x)^2}}\\
		&~=~ 9^d \cdot \inparen{\sum_S \what{A}_S^2} \cdot \inparen{\sum_S \what{B}_S^2}\\
		&~=~ 9^d \cdot \norm{2}{A}^2 \cdot \norm{2}{B}^2.
	\end{align*}
	Thus, overall we get that, $\Var_M(F(M)) \le \delta$.\\
	
	\noindent This completes the proof of \Cref{lem:mean_var_bound} for an explicit choice of $D$ that is upper bounded by $d^{O(d)}/\delta^2$.
\end{proofof}

%% file: apx_smoothing.tex
\section{Proof of Main Smoothing Lemma} \label{apx:smoothing}

\noindent In order to prove \Cref{lem:smoothing_main}, we consider the definition of {\em low-degree truncation}.

\begin{defn}[Low-degree truncation]\label{def:low_deg_truncation}
	We define this for functions in $L^2(\calZ^n, \mu^{\otimes n})$ and also for those in $L^2(\bbR^n, \gamma_n)$.\\
	
	\noindent {\bf Discrete Hypercube:} Suppose $A \in L^2(\calZ^n, \mu^{\otimes n})$ is given by the Fourier expansion $A(\bx) = \sum\limits_{\bsigma \in \bbZ_q^n} \what{A}_{\bsigma} \calX_{\bsigma}(\bx)$. The {\em degree-$d$ truncation} of $A$ is defined as the function $A^{\le d} \in L^2(\calZ^n, \mu^{\otimes n})$ given by
	\[A^{\le d}(\bx) ~:=~ \sum_{\substack{\bsigma \in \bbZ_q^n \\ |\bsigma| \le d}} \what{A}_{\bsigma} \calX_{\bsigma}(\bx).\]
	That is, $A^{\le d}$ is obtained by retaining only the terms with degree at most $d$ in the Fourier expansion of $A$, where recall that for $\bsigma \in \bbZ_q^n$, its degree is defined as $|\bsigma| = \setdef{i \in [n]}{\bsigma_i \ne 0}$.\\
	
	\noindent {\bf Gaussian:} Suppose $A \in L^2(\bbR^n, \gamma_n)$ is given by the Hermite expansion $A(\bX) = \sum_{\bsigma \in \bbZ_{\ge 0}^n} \what{A}_\bsigma H_{\bsigma}(\bX)$. The {\em degree-$d$ truncation} of $A$ is defined as the function $A^{\le d} \in L^2(\bbR^n, \gamma_n)$ given by
	\[A^{\le d}(\bX) ~:=~ \sum_{\substack{\bsigma \in \bbZ_{\ge 0}^n \\ |\bsigma| \le d}} \what{A}_{\bsigma} H_{\bsigma}(\bX).\]
	That is, $A^{\le d}$ is obtained by retaining only the terms with degree at most $d$ in the Hermite expansion of $A$, where recall that for $\bsigma \in \bbZ_{\ge 0}^n$, its degree is defined as $|\bsigma| = \sum_{i=1}^{n} \sigma_i$.\\
	
	\noindent For convenience, in either case, define $A^{> d} := A - A^{\le d}$. Also, for vector valued functions $A$, we define $A^{\le d}$ as the function obtained by applying the above low-degree truncation on each coordinate.
\end{defn}
To prove the discrete part of \Cref{lem:smoothing_main}, we will use a lemma from \cite[Lemma 6.1]{mossel2010gaussianbounds}, which is proved using Efron-Stein decompositions. To state this lemma, we first introduce the Bonami-Beckner operator.

\begin{defn}[Bonami-Beckner operator]
	For any $\nu \in [0,1]$, the Bonami-Beckner operator $T_{\nu}$ on a probability space $(\calZ, \mu)$ is given by its action on any $f : \calZ \to \bbR$, as follows,
	$$(T_\nu f)(x) = \Ex[f(Y) | X = x]$$
	where the conditional distribution of \ $Y$ given $X = x$ is $\nu \delta_x + (1-\nu) \mu$ where $\delta_x$ is the delta measure on $x$. In other words, given $X = x$, we obtain $Y$ by either setting it to $x$ with probability $\nu$ or independently sampling from $\mu$ with probability $(1-\nu)$.
	
	For the product space $(\calZ^n, \mu^{\otimes n})$, we define the Bonami-Beckner operator $T_{\nu}$ as, $T_{\nu} = \otimes_{i=1}^n T_{\nu}^{(i)}$, where $T_{\nu}^{(i)}$ is the Bonami-Beckner operator on the $i$-th coordinate $(\calZ, \mu)$.
\end{defn}

\noindent We now state a specialized version of Mossel's lemma, which suffices for our application.

\begin{lem}[\cite{mossel2010gaussianbounds}] \label{lem:mossel_smoothing}
	Let $(\calZ \times \calZ, \mu)$ be finite joint probability space, such that $\rho(\calZ, \calZ; \mu) = \rho$ for some $\rho \in [0,1]$. Let $P \in L^2(\calZ^n, \mu_A^{\otimes n})$ and $Q \in L^2(\calZ^n, \mu_B^{\otimes n})$ be multi-linear polynomials. Let $\eps > 0$ and $\nu$ be chosen sufficiently close to $1$ so that,
	\[\nu \ge (1-\eps)^{\log \rho/(\log \eps + \log \rho)}\]
	Then:
	\[\inabs{\inangle{P, Q}_{\mu^{\otimes n}} - \inangle{T_{\nu}P, T_{\nu}Q}_{\mu^{\otimes n}}} \le \eps \cdot \sqrt{\Var[P] \Var[Q]}\]
	In particular, there exists an absolute constant $C$ such that it suffices to take
	$$\nu ~\defeq~ 1 - C \frac{(1-\rho)\eps}{\log (1/\eps)}$$
\end{lem}

To prove the Gaussian version of \Cref{lem:smoothing_main}, we will need the analog of the above lemma for correlated Gaussian spaces which can be proved in a similar way, using Hermite expansions instead of Efron-Stein decompositions. Here, we use the Ornstein-Uhlenbeck operator $U_\nu$ instead of the Bonami-Beckner operator $T_\nu$. In particular, the following lemma holds.

\begin{lem} \label{lem:gaussian_smoothing}
	Consider the correlated Gaussian space $\calG_\rho^{\otimes n}$ for some $\rho \in [0,1]$.
	Let $P \in L^2(\bbR^n, \gamma_{n})$ and $Q \in L^2(\bbR^n, \gamma_{n})$. Let $\eps > 0$ and $\nu$ be chosen sufficiently close to $1$ so that,
	\[\nu \ge (1-\eps)^{\log \rho/(\log \eps + \log \rho)}\]
	Then:
	\[\inabs{\inangle{P, Q}_{\calG_\rho^{\otimes n}} - \inangle{U_{\nu}P, U_{\nu}Q}_{\calG_\rho^{\otimes n}}} \le \eps \cdot \sqrt{\Var[P] \Var[Q]}\]
	In particular, there exists an absolute constant $C$ such that it suffices to take
	$$\nu ~\defeq~ 1 - C \frac{(1-\rho)\eps}{\log (1/\eps)}$$
\end{lem}
\begin{proof}%[Proof of \Cref{lem:gaussian_smoothing}]
	Consider the Hermite expansions of $P$ and $Q$. That is,
	\[ P(\bX) ~=~ \sum_{\bsigma \in \bbZ_{\ge 0}^n} \what{P}(\bsigma) H_{\bsigma}(\bX) \qquad \sAND \qquad Q(\bY) ~=~ \sum_{\bsigma \in \bbZ_{\ge 0}^n} \what{Q}(\bsigma) H_{\bsigma}(\bY).\]
	Using properties of Hermite polynomials, namely, $U_{\nu} H_{\bsigma} = \nu^{|\bsigma|} H_{\bsigma}$, we get that,
	\[ U_\nu P(\bX) ~=~ \sum_{\bsigma \in \bbZ_{\ge 0}^n} \nu^{|\bsigma|} \what{P}(\bsigma) H_{\bsigma}(\bX) \qquad \sAND \qquad U_\nu Q(\bY) ~=~ \sum_{\bsigma \in \bbZ_{\ge 0}^n} \nu^{|\bsigma|} \what{Q}(\bsigma) H_{\bsigma}(\bY).\]
	Note that our choice of $\nu$ gives us that, $\rho^{d} \inparen{1 - \nu^{2d}} \le \eps$ for all $d \in \bbN$. Thus, we get that,
	\begin{align*}
	&\inabs{\inangle{P, Q}_{\calG_\rho^{\otimes n}} - \inangle{U_{\nu}P, U_{\nu}Q}_{\calG_\rho^{\otimes n}}}\\
	&~=~ \inabs{\sum\limits_{\bsigma \in \bbZ_{\ge 0}^n} \rho^{|\bsigma|} \cdot \what{P}(\bsigma) \what{Q}(\bsigma) \cdot \inparen{1 - \nu^{2|\bsigma|}}}\\
	&~\le~ \sum\limits_{\bsigma \in \bbZ_{\ge 0}^n} \inabs{\what{P}(\bsigma) \what{Q}(\bsigma)} \cdot \rho^{|\bsigma|} \inparen{1 - \nu^{2|\bsigma|}}\\
	%&~\le~ \sum\limits_{\bsigma \in \bbZ_{\ge 0}^n} \inabs{\what{P}_{\bsigma} \what{Q}_{\bsigma}} \cdot \min\set{\rho^{|\bsigma|}, \inparen{1 - \nu^{2|\bsigma|}}}\\
	&~\le~ \eps \cdot \sum\limits_{\bsigma \in \bbZ_{\ge 0}^n} \inabs{\what{P}(\bsigma) \what{Q}(\bsigma)} \qquad \qquad \ldots (\text{since, } \rho^{d} \inparen{1 - \nu^{2d}} \le \eps \text{ for all } d \in \bbN)\\
	&~\le~ \eps \cdot \sqrt{\Var[P] \Var[Q]} \qquad \qquad \ldots (\text{Cauchy-Schwarz inequality})
	\end{align*}
	%In the last line, we use that for any $\bsigma \ne 0^n$, it holds that $\rho^{|\bsigma|} \inparen{1 - \nu^{2|\bsigma|}} \le \eps$ for the said choice of $\nu$. This is because our choice of $\nu$ was such that at least one of $\rho^{|\bsigma|}$ or $\inparen{1 - \nu^{2|\bsigma|}}$ is at most $\eps$, whereas both are always upper bounded by $1$.
\end{proof}

\noindent We are now ready to prove our main smoothing lemma (\Cref{lem:smoothing_main}).

\begin{proof}[Proof of \Cref{lem:smoothing_main}]
	We prove the lemma for the case of correlated discrete hypercubes, i.e. $(\calZ \times \calZ, \mu)$. The proof for the correlated Gaussian case follows similarly.
	
	We obtain $A^{(1)}$ and $B^{(1)}$ in two steps. In the first step we apply some suitable amount of noise to the functions such that the functions have decaying Fourier tails. In the second step, we truncate the Fourier coefficients corresponding to terms larger than degree $d$.\\
	
	\noindent {\bf Noising step.} In this step, we obtain intermediate functions $\overline{A} : \calZ^n \to \bbR^k$ and $\overline{B} : \calZ^n \to \bbR^k$ such that,
	\begin{enumerate}
		\item $\overline{A}$ and $\overline{B}$ have decaying Fourier tails. In particular, for any $j \in [k]$ : $\norm{2}{\overline{A}_j^{>d}}, \norm{2}{\overline{B}_j^{>d}} \le \frac{\delta}{2 \sqrt{k}}$.
		\item $\Var(\overline{A}_j) \le \Var(A_j)$ and $\Var(\overline{B}_j) \le \Var(B_j)$, for any $j \in [k]$.
		\item $\norm{2}{\calR(\overline{A}) - \overline{A}} \le \norm{2}{\calR(A) - A}$ and $\norm{2}{\calR(\overline{B}) - \overline{B}} \le \norm{2}{\calR(B) - B}$.
		\item For every $i, j \in [k]$ : $\inabs{\inangle{\overline{A}_i, \overline{B}_j}_{\mu^{\otimes n}} - \inangle{A_i, B_j}_{\mu^{\otimes n}}} ~\le~ \frac{\delta}{2 \sqrt{k}}$.
	\end{enumerate}
	
	Firstly, note that we have $\Var[A_j], \Var[B_j] \le 1$ for any $j \in [k]$. Given parameter $\delta$, we first choose $\eps$ and $\nu$ in Lemma~\ref{lem:mossel_smoothing}, such that $\eps = \frac{\delta}{2\sqrt{k}}$ and then $\nu = 1 - C \frac{(1-\rho)\eps}{\log (1/\eps)}$ as required. We choose $d$ to be large enough such that $\nu^{2d} \le \frac{\delta}{2\sqrt{k}}$, that is, $d = O\inparen{\frac{\log (k/\delta)}{\log (1/\nu)}} = O\inparen{\frac{\sqrt{k}\log^2(k/\delta)}{\delta (1-\rho)}}$.
	
	Let $\overline{A} = T_{\nu}A$ and $\overline{B} = T_{\nu}B$. We get the above statements as follows,
	\begin{enumerate}
		\item $\norm{2}{\overline{A}_j^{> d}} ~=~ \sum\limits_{\substack{\bsigma \in \bbZ_q^n \\ |\bsigma| > d}} \nu^{2|\bsigma|} \cdot \what{A_j}(\bsigma)^2 ~\le~ \nu^{2d} \cdot \Var(A_j) ~\le~ \frac{\delta}{2\sqrt{k}}$. Similarly, $\norm{2}{\overline{B}_j^{> d}} ~\le~ \frac{\delta}{2\sqrt{k}}$.
		
		\item $\Var(\overline{A}_j) ~=~ \sum\limits_{\bsigma \in \bbZ_q^n} \nu^{2|\bsigma|} \cdot \what{A_j}(\bsigma)^2 ~\le~ \Var(A_j)$. Similarly, $\Var(\overline{B}) \le \Var(B)$.
		
		\item Observe that $\norm{2}{\calR(v) - v}$ is the Euclidean distance of $v \in \bbR^k$ from the simplex $\Delta_k$, which is a convex body. Hence $\norm{2}{\calR(v) - v}^2$ is convex function in $v$. Thus, we have that,
		\begin{align*}
		& \norm{2}{\calR(\overline{A}) - \overline{A}}^2\\
		&~=~\Ex\limits_{\bx \sim \mu_A^{\otimes n}} \norm{2}{\calR(\overline{A}(\bx)) - \overline{A}(x)}^2\\
		&~=~ \Ex\limits_{\bx \sim \mu_A^{\otimes n}} \norm{2}{\calR\inparen{\Ex\limits_{\bx' \sim T_\nu(\bx)} A(\bx')} - \Ex\limits_{\bx' \sim T_\nu(\bx)} A(\bx')}^2\\
		&~\le~ \Ex\limits_{\bx \sim \mu_A^{\otimes n}} \Ex\limits_{\bx' \sim T_\nu(\bx)} \norm{2}{\calR\inparen{A(\bx')} - A(\bx')}^2 \qquad \qquad \ldots(\text{using convexity of } \norm{2}{\calR(v) - v}^2)\\
		&~=~ \Ex\limits_{\bx' \sim \mu_A^{\otimes n}} \norm{2}{\calR\inparen{A(\bx')} - A(\bx')}^2\\
		&~=~ \norm{2}{\calR(A) - A}^2\;.
		\end{align*}
		Similar argument holds for $\overline{B}$.
		
		\item For every $i, j \in [k]$, we simply have from \Cref{lem:mossel_smoothing} that
		$\inabs{\inangle{\overline{A}_i, \overline{B}_j}_{\mu^{\otimes n}} - \inangle{A_i, B_j}_{\mu^{\otimes n}}} ~\le~ \eps ~=~ \frac{\delta}{2 \sqrt{k}}$.
	\end{enumerate}
	
	\noindent {\bf Low-degree truncation step.} In this step, we obtain the final $A^{(1)}$ and $B^{(1)}$ such that,
	\begin{enumerate}
		\item $A^{(1)}$ and $B^{(1)}$ have degree at most $d$.
		\item $\Var(A^{(1)}) \le \Var(\overline{A})$ and $\Var(B^{(1)}) \le \Var(\overline{B})$.
		\item $\norm{2}{\calR(A^{(1)}) - A^{(1)}} \le \norm{2}{\calR(\overline{A}) - \overline{A}} + \delta/2$ and $\norm{2}{\calR(B^{(1)}) - B^{(1)}} \le \norm{2}{\calR(\overline{B}) - \overline{B}} + \delta/2$
		\item For every $i, j \in [k]$ : $\inabs{\inangle{A^{(1)}_i, B^{(1)}_j}_{\mu^{\otimes n}} - \inangle{\overline{A}_i, \overline{B}_j}_{\mu^{\otimes n}}} ~\le~ \frac{\delta}{2 \sqrt{k}}$
	\end{enumerate}
	It is easy to see that combining statements 1-4 above, with statements 1-4 in the Noising step, we get all the desired conditions in \Cref{lem:smoothing_main}.\\
	
	\noindent In this step, we let $A^{(1)} = \overline{A}^{\le d}$ and $B^{(1)} = \overline{B}^{\le d}$ (i.e. degree-$d$ truncation on every coordinate $j \in [k]$). We get the above statements as follows,
	\begin{enumerate}
		\item By definition of degree-$d$ truncation, we have that $A^{(1)}$ and $B^{(1)}$ have degree at most $d$.
		\item $\Var(A^{(1)}_j) ~=~ \sum\limits_{\substack{\bsigma \in \bbZ_q^n \\ |\bsigma| \le d}} \nu^{2|\bsigma|} \cdot \what{A}_j^{(1)}(\bsigma)^2 ~\le~ \Var(\overline{A}_j)$. Similarly, $\Var(B^{(1)}_j) ~\le~ \Var(\overline{B}_j)$.
		\item We have that,
		\begin{align*}
		\norm{2}{\calR(A^{(1)}) - A^{(1)}}
		&~\le~ \norm{2}{\calR(\overline{A}) - \overline{A}} + \norm{2}{\overline{A} - A^{(1)}} && \text{(\Cref{lem:close-to-simplex})}\\
		&~=~ \norm{2}{\calR(\overline{A}) - \overline{A}} + \norm{2}{\overline{A}^{> d}}\\
		&~\le~ \norm{2}{\calR(\overline{A}) - \overline{A}} ~+~ \sqrt{k} \cdot \frac{\delta}{2\sqrt{k}}\\
		&~\le~ \norm{2}{\calR(\overline{A}) - \overline{A}} ~+~ \delta/2\;.
		\end{align*}
		Similarly for $B^{(1)}$.
		\item We have that $\norm{2}{A^{(1)}_i - \overline{A}_i}\le \frac{\delta}{2\sqrt{k}}$ and $\norm{2}{B^{(1)}_j - \overline{B}_j} \le \frac{\delta}{2\sqrt{k}}$. Hence, using \Cref{lem:close-strategies}, we get that for every $i, j \in [k]$
		\[\inabs{\inangle{A^{(1)}_i, B^{(1)}_j}_{\mu^{\otimes n}} - \inangle{\overline{A}_i, \overline{B}_j}_{\mu^{\otimes n}}} ~\le~ \frac{\delta}{\sqrt{k}}\]
	\end{enumerate}
	
	\noindent This completes the proof of \Cref{lem:smoothing_main} for the case of correlated discrete hypercubes. The proof for the case of correlated Gaussians follows in almost the same way. The only change is that we use $U_\nu$ operator instead of $T_\nu$ operator, and use \Cref{lem:gaussian_smoothing} instead of \Cref{lem:mossel_smoothing}.
\end{proof}

%% file: apx_non-multilinear.tex
\section{Proof of Multi-linearization Lemma} \label{apx:non-multilinear}

\noindent In order to prove the lemma, we consider the definition of a {\em multi-linear truncation}.

\begin{defn}[Multilinear truncation]\label{def:ml_truncation}
	Suppose $A \in L^2(\bbR^n, \gamma_n)$ is given by the Hermite expansion $A(x) = \sum\limits_{\bsigma \in \bbZ_{\ge 0}^n} \what{A}_{\bsigma} H_{\bsigma}(x)$. The {\em multilinear truncation} of $A$ is defined as the function $A^{\ml} \in L^2(\bbR^n, \gamma_n)$ given by
	\[A^{\ml}(x) ~:=~ \sum_{\bsigma \in \bit^n} \what{A}_{\bsigma} H_{\bsigma}(x).\]
	That is, $A^{\ml}$ is obtained by retaining only the multilinear terms in the Hermite expansion of $A$.\\
	For convenience, also define $A^{\nml} := A - A^{\ml}$. Also, for vector valued functions $A$, we define $A^{\ml}$ as the function obtained by applying the above multilinear truncation on each coordinate.
\end{defn}

%The proof of \Cref{lem:non-ml-to-ml} appears in \Cref{sec:app_non_mult}.
The proof of \Cref{lem:multilin_main} will proceed by simply applying the transformation given in following lemma to each coordinate of $A$ and $B$. The following lemma shows that low-degree polynomials over $\bbR^n$ can be converted to multilinear polynomials without hurting the correlation. This is done by slightly increasing the number of variables. In addition, we also get that these new polynomials have small individual influences.

\begin{lem}\label{lem:non-ml-to-ml}
	Given parameters $\rho \in [0,1]$, $\delta>0$ and $d \in \bbZ_{\ge 0}$, there exists $t = t(d, \delta)$ such that the following holds:
	
	Let $A, B \in L^2(\bbR^n, \gamma_n)$ be degree-$d$ polynomials, such that $\norm{2}{A}, \norm{2}{B} \le 1$. Define polynomials $\overline{A}, \overline{B} \in L^2(\bbR^{nt}, \gamma_{nt})$ over variables $\overline{\bX} := \setdef{X^{(i)}_j}{(i,j) \in [n] \times [t]}$ and $\overline{\bY} := \setdef{Y^{(i)}_j}{(i,j) \in [n] \times [t]}$ respectively, as,
	\[ \overline{A}\inparen{\overline{\bX}} := A(X^{(1)}, \ldots, X^{(n)}) \qquad \sAND \qquad \overline{B}\inparen{\overline{\bY}} := B(Y^{(1)}, \ldots, Y^{(n)})
	\]
	where $X^{(i)} = \inparen{X^{(i)}_1 + \cdots + X^{(i)}_t}/\sqrt{t}$ and $Y^{(i)} = \inparen{Y^{(i)}_1 + \cdots + Y^{(i)}_t}/\sqrt{t}$. Note that, intuitively this doesn't change the ``structure'' of $A$ and $B$. In particular, it is easy to see that,
	\[\inangle{\overline{A}, \overline{B}}_{\calG_\rho^{\otimes nt}} \ = \ \inangle{A, B}_{\calG_\rho^{\otimes n}} \quad \sAND \quad \norm{2}{\overline{A}} = \norm{2}{A} \quad \sAND \quad \norm{2}{\overline{B}} = \norm{2}{B}\]
	
	\noindent Next, let $\overline{A}^{\ml}, \overline{B}^{\ml} \in L^2(\bbR^{nt}, \gamma_{nt})$ be the multilinear truncations of $\overline{A}$ and $\overline{B}$ respectively. Then the following hold,
	\begin{enumerate}
		\item $\overline{A}^{\ml}$ and $\overline{B}^{\ml}$ are multilinear with degree $d$.
		\item $\Var(\overline{A}^{\ml}) \le \Var(A) \le 1$ and $\Var(\overline{B}^{\ml}) \le \Var(B) \le 1$.
		\item $\norm{2}{\overline{A}^{\ml} - \overline{A}}, \norm{2}{\overline{B}^{\ml} - \overline{B}} ~\le~ \delta/2$.
		\item For all $(i,j) \in [n] \times [t]$, it holds that $\Inf_{X^{(i)}_j}\inparen{\overline{A}^{\mathrm{ml}}} \le \delta$ and $\Inf_{Y^{(i)}_j}\inparen{\overline{B}^{\mathrm{ml}}} \le \delta$.
		\item $\inabs{\inangle{\overline{A}^{\mathrm{ml}}, \overline{B}^{\mathrm{ml}}}_{\calG_\rho^{\otimes nt}} - \inangle{A, B}_{\calG_\rho^{\otimes n}}} ~\le~ \delta$.
	\end{enumerate}
	In particular, one may take $t = O\inparen{\frac{d^2}{\delta^2}}$.
\end{lem}

In order to prove \Cref{lem:non-ml-to-ml}, we will need the following multinomial theorem for Hermite polynomials. It can be proved quite easily using the generating function for Hermite polynomials.

\begin{fact}[Multinomial theorem for Hermite polynomials]\label{fact:multinomial_hermite}
	Let $\beta_1, \ldots, \beta_t \in \bbR$ satisfying $\sum_{i=1}^{t} \beta_i^2 = 1$. Then, for any $d \in \bbN$, it holds that
	\[ H_d\inparen{\beta_1 X_1 + \cdots + \beta_t X_t} \ = \ \sum_{\substack{d_1, \ldots, d_t \in \bbZ_{\ge 0} \\ d_1 + \cdots + d_t = d}} \sqrt{\frac{d!}{d_1! \cdots d_t!}} \cdot \prod_{i=1}^{t} \beta_i^{d_i} H_{d_i}(X_i)\;.\]
\end{fact}

%\noindent We are now ready to \Cref{lem:non-ml-to-ml}.

\begin{proof}[Proof of \Cref{lem:non-ml-to-ml}]
	Before we prove the theorem, we will first understand the effect of the transformation from $X$ to $\overline{X}$ for a single Hermite polynomial. Instantiating $\beta_i$'s in \Cref{fact:multinomial_hermite} with $1/\sqrt{t}$, we get that,
	\[ H_d\inparen{\frac{X_1 + \cdots + X_t}{\sqrt{t}}} \ = \ \sum_{\substack{d_1, \ldots, d_t \in \bbZ_{\ge 0} \\ d_1 + \cdots + d_t = d}} \sqrt{\frac{d!}{d_1! \cdots d_t!}} \cdot \frac{\prod_{i=1}^{t} H_{d_i}(X_i)}{t^{d/2}}\;.\]
	%	We split the right hand side into multilinear and non-multilinear parts. That is,
	%	\[ H_d\inparen{\frac{X_1 + \cdots + X_t}{\sqrt{t}}} \ = \ \frac{1}{t^{d/2}} \cdot \inparen{\sum_{\substack{d_1, \ldots, d_t \in \bit \\ d_1 + \cdots + d_t = d}} \sqrt{d!} \cdot \prod_{i=1}^{t} H_{d_i}(X_i) + \sum_{\substack{d_1, \ldots, d_t \in \bbZ_{\ge 0} \\ d_1 + \cdots + d_t = d \\ \exists i\,\ d_i \ge 2}} \sqrt{\frac{d!}{d_1! \cdots d_t!}} \cdot \prod_{i=1}^{t} H_{d_i}(X_i)}\;.\]
	We will split the terms into multilinear and non-multilinear terms, writing the above as $H_d^{\ml} + H_d^{\nml}$. Note that there are at most $O(\frac{d^2 t^{d-1}}{d!})$ non-multilinear terms (for $t \gg d^2$). Also, note that each coefficient $\frac{1}{t^{d/2}} \cdot \sqrt{\frac{d!}{d_1! \cdots d_t!}}$ is at most $\sqrt{\frac{d!}{t^d}}$. Thus, we can bound $\norm{2}{H_d^{\nml}}$ as follows,
	\begin{equation}
	\norm{2}{H^{\nml}}^2 \quad = \quad \sum_{\substack{d_1, \ldots, d_t \in \bbZ_{\ge 0} \\ d_1 + \cdots + d_t = d \\ \exists i\,\ d_i \ge 2}} \inparen{\frac{1}{t^{d/2}} \cdot \sqrt{\frac{d!}{d_1! \cdots d_t!}}}^2 \quad \le \quad O\inparen{\frac{d^2 t^{d-1}}{d!}} \cdot \frac{d!}{t^d} \quad \le \quad O\inparen{\frac{d^2}{t}} \label{eqn:nml_small}
	\end{equation}
	
	\noindent More generally, if we consider a term $\overline{H}_\bsigma\inparen{\overline{\bX}} = H_{\sigma_1}(X^{(1)}) \cdot H_{\sigma_2}(X^{(2)}) \cdots H_{\sigma_N}(X^{(N)})$, where each $X^{(i)} = \inparen{X^{(i)}_1 + \cdots + X^{(i)}_t}/\sqrt{t}$. Let's write $\overline{H}_\bsigma\inparen{\overline{\bX}} = \overline{H}_\bsigma^{\ml}\inparen{\overline{\bX}} + \overline{H}_\bsigma^{\nml}\inparen{\overline{\bX}}$, that is, separating out the multilinear and non-multilinear terms. Similarly, for any $i$, let $H_{\sigma_i}(X^{(i)}) = H_{\sigma_i}^{\ml}(X^{(i)}) + H_{\sigma_i}^{\nml}(X^{(i)})$. %Observe that $\overline{H}_{\bsigma}^{\nml}(\overline{X}) = \prod\limits_{i=1}^{n}(H_{\sigma_i}^{\ml} + H_{\sigma_i}^{\nml}(X^{(i)})) - \prod\limits_{i=1}^n H_{\sigma_i}^{\ml}(X^{(i)})$.
	We wish to bound $\norm{2}{\overline{H}_{\bsigma}^{\nml}}$, which can be done as follows,
	\begin{align}
	\norm{2}{\overline{H}_{\bsigma}^{\nml}}^2 &~=~ \norm{2}{\prod\limits_{i=1}^{n}(H_{\sigma_i}^{\ml} + H_{\sigma_i}^{\nml}) - \prod\limits_{i=1}^n H_{\sigma_i}^{\ml}}^2 \nonumber\\
	&~\le~ \prod_{i=1}^n \inparen{1+O\inparen{\frac{\sigma_i^2}{t}}} - 1 && \text{(from \Cref{eqn:nml_small})} \nonumber\\
	&~\le~ O\inparen{\frac{|\bsigma|^2}{t}} && \text{(since, $t \gg |\bsigma|^2$)} \nonumber\\
	\text{Thus, }\quad \norm{2}{\overline{H}_{\bsigma}^{\nml}}^2 &~<~ \delta^2/4. && (\text{for } t = \Theta(d^2/\delta^2))\label{eqn:global_nml_small}
	\end{align}
	
	\noindent We are now ready to prove the parts of our \Cref{lem:non-ml-to-ml}.
	
	\begin{enumerate}
		\item It holds by definition that $\overline{A}^{\ml}$ and $\overline{B}^{\ml}$ are multilinear. Also, note that the transformation from $A$ to $\overline{A}$ and finally to $\overline{A}^{\ml}$ does not increase the degree. So both  $\overline{A}^{\ml}$ and $\overline{B}^{\ml}$ have degree at most $d$.
		
		\item It is easy to see that $\Var(\overline{A}) = \Var(A)$. Since $\Var(\overline{A}^{\ml})$ is obtained by truncating certain Hermite coefficients of $\overline{A}$, it immediately follows that $\Var(\overline{A}^{\ml}) \le \Var(\overline{A}) = \Var(A) \le 1$. Similarly, $\Var(\overline{B}^{\ml}) \le \Var(B) \le 1$.
		
		\item Recall that $\overline{A}^{\nml} = \overline{A} - \overline{A}^{\mathrm{ml}}$. We wish to bound $\norm{2}{\overline{A}^{\nml}}^2 \le \delta^2/4$. Consider the Hermite expansion of $A$, namely $A(\bX) = \sum_{\bsigma \in \bbZ_{\ge 0}^N} \what{A}(\bsigma) \cdot H_{\bsigma}(\bX)$. Note that, $\overline{A}^{\nml}\inparen{\overline{\bX}} = \sum_{\bsigma \in \bbZ_{\ge 0}^N} \what{A}(\bsigma) \cdot \overline{H}^{\nml}_{\bsigma}\inparen{\overline{\bX}}$, where recall that $\overline{H}_{\bsigma}^{\nml}$ is the non-multilinear part of $\overline{H}_{\bsigma}\inparen{\overline{\bX}} = H_{\sigma_1}(X^{(1)}) \cdot H_{\sigma_2}(X^{(2)}) \cdots H_{\sigma_n}(X^{(N)})$, where each $X^{(i)} = \inparen{X^{(i)}_1 + \cdots + X^{(i)}_t}/\sqrt{t}$.
		
		From Equation~\ref{eqn:global_nml_small}, we have that for any $\bsigma \in \bbZ_{\ge 0}^N$, it holds that $\norm{2}{\overline{H}_{\bsigma}^{\nml}\inparen{\overline{X}}}^2 < \delta^2/4$. And hence we get that,
		\[ \norm{2}{\overline{A}^{\nml}}^2 ~=~ \sum_{\bsigma} \what{A}(\bsigma)^2 \cdot \norm{2}{\overline{H}_{\bsigma}^{\nml}}^2 ~\le~ \sum_{\bsigma} \what{A}(\bsigma)^2 \cdot (\delta^2/4) ~=~ (\delta^2/4) \norm{2}{A}^2 ~\le~ (\delta^2/4). \]
		Note that, here we use that $\overline{H}_{\bsigma}(\overline{\bX})$ are mutually orthogonal for different $\bsigma$. Similarly, we can also get that $\norm{2}{\overline{B}^{\nml}}^2 \le \delta^2/4$.
		
		\item We prove that $\Inf_{X^{(i)}_j}\inparen{\overline{A}^{\ml}} ~\le~ d/t ~<~ \delta^2/d ~<~ \delta$. The case $\Inf_{Y^{(i)}_j}\inparen{\overline{B}^{\ml}}$ will follow similarly.
		
		For simplicity, let's first consider the case of a univariate polynomial $P \in L^2(\bbR, \gamma_1)$ of degree-$d$, such that $\norm{2}{P} \le 1$. We will show that for the function $\overline{P}(X_1, \cdots, X_t) := P \inparen{\frac{X_1 + \cdots + X_t}{\sqrt{t}}}$, it holds that $\Inf_{X_i}\inparen{\overline{P}^{\ml}} \le \Var(P)\cdot (d/t)$. This follows from a few simple observations:
		\begin{enumerate}
			\item By symmetry, $\Inf_{X_i}\inparen{\overline{P}^{\ml}}$ is the same for all $i$.
			\item For degree-$d$ multilinear polynomials, $\sum_{i \in [t]} \Inf_{X_i}\inparen{\overline{P}^{\ml}} \le d \Var\inparen{\overline{P}^{\ml}}$. %\Red{(see Fact~\ref{??})}.
			\item $\Var\inparen{\overline{P}^{\ml}} \le \Var\inparen{\overline{P}} = \Var(P)$
		\end{enumerate}
		Thus, combining the above, we get that $\Inf_{X_i}\inparen{\overline{P}^{\ml}} \le \Var(P)\cdot (d/t)$.\\
		
		In the more general case $n$-variate case, we observe that,
		\begin{align*}
		\Inf_{X^{(i)}_j}\inparen{\overline{A}^{\ml}} &~=~ \Ex_{\overline{\bX} \setminus \set{X^{(i)}_j}} \ \Var_{X^{(i)}_j} \inparen{\overline{A}^{\ml}_{\overline{X} \setminus \set{X^{(i)}_j}} \inparen{X^{(i)}_j}}\\
		&~=~ \Ex_{\bX^{(-i)}} \Ex_{\bX^{(i)}_{-j}} \Var_{X^{(i)}_j} \inparen{\overline{A}^{\ml}_{\overline{X} \setminus \set{X^{(i)}_j}} \inparen{X^{(i)}_j}} && \text{(for simplicity, $\bX^{(-i)} := \overline{\bX} \setminus \setdef{X^{(i)}_j}{j \in [t]}$)}\\
		&~=~ \Ex_{\bX^{(-i)}} \Inf_{X^{(i)}_j} \inparen{\overline{A}^{\ml}_{X^{(-i)}}(X^{(i)})}\\
		&~=~ \Ex_{\bX^{(-i)}} \Var\inparen{\overline{A}^{\ml}_{X^{(-i)}}} \cdot (d/t)\\
		&~\le~ (d/t) \cdot \Inf_i(A)\\
		&~\le~ d/t\\
		&~<~ \delta.
		\end{align*}
		where, in last two inequalities, we use that $\Inf_i(A) \le \Var(A) \le 1$ and that $t = \Theta(d^2/\delta^2)$.
		
		\item Note that we already have,
		\[ \inangle{\overline{A}, \overline{B}}_{\calG_\rho^{\otimes nt}} \ = \ \inangle{A, B}_{\calG_\rho^{\otimes n}}\;. \]
		And combining Part 3 and \Cref{lem:close-strategies}, we immediately get that
		\[\inabs{\inangle{\overline{A}^{\mathrm{ml}}, \overline{B}^{\mathrm{ml}}}_{\calG_\rho^{\otimes nt}} - \inangle{\overline{A}, \overline{B}}_{\calG_\rho^{\otimes nt}}} ~\le~ \delta\] % as follows,
		where we use that $\norm{2}{\overline{B}^{\ml}} \le \norm{2}{\overline{B}} \le 1$ and $\norm{2}{\overline{A}^{\ml}} \le \norm{2}{\overline{A}} \le 1$.
	\end{enumerate}
\end{proof}

\noindent We are now able to prove \Cref{lem:multilin_main}.

\begin{proof}[Proof of \Cref{lem:multilin_main}]
	We apply the transformation in \Cref{lem:non-ml-to-ml}, with parameter $\delta$ being $\delta/\sqrt{k}$, to each of the $k$-coordinates of $A : \bbR^n \to \bbR^k$ and $B : \bbR^n \to \bbR^k$ to get function $A^{(1)} : \bbR^{nt} \to \bbR^k$ and $B^{(1)} : \bbR^n \to \bbR^k$. Namely, for any $j \in [k]$, we set $A^{(1)}_j(\overline{\bX}) = \overline{A}^{\ml}_j(\overline{\bX})$ and $B^{(1)}_j(\overline{\bY}) = \overline{B}^{\ml}_j(\overline{\bY})$as described in \Cref{lem:non-ml-to-ml}.
	
	It is easy to see that parts 1, 2, 4, 5 follow immediately from the conditions satisfied in \Cref{lem:non-ml-to-ml}. For part 3, we note that we have that $\norm{2}{\overline{A}^{\ml}_j - \overline{A}_j} \le \delta/\sqrt{k}$ for every $j \in [k]$. Thus, combining these for all $j \in [k]$, we get that $\norm{2}{\overline{A}^{\ml} - \overline{A}} \le \delta$. Now, using \Cref{lem:close-to-simplex}, we immediately get that,
	\[\norm{2}{\calR(\overline{A}^{\ml}) - \overline{A}^{\ml}} \le \norm{2}{\calR(\overline{A}) - \overline{A}} + \delta\;.\]
	Finally, we observe that $\norm{2}{\calR(\overline{A}) - \overline{A}} = \norm{2}{\calR(A) - A}$, to conclude that 
	\[\norm{2}{\calR(A^{(1)}) - A^{(1)}} \le \norm{2}{\calR(A) - A} + \delta\;.\]
	Similarly, $\norm{2}{\calR(B^{(1)}) - B^{(1)}} \le \norm{2}{\calR(B) - B} + \delta$. This concludes the proof.
	
\end{proof}

%% file: sec_regularity.tex
\section{Regularity Lemma for low-degree functions} \label{sec:regularity}
In this section we state and prove a regularity lemma that we need for proving \Cref{th:non-int-sim}, i.e. non-interactive simulation from discrete sources. Our regularity lemma follows immediately from the version stated in \cite{GKS_NIS_decidable}, which was inspired from \cite{diakonikolas2010regularity}.

We begin by first recalling the basic notions of influences and partial restrictions of functions over product spaces. %These definitions are relevant only for the section on non-interactive simulation from discrete sources (i.e. \Cref{sec:non-int-sim}).

\begin{definition}[Influence]
	For every coordinate $i \in [n]$, $\Inf_i(f)$ is the $i$-th influence of $f$, and $\Inf(f)$ is the total influence, which are defined as
	$$\Inf_i(f) \defeq \Ex_{\bx_{-i}} \insquare{\Var_{x_i} \, [f(\bx)]} \quad \quad \quad \Inf(f) \defeq \sum\limits_{i=1}^n \Inf_i(f)$$
\end{definition}

\noindent The basic properties of influence are summarized in the following fact.

\begin{fact}\label{fact:influences}
	For any function $f \in L^2(\calZ^n, \mu_A^{\otimes n})$, we have the following:
	\begin{enumerate}
		\item[(i)] $\Inf_i(f) = \sum\limits_{\bsigma : \sigma_i \ne \mathbf{0}} \what{f}(\bsigma)^2$ and hence, for all $i$, $\Inf_i(f) \le \Var(f)$
		\item[(ii)] $\Inf(f) = \sum\limits_{\bsigma} |\bsigma| \cdot \what{f}(\bsigma)^2$
		\item[(iii)] If $\deg(f) = d$, then $\Inf(f) \le d \cdot \Var[f]$.
	\end{enumerate}
\end{fact}

\noindent In the regularity lemma, we deal with restrictions of polynomials. For any subset $H \subseteq [n]$, we will use $\bx_H$ to denote the tuple of variables in $\bx$ with indices in $H$. For any function $P \in L^2(\calZ^n, \mu^{\otimes n})$, and any $\xi \in \calZ^{H}$, we will use $P^{\xi}$ to denote the function obtained by restriction of $\bx_H$ to $\xi$, that is, $P^{\xi}(\bx_{T}) = P(\bx_H \gets \xi, \bx_T)$ (where $T = [n] \setminus H$); whenever we use such terminology, the subset $H$ will be clear from context.	We will use the phrase ``$\xi$ fixes $H$ over $\calA$'' to mean such a restriction. We will use $\set{\bsigma_H}$ to denote all degree sequences in $\bbZ_q^H$, and similarly $\set{\bsigma_T}$ to denote all degree sequences in $\bbZ_q^T$. We use $\bsigma_H \circ \bsigma_T$ to denote $\bsigma \in \bbZ_q^n$ such that $\sigma_i = (\sigma_H)_i \text{ if } i \in H \text{ or } (\sigma_T)_i \text{ if } i \in T$.\\

\noindent We now state our main Joint Regularity Lemma.

\begin{lem}[Joint Regularity Lemma]\label{lem:joint_regularity}
	Let $(\calZ \times \calZ, \mu)$ be a joint probability space. Let $k$, $d \in \bbN$ and $\tau > 0$ be any given constant parameters. There exists $h \defeq h((\calZ \times \calZ, \mu), k, d, \tau)$ such that the following holds:
	
	For all degree-$d$ polynomials $P : \calZ^n \to \bbR^k$ and $Q : \calZ^n \to \bbR^k$ such that, $\Var[P_{j}] \le 1$ and $\Var[Q_{j}] \le 1$ for all $j \in [k]$, there exists a subset of indices $H \subseteq [n]$ with $|H| \le h$ such that with probability at least $(1-\tau)$ over the assignment $(\xi_A, \xi_B) \sim \mu^{\otimes h}$, the following holds for any $j \in [k]$ (where we denote $T = [n] \setminus H$),
	\begin{itemize}
		\item the restriction $P_{j}^{\xi_A}(x_T)$ is such that for all $i \in T$, it holds that $\Inf_i(P_{j}^{\xi_A}(x_T)) \le \tau$,
		\item the restriction $Q_{j}^{\xi_B}(y_T)$ is such that for all $i \in T$, it holds that $\Inf_i(Q_{j}^{\xi_B}(y_T)) \le \tau$.
	\end{itemize}
	\noindent In particular, one may take $h = \frac{dk^2}{\tau} \cdot \inparen{\frac{C_4(\alpha)}{\alpha} \log \frac{C_4(\alpha) \cdot k}{\alpha \cdot d \cdot \tau}}^{O(d)}$ which is a constant that depends on $k$, $d$, $\tau$ and $\alpha \defeq \alpha(\mu)$, which is the minimum non-zero probability in $\mu$.
\end{lem}

The proof of the Joint Regularity lemma follows quite easily by applying the Regularity Lemma for degree-$d$ polynomials (cf. Lemma 5.2 in \cite{GKS_NIS_decidable}).

\begin{lem}[Regularity Lemma for degree-$d$ functions]\label{lem:single_reg_lem}
	Let $(\calZ, \mu_A)$ be a probability space. Let $d \in \bbN$ and $\tau > 0$ be any given constant parameters. There exists $h \defeq h((\calZ, \mu_A), d, \tau)$ such that the following holds:
	
	For all degree-$d$ multilinear polynomials $P \in L^2(\calZ^n, \mu_{A}^{\otimes n})$ with $\Var[P] \le 1$, there exists a subset of indices $H_0 \subseteq [n]$ with $|H_0| \le h$, such that for any superset $H \supseteq H_0$, the restrictions of $P$ obtained by evaluating the coordinates in $H$ according to distribution $\mu_A$, satisfies the following (where we denote $T = [n] \setminus H$):
	$$\Pr\limits_{\xi \sim \mu_A^{\otimes |H|}} \insquare{\forall i \in T : \Inf_i(P^{\xi}(x_T)) \le \tau} \ge 1-\tau$$
	In other words, with probability at least $1-\tau$ over the random restriction $\xi \sim \mu_{A}^{\otimes |H|}$, the restricted function $P^{\xi}(x_T)$ is such that $\Inf_i(P^{\xi}(x_T)) \le \tau$ for all $i \in T$.
	
	\noindent In particular, one may take $h = \frac{d}{\tau} \cdot \inparen{\frac{C_4(\alpha)}{\alpha} \log \frac{C_4(\alpha)}{\alpha \cdot d \cdot \tau}}^{O(d)}$ which is a constant that depends on $d$, $\tau$ and $\alpha \defeq \alpha(\mu_A)$.
\end{lem}

\noindent \Cref{lem:joint_regularity} follows quite easily from the above lemma.

\begin{proof}[Proof of \Cref{lem:joint_regularity}]
	We invoke the Regularity Lemma for individual degree-$d$ polynomials \Cref{lem:single_reg_lem}, namely $P_{j}$'s and $Q_{j}$'s and then applying a union bound.
	
	In particular, given our parameters $d$ and $\tau$, we invoke \Cref{lem:single_reg_lem} with parameters $d$ and $\tau/2k$ for each $P_{j}$ and $Q_{j}$. Suppose we get the set $H_A^{(j)}$ (resp. $H_B^{(j)}$) when applying the regularity lemma on $P_{j}$ (resp. $Q_{j}$). We let $H = \bigcup_{j \in [k]} \inparen{H_A^{(j)} \cup H_B^{(j)}}$. Note that, $|H| \le 2k \cdot \frac{2dk}{\tau} \cdot \inparen{\frac{C_4(\alpha)}{\alpha} \log \frac{C_4(\alpha) \cdot k}{\alpha \cdot d \cdot \tau}}^{O(d)}$.
	
	\Cref{lem:single_reg_lem} gives us that for this $H$, it holds for any $P_{j}$ and $Q_{j}$ that,
	\[ \Pr\limits_{\xi_A \sim \mu_A^{\otimes |H|}} \insquare{\forall i \in T : \Inf_i\inparen{P_{j}^{\xi_A}(x_T)} \le \tau/2k} \ge 1-(\tau/2k) \]
	\[ \Pr\limits_{\xi_B \sim \mu_B^{\otimes |H|}} \insquare{\forall i \in T : \Inf_i\inparen{Q_{j}^{\xi_B}(y_T)} \le \tau/2k} \ge 1-(\tau/2k) \]
	Taking a union bound over $2k$ such statements in total, we get the conclusion we desire.
	
\end{proof}

%% file: sec_invariance.tex
\section{Invariance Principle} \label{sec:invariance}

In this section, we prove an invariance principle statement tailor-made for our application, that follows readily as a special case of known invariance principle statements \cite{mossel2005noise, mossel2010gaussianbounds, isaksson2012maximally}. In particular, we desire a statement as follows.

\begin{lem}\label{lem:our_invariance}
	Let $(\calZ \times \calZ, \mu)$ be a finite joint probability space, such that $|\calZ| = q$ and $\alpha := \alpha(\mu) > 0$ is the minimum probability of any atom in $\mu$. Given parameters $k, d \in \bbN$ and $\delta > 0$, there exists $\tau = \tau((\calZ \times \calZ, \mu), k, d, \delta)$ such that the following holds:
	
	Let $A : \calZ^n \to \bbR^k$ and $B : \calZ^n \to \bbR^k$ be degree-$d$ multilinear polynomials, such that, $\Var(A_j), \Var(B_j) \le 1$ and $\Inf_\ell(A_j), \Inf_\ell(B_j) \le \tau$ for all $\ell \in [n]$ and $j \in [k]$. Then, there exist degree-$d$ multilinear polynomials $\wtilde{A} : \bbR^{n \cdot (q-1)} \to \bbR^k$ and $\wtilde{B} : \bbR^{n \cdot (q-1)} \to \bbR^k$, such that, for all $i, j \in [k]$, it holds that,
	\begin{equation}\label{eqn:our_invariance}
		\inabs{\inangle{\calR_i(\wtilde{A}), \calR_j(\wtilde{B})}_{\calG_\rho^{\otimes n (q-1)}} ~-~ \inangle{\calR_i(A), \calR_j(B)}_{\mu^{\otimes n}}} \le \delta\;,
	\end{equation}
	where $\rho = \rho(\calZ, \calZ; \mu)$ is the maximal correlation of $\mu$. In particular, one may take $\tau =  O\inparen{\frac{\delta^{1.5} \cdot \alpha^{d/2}}{2^{O(d)} \cdot 2^{O(k)}}}$.
	
	Additionally, the theorem also works in reverse, namely, given degree-$d$ multilinear polynomials $\wtilde{A} : \bbR^{n} \to \bbR^k$ and $\wtilde{B} : \bbR^{n} \to \bbR^k$, such that $\Var(\wtilde{A}_j), \Var(\wtilde{B}_j) \le 1$ and  $\Inf_\ell(\wtilde{A}_j), \Inf_\ell(\wtilde{B}_j) \le \tau$ for all $\ell \in [n]$ and $j \in [k]$, there exist $A : \calZ^n \to \bbR^k$ and $B : \calZ^n \to \bbR^k$, such that \Cref{eqn:our_invariance} holds.
\end{lem}

\noindent The proof is pretty standard, nevertheless we provide a proof for completeness. We use the vector-valued invariance principle from \cite{isaksson2012maximally}, which builds on \cite{mossel2005noise, mossel2010gaussianbounds}. We state a version that is more tailored to our application, modified from \cite[Theorem 3.4]{isaksson2012maximally}.

\begin{lem}[Invariance Principle (cf. \cite{isaksson2012maximally})] \label{lem:invariance}
	Let $(\Omega^n, \mu^{\otimes n})$ be a finite probability space, such that $\alpha > 0$ is the minimum probability of any atom in $\mu$. Let $\calF = (\calF_1, \ldots, \calF_n)$ be an independent sequence of orthonormal ensembles such that $\calF_\ell$  is a basis for functions $\Omega \to \bbR$, for any $\ell \in [n]$. Let $P$ be a $K$-dimensional multilinear polynomial such that for every $j \in [K]$, it holds that $\Var(P_j) \le 1$, $\deg(P_j) \le d$ and $\Inf_i (P_j) \le \tau$ for every $i \in [n]$. Finally, let $\Psi: \bbR^K \to \bbR$ be Lipschitz continuous with Lipschitz constant $L$. Then,
	\[ \inabs{\Ex_{\calF} \Psi(P(\calF)) - \Ex_{\calG} \Psi(P(\calG))} ~\le~ D_K \cdot L \cdot \inparen{d (8/\sqrt{\alpha})^d \sqrt{\tau}}^{1/3}\;,\]
	where $\calG$ is an independent sequence of Gaussian ensembles with same covariance structure as $\calF$ and $D_K = 2^{O(K)}$.\footnote{Note that the lemma stated in \cite{isaksson2012maximally} does not state the explicit bound of $D_K = 2^{O(K)}$. However, it is possible to infer this bound from their proof.}
	
%	Let $\calZ = (\calZ_1, \cdots, \calZ_n)$ be an independent sequence of ensembles, such that $\Pr(\calZ_i = x) \ge \alpha > 0$. Let $Q$ be a $k$-dimensional multilinear polynomial such that $\Var(Q_j) \le 1$, $\deg(Q_j) \le d$ and $\Inf_i Q_j \le \tau$. Finally, let $\Psi: \bbR^k \to \bbR$ be a Lipschitz continuous with Lipschitz constant $L$. Then,
%	\[ \inabs{\Ex_{\calX} \Psi(Q(\calX)) - \Ex_{\calG} \Psi(Q(\calG))} ~\le~ D_k \cdot L \cdot \inparen{d (8/\sqrt{\alpha})^d \sqrt{\tau}}^{1/3}\;,\]
%	where $\calG$ is an independent sequence of Gaussian ensembles with the same covariance structure as $\calX$ and $D_k$ is a universal constant.
\end{lem}

\begin{proof}[Proof of \Cref{lem:our_invariance}]
We will apply \Cref{lem:invariance} on the space $(\Omega^n, \mu^{\otimes n}) = (\calZ^n \times \calZ^n, \mu^{\otimes n})$. We consider the independent sequence of orthonormal ensembles $\calF$ given by $\calF_\ell = \setdef{\calX^{(\ell)}_i(x) \calY^{(\ell)}_j(y)}{i, j \in \set{0, \ldots, q-1}}$ for any $\ell \in [n]$ (where, recall that $|\calZ| = q$). Although as defined $\calF_\ell$ has $q^2$ elements, the polynomials we consider will only depend on the subset of characters $\setdef{\calX^{(\ell)}_1, \ldots, \calX^{(\ell)}_{q-1}, \calY^{(\ell)}_1, \ldots, \calY^{(\ell)}_{q-1}}{\ell \in [n]}$. Also, observe that we can choose the characters $\set{\calX^{(\ell)}_i}$ and $\set{\calY^{(\ell)}_i}$ such that $\inangle{\calX^{(\ell)}_i, \calY^{(\ell)}_j}_\mu = \rho_i \cdot \mathbf{1}_{\set{i=j}}$, where $\rho_1 = \rho(\calZ, \calZ; \mu)$ is the maximal correlation, and $\rho_i \le \rho$ for all $1 \le i \le q-1$.

Given degree-$d$ multilinear polynomials $A : \calZ^n \to \bbR^k$ and $B : \calZ^n \to \bbR^k$, we consider the polynomial $P : \calZ \times \calZ \to \bbR^{2k}$ given by $P = (P_1, \ldots, P_{2k})$, where for $j \le k$, we take $P_j(\bx, \by) = A_j(\bx)$ and for $j > k$, we take $P_j(\bx, \by) = B_{j-k}(\by)$. In other words, $P$ is a {\em concatenation} of $A$ and $B$. It is clear that $P$ is also a degree-$d$ multilinear polynomial in $\calF$. Also, since $A$ and $B$ are such that $\Var(A_j), \Var(B_j) \le 1$ and $\Inf_\ell(A_j), \Inf_\ell(B_j) \le \tau$ for all $\ell \in [n]$ and $j \in [k]$, we have that for all $j \in [2k]$, it holds that, $\Var(P_j) \le 1$ and $\Inf_\ell(P_j) \le \tau$. Thus, $P$ satisfies all the conditions needed to prove \Cref{lem:invariance}.

We will use the test function $\Psi : \bbR^{2k} \to \bbR$, given by $\Psi(\bu,\bv) := \calR_i(\bu) \cdot \calR_j(\bv)$, for any given $i, j \in [k]$. It is easy to show that $\calR(\cdot)$ is a contraction map (since it is rounding to $\Delta_k$ which is a convex body), and hence any coordinate $\calR_i(\cdot)$ has Lipschitz constant of at most $1$. This gives us that our test function $\Psi(\bu, \bv) = \calR_i(\bu) \cdot \calR_j(\bv)$ also has a Lipschitz constant of at most $1$.

We can interpret the invariance principle, as substituting $\setdef{\calX^{(\ell)}_1, \ldots, \calX^{(\ell)}_{q-1}, \calY^{(\ell)}_1, \ldots, \calY^{(\ell)}_{q-1}}{\ell \in [n]}$ by correlated multivariate Gaussians $({\bf g}^{(\ell)}, {\bf h}^{(\ell)}) = (g^{(\ell)}_1, \cdots, g^{(\ell)}_{q-1}, h^{(\ell)}_1, \ldots, h^{(\ell)}_{q-1})$, such that $\Ex g^{(\ell)}_i h^{(\ell)}_j = \rho_i \cdot \mathbf{1}_{\set{i=j}}$ and $\Ex g^{(\ell)}_i g^{(\ell)}_j = \Ex h^{(\ell)}_i h^{(\ell)}_j = \mathbf{1}_{\set{i=j}}$. Thus the invariance principle is taking $P : \calZ^n \times \calZ^n \to \bbR^{2k}$ and producing $P' : \bbR^{(q-1)n} \times \bbR^{(q-1)n} \to \bbR^{2k}$. Note that the first $k$ coordinates of $P$ are polynomials over the subset $\setdef{\calX^{(\ell)}_1, \ldots, \calX^{(\ell)}_{q-1}}{\ell \in [n]}$ and the latter $k$ coordinates are polynomials over $\setdef{\calY^{(\ell)}_1, \ldots, \calY^{(\ell)}_{q-1}}{\ell \in [n]}$. Hence we can interpret the first $k$ coordinates as $A' : \bbR^{(q-1)n} \to \bbR^k$ and latter $k$ coordinates as $B' : \bbR^{(q-1)n} \to \bbR^k$, and \Cref{lem:invariance} gives us that,
\[ \inabs{\inangle{\calR_i(A'), \calR_j(B')}_{\calG^{\otimes n}} - \inangle{\calR_i(A), \calR_j(B)}_{\mu^{\otimes n}}} ~\le~ 2^{O(k)} \cdot \inparen{d (8/\sqrt{\alpha})^d \sqrt{\tau}}^{1/3} ~\le~ \delta \]
for a choice of $\tau \le O\inparen{\frac{\delta^{1.5} \cdot \alpha^{d/2}}{2^{O(d)} \cdot 2^{O(k)}}}$.

We are still not done though! We want $\wtilde{A}$ and $\wtilde{B}$ which act on coorrelated inputs sampled from $\calG_\rho^{\otimes (q-1)n}$, that is, $({\bf g}^{(\ell)}, {\bf h}^{(\ell)}) = (g^{(\ell)}_1, \cdots, g^{(\ell)}_{q-1}, h^{(\ell)}_1, \ldots, h^{(\ell)}_{q-1})$, such that $\Ex g^{(\ell)}_i h^{(\ell)}_j = \rho \cdot \mathbf{1}_{\set{i=j}}$ and $\Ex g^{(\ell)}_i g^{(\ell)}_j = \Ex h^{(\ell)}_i h^{(\ell)}_j = \mathbf{1}_{\set{i=j}}$. However the correlation pattern obtained in $\calG$ is not exactly this. But note that each $\rho_i \le \rho$. Thus, given $(g^{(\ell)}_i, h^{(\ell)}_i)$ with correlation $\rho$, we can simply apply $U_{\rho_i/\rho}$ operator on $h^{(\ell)}_i$ to bring down the correlation from $\rho$ to $\rho_i$. Applying this appropriate operation for every $\ell \in [n]$ and $i \in [q]$, we get our desired $\wtilde{A} : \bbR^{n \cdot (q-1)} \to \bbR^k$ and $\wtilde{B} : \bbR^{n \cdot (q-1)} \to \bbR^k$.

The reverse part of the theorem follows similarly as well. Here, we get polynomials $A$ and $B$, that depend only on $\set{\calX^{(\ell)}_1}$ and $\set{\calY^{(\ell)}_1}$ respectively.
\end{proof}

%\begin{proof}
%	The proof uses \Cref{lem:invariance} and instantiates $\Psi(u_1, \ldots, u_k, v_1, \ldots, v_k)$ differently to prove the three parts.
%	\begin{enumerate}
%		\item For $i, j \in [k]$ : choose $\Psi_1(\bu, \bv) = u_i v_j$
%		\item $\Psi_2(\bu, \bv) = \norm{2}{\calR(\bu) - \bu}^2$
%		\item $\Psi_3(\bu, \bv) = \norm{2}{\calR(\bv) - \bv}^2$.
%	\end{enumerate}
%\end{proof}

%% file: sec_decidability.tex
\section{Decidability of Non-Interactive Simulation} \label{sec:decidability}

In this section, we prove \Cref{thm:decidability} showing the decidability of the $\ANIS$ problem.

\begin{proof}[Proof of \Cref{thm:decidability}]
	If we were in the YES case of $\ANIS((\calZ \times \calZ, \mu), V, k, \eps)$, then we have that there exists an $N$ and functions $A : \calZ^N \to \Delta_k$ and $B : \calZ^N \to \Delta_k$, such that the distribution $\nu' = (A(\bx), B(\by))_{(\bx, \by) \sim \mu^{\otimes N}}$ is such that $\dTV(\nu', \nu) \le \eps$ for some $\nu \in V$. Using \Cref{thm:non-int-sim}, with parameter $\eps/3$, we get that there exists functions $\wtilde{A} : \calZ^{n_0} \to \Delta_k$ and $\wtilde{B} : \calZ^{n_0} \to \Delta_k$ such that the distribution $\nu'' = (\wtilde{A}(\bx), \wtilde{B}(\by))_{(\bx, \by) \sim \mu^{\otimes n_0}}$ is such that $\dTV(\nu'', \nu') \le \eps/3$. Hence, $\dTV(\nu'', \nu) \le 4\eps/3$ for some $\nu \in V$.
	
	In the NO case of $\ANIS((\calZ \times \calZ, \mu), V, k, \eps)$, we have that for all $N$, in particular for $N = n_0$, and for all functions $A : \calZ^{n_0} \to \Delta_k$ and $B : \calZ^{n_0} \to \Delta_k$ it holds that the distribution $\nu'' = (A(\bx), B(\by))_{(\bx, \by) \sim \mu^{\otimes n_0}}$ satisfies $\dTV(\nu'', \nu) > 2\eps$ for all $\nu \in V$.
	
	This naturally gives us a brute force algorithm: Analyze all possible functions $\wtilde{A} : \calZ^{n_0} \to \Delta_k$ and $\wtilde{B} : \calZ^{n_0} \to \Delta_k$ to check if there exist functions $\wtilde{A}$ and $\wtilde{B}$ with distribution $\nu'' = (\wtilde{A}(\bx), \wtilde{B}(\by))_{(\bx, \by) \sim \mu^{\otimes n_0}}$ satisfying $\dTV(\nu'', \nu) \le 4\eps/3$ for some $\nu \in V$.
	
	For purposes of our algorithm we can replace the range $\Delta_k$ by any $(\eps/3)$-cover $C$, that is, a set of discrete points in $\Delta_k$ such that any point in $\Delta_k$ is within an $\ell_1$ distance of $\eps/3$ from some point in $C$. Note that we could choose such a $C$ of size at most $(1/\eps)^{\wtilde{O}(k)}$. This ensures that if indeed such a desired $\wtilde{A}$ and $\wtilde{B}$ exist, then we will find functions $\wtilde{A}' : \calZ^{n_0} \to C$ and $\wtilde{B}' : \calZ^{n_0} \to C$ such that distribution $\nu''' = (\wtilde{A}'(\bx), \wtilde{B}'(\by))_{(\bx, \by) \sim \mu^{\otimes n_0}}$ satisfying $\dTV(\nu''', \nu) \le 5\eps/3 < 2\eps$ for some $\nu \in V$. In the YES case, we will find such functions, whereas in the NO case, $\wtilde{A}'$ and $\wtilde{B}'$ as above simply don't exist.
	
	The number of pair of functions $(\wtilde{A}, \wtilde{B})$ to brute force over is $|C|^{O(|\calZ|^{n_0})}$, which gives us an upper bound on the running time as
	\[\exp\exp\exp\inparen{\poly\inparen{k, \ \frac{1}{\eps}, \ \frac{1}{1-\rho_0}, \ \log\inparen{\frac{1}{\alpha}}}}\;.\]
\end{proof}